\documentclass[letterpaper]{article} 
\usepackage{aaai25}  
\usepackage{times}  
\usepackage{helvet}  
\usepackage{courier}  
\usepackage[hyphens]{url}  
\usepackage{graphicx} 
\usepackage{natbib}  
\usepackage{caption} 
\usepackage{subcaption}
\usepackage{algorithm}
\usepackage{algorithmicx}
\usepackage{algpseudocodex}
\usepackage{amsmath}
\usepackage{amssymb}
\usepackage{amsthm}
\usepackage{booktabs}
\usepackage{multirow}
\usepackage{makecell}
\usepackage{svg}

\usepackage{newfloat}
\usepackage{listings}
\DeclareCaptionStyle{ruled}{labelfont=normalfont,labelsep=colon,strut=off} 
\floatstyle{ruled}
\newfloat{listing}{tb}{lst}{}
\floatname{listing}{Listing}

\usepackage{xcolor}
\usepackage{cleveref} 

\newtheorem{definition}{Definition}
\newtheorem{theorem}{Theorem}
\newtheorem{lemma}{Lemma}
\newtheorem{corollary}{Corollary}
\newtheorem{proposition}{Proposition}

\title{Speedup Techniques for Switchable Temporal Plan Graph Optimization}
\author {
    He Jiang,
    Muhan Lin,
    Jiaoyang Li
}
\affiliations {
    Carnegie Mellon University\\
    \{hejiangrivers, muhanlin, jiaoyangli\}@cmu.edu
}

\usepackage{bibentry}

\begin{document}

\maketitle

\begin{abstract}
Multi-Agent Path Finding (MAPF) focuses on planning collision-free paths for multiple agents. However, during the execution of a MAPF plan, agents may encounter unexpected delays, which can lead to inefficiencies, deadlocks, or even collisions. To address these issues, the Switchable Temporal Plan Graph provides a framework for finding an acyclic Temporal Plan Graph with the minimum execution cost under delays, ensuring deadlock- and collision-free execution. Unfortunately, existing optimal algorithms, such as Mixed Integer Linear Programming and Graph-Based Switchable Edge Search (GSES), are often too slow for practical use. This paper introduces Improved GSES, which significantly accelerates GSES through four speedup techniques: stronger admissible heuristics, edge grouping, prioritized branching, and incremental implementation. Experiments conducted on four different map types with varying numbers of agents demonstrate that Improved GSES consistently achieves over twice the success rate of GSES and delivers up to a 30-fold speedup on instances where both methods successfully find solutions.

\end{abstract}

%

\section{Introduction}
Multi-Agent Path Finding (MAPF)~\citep{SternSoCS19} focuses on planning collision-free paths for multiple agents to navigate from their starting locations to destinations. However, during execution, unpredictable delays can arise due to mechanical differences, accidental events, or sim-to-real gaps. For instance, in autonomous vehicle systems, a pedestrian crossing might force a vehicle to stop, delaying its movements and potentially causing subsequent vehicles to halt as well. If such delays are not managed properly, they can lead to inefficiencies, deadlocks, or even collisions.

To address these issues, the Temporal Plan Graph (TPG) framework was first introduced to ensure deadlock- and collision-free execution \citep{honig2016multi, ma2017multi}. TPG encodes and enforces the order in which agents visit the same location using a directed acyclic graph, where directed edges represent precedence. However, the strict precedence constraints often lead to unnecessary waiting. For example, as shown in \Cref{fig: delayed_TPG}, a delay by Agent $1$ causes Agent $2$ to wait unnecessarily at vertex $F^2$ because TPG enforces the original order of visiting $G$. In such cases, switching the visiting order of these two agents would allow Agent $2$ to proceed first, avoiding unnecessary delays.

\begin{figure}[tb]
    \centering
    \begin{subfigure}[b]{0.15\textwidth}
      \raggedright
      \includegraphics[width=1\textwidth]{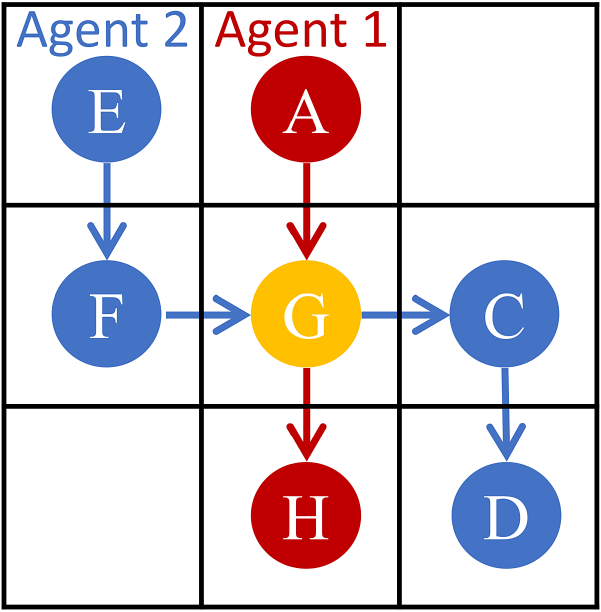}
      \caption{Grid}
      \label{fig: grid}
    \end{subfigure}%
    \hfill
    \begin{subfigure}[b]{0.31\textwidth}
      \raggedleft
      \includegraphics[width=1\textwidth]{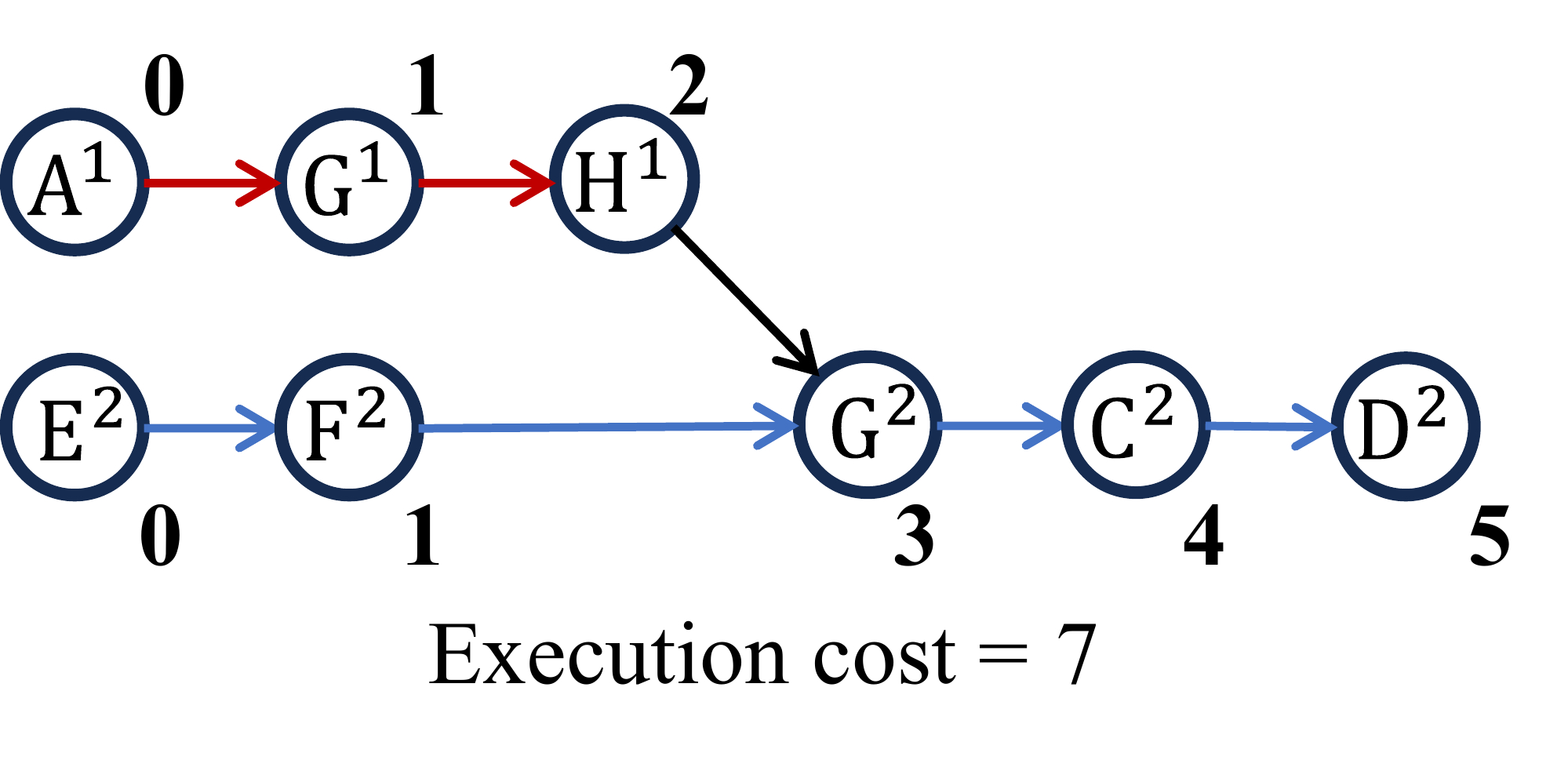}
      \caption{TPG}
      \label{fig: TPG}
    \end{subfigure}%
    \hfill

    \begin{subfigure}[b]{0.42\textwidth}
      \centering
      \includegraphics[width=1\textwidth]{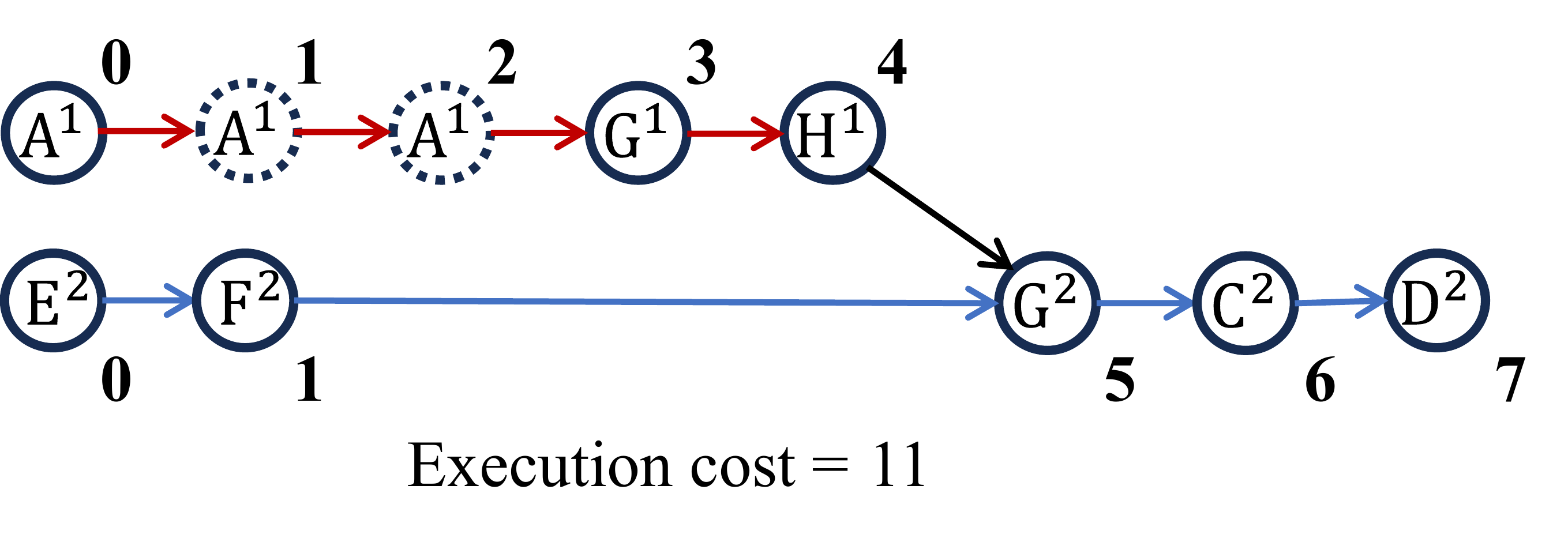}
      \caption{Delayed TPG}
      \label{fig: delayed_TPG}
    \end{subfigure}
    
    \caption{(a) shows an example of a MAPF problem. In the initial MAPF plan, Agent $1$  moves from $A$ to $H$ (red arrows) and visits $G$ first. Agent $2$ moves from $E$ to $D$ (blue arrows) and visits $G$ after Agent $1$. 
    (b) shows the TPG of the initial plan.
    Each row of vertices encodes an agent's path with superscripts  $1,2$ differentiating agents.
    The black arrow encodes the precedence that agent $2$ can only visit $G^2$ after agent $1$ arrives $H^1$.
    The bold number near a vertex $v^i$ shows the earliest possible time for agent $i$ to arrive at this vertex. The execution cost below the TPG shows the overall cost of these two agents. 
    (c) shows the TPG with a $2$-timestep delay at vertex $A^1$, which is encoded by two dashed vertices.
    }
    \label{fig: mapf_example}
\end{figure}

To capture this idea,~\citet{berndt2023receding} introduced the notion of a switchable edge that can be settled in one of two directions, each representing one possible visiting order of two agents at a location. Then, they proposed the Switchable Temporal Plan Graph (STPG), a superclass of TPG that contains a set of switchable edges. By gradually settling these switchable edges, STPG allows the search for a Temporal Plan Graph (TPG) that encodes optimal visiting orders to minimize total travel time under delays. However, both Berndt's original solution based on Mixed Integer Linear Programming (MILP) and the subsequent Graph-Based Switchable Edge Search (GSES) algorithm~\citep{FengICAPS24} are too slow to scale to scenarios with just 100 agents in the experiments.

This work makes a significant advancement in tackling the scalability issue of STPG. Specifically, we improve GSES with a series of speedup techniques that leverage the underlying structure of STPG. By estimating the future cost induced by currently unsettled switchable edges, we propose stronger admissible heuristics. To reduce the size of the search tree, we introduce a novel algorithm that identifies all the maximal groups of switchable edges whose directions can be determined simultaneously \citep{berndt2023receding}. Additionally, we introduce two straightforward yet highly effective engineering techniques---prioritized branching and incremental implementation---to enhance the efficiency of the search. We prove the correctness of these techniques and demonstrate through experiments that our final algorithm, Improved GSES (IGSES), significantly outperforms the baseline GSES and other approaches. For instance, IGSES consistently achieves more than double the success rates of GSES and shows a 10- to 30-fold speedup in average search time on instances solved by both algorithms.

\section{Related Work}
Approaches for handling delays in MAPF execution can be broadly categorized into offline and online methods.

Offline approaches consider delays during planning. Robust MAPF planning methods~\citep{atzmon2018robust,atzmon2020robust} generate robust plans under bounded or probabilistic delay assumptions. While these methods improve robustness, they tend to produce conservative solutions and require significantly more computation than standard MAPF algorithms.

Online approaches, on the other hand, react to delays as they occur. The simplest method is to replan paths for all agents upon a delay, but this is computationally expensive. TPG~\citep{honig2016multi,ma2017multi} avoids heavy replanning by adhering to the original paths and precedence constraints, but it often results in unnecessary waits due to overly strict ordering. To address this problem, \citet{berndt2023receding} introduce the STPG framework to optimize the precedence by MILP. Despite its generality, the MILP approach is too slow for large-scale problems. To improve the scalability, ~\citet{FengICAPS24} propose a dedicated search algorithm, GSES, to replace the MILP. However, GSES still suffers from inefficiencies due to the large search tree and redundant computation. In contrast,~\citet{kottinger2024introducing} proposed an alternative formulation that solves the same problem as STPG by introducing delays to agents' original plans. It constructs a graph for each agent and then applies standard MAPF algorithms to search for solutions. We refer to the optimal version of their algorithm, which applies Conflict-Based Search (CBS)~\citep{Sharon2015cbs}, as CBS with Delays (CBS-D). MILP, GSES, and CBS-D are all optimal algorithms, and we will compare our method against them in terms of planning speed in the experiments.

Different from these optimal methods, ~\citet{liu2024multi} propose a non-optimal heuristic approach called Location Dependency Graph (LDG), which reduces waits online using a formulation similar to STPG. Bidirectional TPG (BTPG)~\citep{su2024bidirectional} is another non-optimal approach based on TPG, but it falls into the offline category. BTPG post-processes a MAPF plan and produces an extended TPG with special bidirectional edges. These edges enable agents to switch visiting orders at certain locations in a first-come-first-served manner during execution. 

\section{Background}

In this section, we provide the necessary definitions and background for STPG optimization. For more details, please refer to the GSES paper~\cite{FengICAPS24}. 

\subsection{Preliminaries: MAPF, TPG, and STPG}
\begin{definition}[MAPF]
\label{definition: mapf}
MAPF problem aims at finding collision-free paths for a team of agents indexed by $i\in \mathcal{I}$ on a given graph, where each agent $i$ has a start location $s^i$ and a goal location $g^i$. At each timestep, an agent can move to an adjacent location or wait at its current location. We disallow two types of conflicts like previous works \cite{FengICAPS24}:
\begin{enumerate}
    \item \textbf{Vertex conflict}: two agents take the same location at the same timestep.
    \item  \textbf{Following conflict}: one agent enters a location occupied by another agent at the previous timestep.
\end{enumerate}

\end{definition}


\begin{definition} [TPG]
A Temporal Plan Graph (TPG) is a directed graph $\mathcal{G}=(\mathcal{V},\mathcal{E}_1,\mathcal{E}_2)$ that encodes a MAPF plan by recording its precedence of visiting locations. 

The vertex set $\mathcal{V}=\{v_p^i: p\in[0,z^i], i \in \mathcal{I}\}$ records all locations that must be visited sequentially by each agent $i$.\footnote{Following the previous work, we merge consecutive vertices that represent the same location into one.} Specifically, $z^i$ is the number of locations for agent $i$, and each $v_p^i$ is associated with a specific location, $loc(v_p^i)$.

The edge set $\mathcal{E}_1$ contains all \textbf{Type-1} edges, which encode the precedence between an agent's two consecutive vertices. A Type-1 edge $(v_p^i,v_{p+1}^i)$ must be introduced for each pair of $v_p^i$ and $v_{p+1}^i$ to ensure that $v_p^i$ must be visited before $v_{p+1}^i$.

The edge set $\mathcal{E}_2$ contains all \textbf{Type-2} edges, which specify the order of two agents visiting the same location. Exactly one Type-2 edge must be introduced for each pair of $v_q^j$ and $v_p^i, i\neq j$, with $loc(v_q^j)=loc(v_p^i)$. We can introduce the Type-2 edge $(v_{q+1}^j,v_{p}^i)$ to encode that agent $i$ can only enter $v_p^i$ after agent $j$ leaves $v_{q}^j$ and arrives at $v_{q+1}^j$, or introduce its \textbf{reversed edge} $(v_{p+1}^i,v_{q}^j)$ to specify that agent $j$ can only enter $v_{q}^j$ after agent $i$ leaves $v_p^i$ and arrives at $v_{p+1}^i$.

\end{definition}

For example, in \Cref{fig: TPG}, the red and blue arrows are Type-1 edges, and the black arrow is a Type-2 edge.

When agent $i$ moves along an edge $e$ from $v_p^i$ to $v_{p+1}^i$ and suffers from a $t$-timestep delay, we insert $t$ vertices along $e$ to encode this delay (e.g., \Cref{fig: delayed_TPG}) so that each edge in the figures always takes $1$ timestep for an agent to move.\footnote{In the implementation, we actually define edge costs to encode delays compactly. However, for easy understanding, we explain our algorithm by inserting extra vertices in this paper.} 

A TPG is executed in the following way. At each timestep, an agent $i$ can move from $v^i_p$ to $v^i_{p+1}$ if for every edge $e=(v^j_q, v^i_{p+1}$), agent $j$ has visited $v^j_q$. The execution of a TPG is completed when all the agents arrive at their goals.

\begin{theorem}
    \label{theorem: acyclic_TPG}
    The execution of a TPG can be completed without collisions in finite time if and only if it is acyclic.~\cite{berndt2023receding}.
\end{theorem}

We define \textbf{the execution cost of an acyclic TPG} as the minimum sum of travel time of all agents to complete the execution of the TPG. 
It can be obtained by simulation, 
but the following theorem provides a faster way to compute it.

\begin{definition}[EAT]
    The \textbf{earliest arrival time (EAT)} at a vertex $v^i_p$ is the earliest possible timestep that agent $i$ can arrive at $v^i_p$ in an execution of TPG.
\end{definition}

\begin{theorem}
\label{theorem: TPG cost}
    The EAT at a vertex $v_p^i$ is equal to the \textbf{forward longest path length (FLPL)}, $L(v_p^i)=\max\{L(s^k,v_p^i), k\in \mathcal{I}, \text{ if }s^k\text{ is connected to }v_p^i\}$, where $L(s^k,v_p^i)$ is the longest path length from agent $k$'s start, $s^k$, to $v_p^i$. The execution cost of an acyclic TPG is equal to the sum of all agents' EATs at their goals, $\sum_{i \in \mathcal{I}}{L(g^i)}$.  \citep{FengICAPS24}. 
\end{theorem}


Since an acyclic TPG is a directed acyclic graph, we can obtain the longest path length of all vertices by first topological sort and then dynamic programming in linear-time complexity $O(|\mathcal{V}|+|\mathcal{E}_1|+|\mathcal{E}_2|)$ \citep{berndt2023receding}. We provide the pseudocode in Appendix A.1.
\footnote{The appendix is available in the extended version of the paper at: https://arxiv.org/abs/2412.15908.}
For example, in \Cref{fig: delayed_TPG}, we annotate the earliest arrival time at each vertex near the circles. The EAT of $G^2$ can be computed by $\max\{L(H^1)+1,L(F^2)+1\}=\max\{5,2\}=5$. Because the EAT of two goals $H^1$ and $D^2$ are $4$ and $7$, the execution cost of this TPG is $4+7=11$.

\begin{definition} [STPG]
\label{definition: STPG}
Given a TPG $\mathcal{G}=(\mathcal{V},\mathcal{E}_1,\mathcal{E}_2)$, a Switchable TPG (STPG) $\mathcal{G}^S=(\mathcal{V},\mathcal{E}_1,(\mathcal{S},\mathcal{N}))$ partitions Type-2 edges $\mathcal{E}_2$ into two disjoint edge sets, $\mathcal{S}$ for switchable edges and $\mathcal{N}$ for non-switchable edges. Any switchable edge $e=(v_{q+1}^j, v_p^i) \in \mathcal{S}$ can be \textbf{settled} to a non-switchable edge by one of the following two operations:
\begin{enumerate}
    \item \textbf{Fix} e: removes $e$ from $\mathcal{S}$ and adds it to $\mathcal{N}$.
    \item \textbf{Reverse} e: removes $e$  from $\mathcal{S}$ and adds its \textbf{reversed edge} $Re(e)=(v_{p+1}^i,v_q^j)$ to $\mathcal{N}$. 
\end{enumerate}
For convenience, we also define the \textbf{direction of an edge} $e=(v_{q+1}^j, v_p^i)$ as from $j$ to $i$ and its reverse direction as from $i$ to $j$, where $i,j$ are the indices of agents.
\end{definition}

\begin{figure*}[t]
\centering
\includegraphics[width=1.0\textwidth]{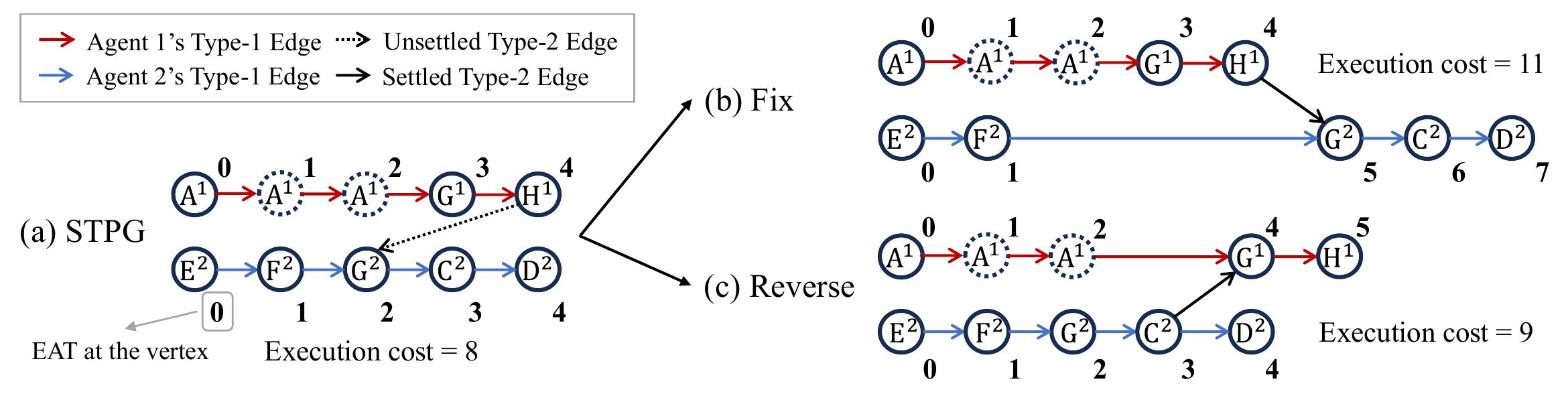}

\caption{(a) is the STPG built from the TPG in \Cref{fig: delayed_TPG}. The only switchable edge $(H^1,G^2)$ is dashed. It can be fixed to edge $(H^1,G^2)$ as in (b) or reversed to edge $(C^2,G^1)$ as in (c). The different choices result in different execution costs. Notably, when computing the execution cost of STPG (defined in \Cref{subsection: baseline}), the dashed switchable edge will be ignored.}
\label{fig: switch}
\end{figure*}

Clearly, STPG is a superclass of TPG. If all switchable edges are settled, it \textbf{produces} a TPG. The goal of \textbf{STPG Optimization} is to find an acyclic TPG with the minimum execution cost among all those produced by the STPG.

\subsection{The Baseline Algorithm: GSES}
\label{subsection: baseline}

\begin{algorithm}[htb]
\caption{\textcolor{red}{Improved} GSES}
\label{algo: GSES}
\begin{algorithmic}[1]
\Function{\textcolor{red}{Improved}GSES}{$\mathcal{G}_{init}^S$}
    \State \textcolor{red}{$Groups \gets \Call{EdgeGrouping}{\mathcal{G}_{init}^S}$}
        \label{line: call edge grouping}
    \State $\text{Node}_{root} \gets \Call{BuildNode}{\mathcal{G}_{init}^S}$
    \State Push $\text{Node}_{root}$ into OpenList
    \While{OpenList is not empty and still time left}
        \State Pop $\text{Node}=(\mathcal{G}^S,\text{$h$-value},L)$ from OpenList
        \State $e \gets \Call{SelectConflictingEdge}{\mathcal{G}^S,L}$
            \label{line: select conflicting edge}
        \If{$e= \text{NULL}$}
            \label{line: solution found start}
            \Return $\Call{FixAll}{\mathcal{G}^S}$
                \label{line: solution found end}
        \EndIf
            \State $Edges \gets \{e\}$
            \State \textcolor{red}{$Edges \gets \Call{getEdgeGroup}{e,Groups}$}
                \label{line: call get edge group}
            \State $\mathcal{G}_{\text{child-1}}^S \gets \Call{Fix}{\mathcal{G}^S, Edges}$
                \label{line: fix edge}
            \State $\mathcal{G}_{\text{child-2}}^S \gets \Call{Reverse}{\mathcal{G}^S, Edges}$
                \label{line: reverse edge}
            \For {$\mathcal{G}_{\text{child}}^S\in \{\mathcal{G}_{\text{child-1}}^S,\mathcal{G}_{\text{child-2}}^S\}$}
                \If {no cycle in $Redu(\mathcal{G}_{\text{child}}^S)$}
                    \label{line: check cycle}
                    \State $\text{Node}_{\text{child}} \gets \Call{BuildNode}{\mathcal{G}_{\text{child}}^S}$
                    \State Push $\text{Node}_{\text{child}}$ into OpenList
                \EndIf
            \EndFor   
    \EndWhile
    \State \Return NULL // Timeout
\EndFunction
\State
\Function{BuildNode}{$\mathcal{G}^S$} 
    \label{func: build_node}
    \State \textcolor{red}{Incrementally} compute for all vertices the forward \textcolor{red}{and backward} longest path lengths on {$Redu(\mathcal{G}^S)$} as a function $L$ 
    \textcolor{red}{based on its parent node}
    \label{line: longest path lengths}
    \State \textcolor{red}{Estimate the future cost increase $\Delta cost$ based on the backward longest path lengths}
    \State $\text{$h$-value} \gets \sum_i{L(g^i)}\textcolor{red}{+\Delta cost}$
    \label{line: h_value}
    \State \Return $(\mathcal{G}^S,\text{$h$-value},L)$
        \label{line: node}
\EndFunction
\State
\Function{SelectConflictingEdge}{$\mathcal{G}^S,L$}
    \State \textcolor{red}{Order all switchable edges in $\mathcal{S}$ by certain priority}
        \label{line: prioritizing branching}
    \ForAll {switchable edge $e\in \mathcal{S}$}
        \If {$Sl(e)<0$}
            \label{line: conflict edge}\Return e
        \EndIf
    \EndFor
    \State \Return NULL
        \label{line: conflicting edge not found}
\EndFunction
\end{algorithmic}
\end{algorithm}

We now introduce our baseline algorithm, Graph-based Switchable Edge Search (GSES)~\citep{FengICAPS24}. We start with some useful definitions.

\begin{definition} [Reduced TPG]
The reduced TPG of an STPG $\mathcal{G}^S=(\mathcal{V},\mathcal{E}_1,(\mathcal{S},\mathcal{N}))$ is the TPG that omits all switchable edges, denoted as $Redu(\mathcal{G}^S)=(\mathcal{V},\mathcal{E}_1,\mathcal{N})$.
\end{definition}

Then, we define the \textbf{execution cost of an STPG} to be the execution cost of its reduced TPG. Clearly, it provides a lower bound for the execution cost of any acyclic TPG that can be produced from this STPG because the execution cost will only increase with more switchable edges settled.

\begin{definition}[Edge Slack]
    \label{definition: edge slack}
    The \textbf{slack of a switchable edge} $e=(v_{q+1}^j,v_p^{i})$ in an STPG is defined as $Sl(e)=L(v_p^{i})-L(v_{q+1}^j)-1$, where $L(v)$ is the EAT at vertex $v$ in the reduced TPG. 
\end{definition}

We say a switchable edge \textbf{conflicts} with the STPG if $Sl(e)<0$. For example, in the \Cref{fig: switch}a, the switchable edge $(H^1,G^2)$ conflicts with the STPG because its slack $Sl(H^1,G^2)=L(G^2)-L(H^1)-1=2-4-1=-3<0$.
For a switchable edge $e$, $Sl(e)\geq 0$ means that the earliest execution of the reduced TPG already satisfied the constraints imposed by $e$, i.e., we can fix $e$ without introducing any cycle or cost increase to the reduced TPG. Therefore, if all remaining switchable edges do not conflict, we can fix them all and obtain an acyclic TPG with the same execution cost as the current reduced TPG \cite{FengICAPS24}. The proof is given in Appendix A.2. This result can help the GSES algorithm, introduced next, stop the search earlier.

\Cref{algo: GSES} shows the pseudocode of GSES,\footnote{The red text belongs to IGSES and can be ignored for now.} a best-first search algorithm with each search node comprising three parts: an STPG, its execution cost, and the EATs of all vertices in its reduced TPG (\Cref{line: node}). The execution cost is used as the $h$-value in the best-first search (The $g$-value is always $0$), and the EATs are used to find conflicting edges.

GSES starts with a root node containing the initial STPG, which is constructed from the TPG of the initial MAPF plan. Almost all Type-2 edges in the TPG can be turned into switchable edges except those pointing to goal vertices and those whose reversed edges point to start vertices \citep{berndt2023receding}. For example, the STPG built from the TPG in \Cref{fig: delayed_TPG} is illustrated in \Cref{fig: switch}a.

When GSES expands a search node, it tries to find a conflicting edge (\Cref{line: select conflicting edge}). If there is none, GSES terminates the search and returns a solution by fixing all the remaining switchable edges (\Cref{line: solution found end}). 
Otherwise, GSES branches the search tree based on the 
conflicting edge. Two child STPGs are generated: one fixes the edge (\Cref{line: fix edge}), and another one reverses it (\Cref{line: reverse edge}). Before generating the child nodes and pushing them into the priority queue, GSES detects cycles in their reduced TPGs and discards cyclic ones (\Cref{line: check cycle}). The longest path lengths (\Cref{line: longest path lengths}) and the $h$-value (\Cref{line: h_value}) are computed for acyclic ones to build child nodes.

\section{Speedup Techniques}
\label{section: techniques}
This section describes our methods to speed up the search in \Cref{algo: GSES}, including constructing stronger admissible heuristics by estimating the future cost (\Cref{subsection: heuristics}), finding switchable edges that could be grouped and branched together (\Cref{subsection: grouping}), different ways to prioritize switchable edges for branching (\Cref{subsection: branching}), and incremental implementation of computing longest path lengths (\Cref{subsection: incremental}).

\subsection{Stronger Admissible Heuristics}
\label{subsection: heuristics}
In the \Cref{line: h_value} of \Cref{algo: GSES}, we use the execution cost of the current STPG
as the admissible $h$-value in the GSES. However, this cost is an estimation based on the reduced TPG with currently settled edges.

Intuitively, we can obtain extra information from the unsettled switchable edges. For example, in \Cref{fig: switch}a, the execution cost of the STPG is $8$. The switchable edge $(H^1,G^2)$ is not settled and, thus, ignored in the computation. 
But we know eventually, this edge will be fixed or reversed. If we fix it (\Cref{fig: switch}b), then the EAT of $G^2$ will increase from $2$ to $5$ because now it must be visited after $H^1$, the EAT of which is $4$. In this way, we can infer that the EAT of vertex $D^2$ will be postponed to 7, leading to an increase of 3 in the overall execution cost. If we reverse the edge, we will get a similar estimation of the future cost increase, which is $1$ in \Cref{fig: switch}c. Then we know the overall cost will increase by at least $\min\{3,1\}=1$ in the future for the STPG in \Cref{fig: switch}a. 

Before the formal reasoning, we introduce another useful concept, vertex slack, for easier understanding later.

\begin{definition} [Vertex Slack]
    The \textbf{slack of a vertex} $v_p^i$ to an agent $j$'s goal location $g^j$ in an STPG is defined as $Sl(v_p^i,g^j)=L(g^j)-L(v_p^i)-L(v_p^i,g^j)$, if $v_p^i$ is connected to $g^j$ in the reduced TPG. $L(v)$ is the EAT at vertex $v$, which is also the forward longest path length. $L(v_p^i,g^j)$ means the longest path length from $v_p^i$ to $g^j$ and we call it the \textbf{backward longest path length (BLPL)}.
\end{definition}

This slack term measures the maximum amount that we can increase the EAT at $v_p^i$ (i.e., $L(v_p^i)$) without increasing the agent $j$'s EAT at its goal (i.e., $L(g^j)$). In other words, $L(v_p^i)+Sl(v_p^i,g^j)$ is the latest time that agent $i$ can reach $v_p^i$ without increasing agent $j$'s execution time. In \Cref{fig: switch}a, the EAT of $G^2$ is $L(G^2)=2$, the longest path length from $G^2$ to $D^2$ is $L(G^2,D^2)=2$ and the EAT of $D^2$ is $L(D^2)=4$. Then the slack of $G^2$ is $Sl(G^2,D^2)=4-2-2=0$. It means that there is no slack, and any increase in the EAT of $G^2$ will be reflected in the EAT of $D^2$. Thus, in \Cref{fig: switch}b, where $L(G^2)$ increases by $3$ , $L(D^2)$ also increases by $3$. 

It is worth mentioning that $L(v_p^i,g^j)$ cannot be obtained during the computation of $L(v_p^i)$ and should be computed again by topological sort and dynamic programming in the backward direction (\Cref{line: longest path lengths} of \Cref{algo: GSES}). Unfortunately, in this computation, we need to maintain the backward longest path lengths from $v_p^i$ to each goal location $g^j$ separately because the increase in a vertex $v_p^i$'s EAT may influence other agents' EATs at their goals, regarding the graph structure. Thus, it takes $O(|I|(|\mathcal{V}|+|\mathcal{N}|))$ time for each STPG to compute $L(v_p^i,g^j)$, while $L(v_p^i)$ only takes $O(|\mathcal{V}|+|\mathcal{N}|)$. The pseudocode is provided in Appendix A.1.

We use the subscript $d$ to represent a descendant search tree node. For example, $L_{d}(v_1,v_2)$ is the longest path length from $v_1$ to $v_2$ in a descendant search tree node $d$. A simple fact is that $L_{d}(v_1,v_2)\geq L(v_1,v_2)$, namely, the longest path length is non-decreasing from an ancestor to a descendant.

Now, we start our reasoning by assuming that we fix a switchable edge $e=(v_{q+1}^j,v_p^{i})$ and add it to the reduced TPG of a descendant search tree node, then we can deduce the increase in the EAT of $v_p^i$. Specifically, $L_{d}(v_p^{i})\geq L_{d}(v_{q+1}^j)+1 \geq L(v_{q+1}^j)+1$. The first inequality holds because of the precedence implied by the new edge. The second inequality exploits the monotonicity of the longest path length.
Thus, the increase of $L(v_p^{i})$ is $\Delta L(v_p^{i})\triangleq L_{d}(v_p^{i})-L(v_p^{i}) \geq  L(v_{q+1}^j)+1-L(v_p^{i})=-Sl(e)$. Namely, the EAT at $v_p^{i}$ will be increased by at least the negative edge slack.

Next, we consider the increase in the EAT of agent $m$'s goal $g^m$, if vertex $v_p^{i}$ is connected to $g^m$. $L_{d}(g^m)\geq L_{d}(v_p^{i})+L_{d}(v_p^{i},g^m) \geq  (L(v_p^{i})+\Delta L(v_p^{i}))+L(v_p^{i},g^m) = \Delta L(v_p^{i})+(L(v_p^{i})+L(v_p^{i},g^m)) = \Delta L(v_p^{i})+(L(g^m)-Sl(v_p^{i},g^m))$. The first inequality is based on the fact that the longest path must be no shorter than any path passing a specific vertex. The second inequality is based on the definition of $\Delta L(v_p^{i})$ and the monotonicity of the longest path length. The last equality is obtained by plugging in the definition of vertex slack. So, we get $L_{d}(g^m)\geq \Delta L(v_p^{i})+(L(g^m)-Sl(v_p^{i},g^m))$. Then, we can move $L(g^m)$ to the left side and obtain the increase $\Delta L(g^m)\triangleq L_{d}(g^m)-L(g^m) \geq \Delta L(v_p^{i})-Sl(v_p^{i},g^m)$. 
Namely, if the increase in EAT at $v_p^{i}$ is larger than the vertex slack at $v_p^{i}$, the EAT at $g^m$ will increase by at least their difference.
Since we obtain $\Delta L(v_p^{i})\geq -Sl(e)$ earlier, we have $\Delta L(g^m) \geq -Sl(e)-Sl(v_p^i,g^m)$.\footnote{Due to the monotonicity, we also have $\Delta L(g^m)=L_{d}(g^m)-L(g^m) \geq 0$. But we will omit this trivial condition for simplicity.}

Similarly, if we reverse the edge $e=(v_{q+1}^j,v_p^{i})$ and add $Re(e)=(v_{p+1}^i,v_q^j)$ to the STPG of a descendant search tree node, we can deduce a similar result for the EAT of agent $n$'s goal $g^n$, $\Delta L(g^n) \geq -Sl(Re(e))-Sl(v_q^{j},g^n)$, if $v_q^j$ is connected to agent $n$'s goal vertex $g^n$.

Since we must either fix or reverse to settle a switchable edge $e$ eventually, we can get a conservative estimation of the future joint increase in $L(g^m)+L(g^n)$, $\Delta L(g^m)+\Delta L(g^n) \geq \Delta cost(g^m,g^n,e) \triangleq  \min\{-Sl(e)-Sl(v_p^{i},g^m),-Sl(Re(e))-Sl(v_q^{j},g^n)\}$. Further, if there are multiple such edges, we can take the maximum of all $\Delta cost(g^m,g^n,e)$. Namely, $\Delta L(g^m)+\Delta L(g^n) \geq \Delta cost(g^m,g^n) \triangleq \max_e\{cost(g^m,g^n,e)\}$.

In summary, we can now get a conservative estimation of the future increase in pairwise joint costs. Then, we can apply Weighted Pairwise Dependency Graph~\citep{li2019improved} to obtain a lower-bound estimation of the overall future cost increase for all agents. Specifically, we build a weighted fully-connected undirected graph $G_{D}=(V_{D},E_{D},W_{D})$, where each vertex $u_i\in V_{D}$ represents an agent, $E_D$ are edges and $W_D$ are edges' weights. The weight of an edge $(u_i,u_j)\in E_{D}$ is defined as the pairwise cost increase, $\Delta cost(g^{i},g^{j})$. Our target is to assign a cost increase $x_i$ to each vertex $u_i$ such that $x_i+x_j\geq \Delta cost(g^{i},g^{j})$ and the overall cost increase $\sum_i x_i$ is minimized. This formulation turns the original problem into an Edge-Weighted Minimum Vertex Cover (EWMVC) problem, an extension of the NP-hard Minimum Vertex Cover problem. Therefore, we choose a fast greedy matching algorithm implemented in the paper ~\citep{li2019improved} with worst-case $O(|V_{D}|^3)$ complexity to get an underestimation of the minimum overall cost increase, denoted as $\Delta cost$. The details of the matching algorithm are given in Appendix A.3. Since the future increase in the overall cost must be no less than the non-negative term, $\Delta cost$, combined with the original $h$-value, $\sum_{i}L(g^i)$, we get a stronger admissible heuristic value, $\sum_{i}L(g^i)+\Delta cost$.

We summarize the computation of stronger admissible heuristics in \Cref{algo: CH}. \Cref{line: CH_compute_forward}-\Cref{line: CH_compute_backward} compute the forward and backward longest path lengths. \Cref{line: CH_compute_delta_cost_start}-\Cref{line: CH_compute_delta_cost_end} estimate the cost increase for each pair of agents. \Cref{line: CH_build_pdg}-\Cref{line: CH_return_heuristics} estimate the overall cost increase.

\begin{algorithm}[tb]
\caption{Compute Stronger Admissible Heuristics}
\label{algo: CH}
\begin{algorithmic}[1]
\Function{ComputeHeuristics}{$\mathcal{G}^S$}
    \State Compute the forward longest path lengths $L(v)$ for all vertices 
        \label{line: CH_compute_forward}
    \State Compute the backward longest path lengths $L(v,g)$ between all vertices and goals ($L(v,g)=\infty$ if $v$ and $g$ are not connected)
        \label{line: CH_compute_backward}
    \State $\Delta cost(g^m,g^n) \gets 0, \forall m,n$
        \label{line: CH_compute_delta_cost_start}
    \ForAll{switchable edge $e=(v_{q+1}^j,v_p^{i})$}
        \ForAll{agent pairs $(m,n)$}
            \If{$L(v_p^{i},g^m)\neq \infty \land L(v_{q}^j,g^n)\neq \infty$}
                \State $\Delta c \gets \min\{-Sl(e)-Sl(v_p^{i},g^m),$
                \State \qquad \qquad $-Sl(Re(e))-Sl(v_q^{j},g^n)\}$
                \If{$\Delta c > \Delta cost(g^m,g^n)$}
                    \State$\Delta cost(g^m,g^n) \gets \Delta c$
                    \State$\Delta cost(g^n,g^m) \gets \Delta c$
                \EndIf
            \EndIf
        \EndFor
    \EndFor
        \label{line: CH_compute_delta_cost_end}
    \State Build the weighted pairwise dependency graph $G_D$ with $\Delta cost(g^m,g^n)$ as edge weights 
        \label{line: CH_build_pdg}
    \State Suboptimally solve the edge-weighted minimum vertex cover of $G_D$ by greedy matching to obtain the overall cost increase $\Delta cost$
        \label{line: CH_solve_ewmvc}
    \State \Return $\sum_i L(g^i)+\Delta cost$
        \label{line: CH_return_heuristics}
\EndFunction
\end{algorithmic}
\end{algorithm}

\subsection{Edge Grouping}
\label{subsection: grouping}

\begin{figure}[tb]
    \centering

    \begin{subfigure}[b]{0.15\textwidth}
      \centering
      \includegraphics[width=1\textwidth]{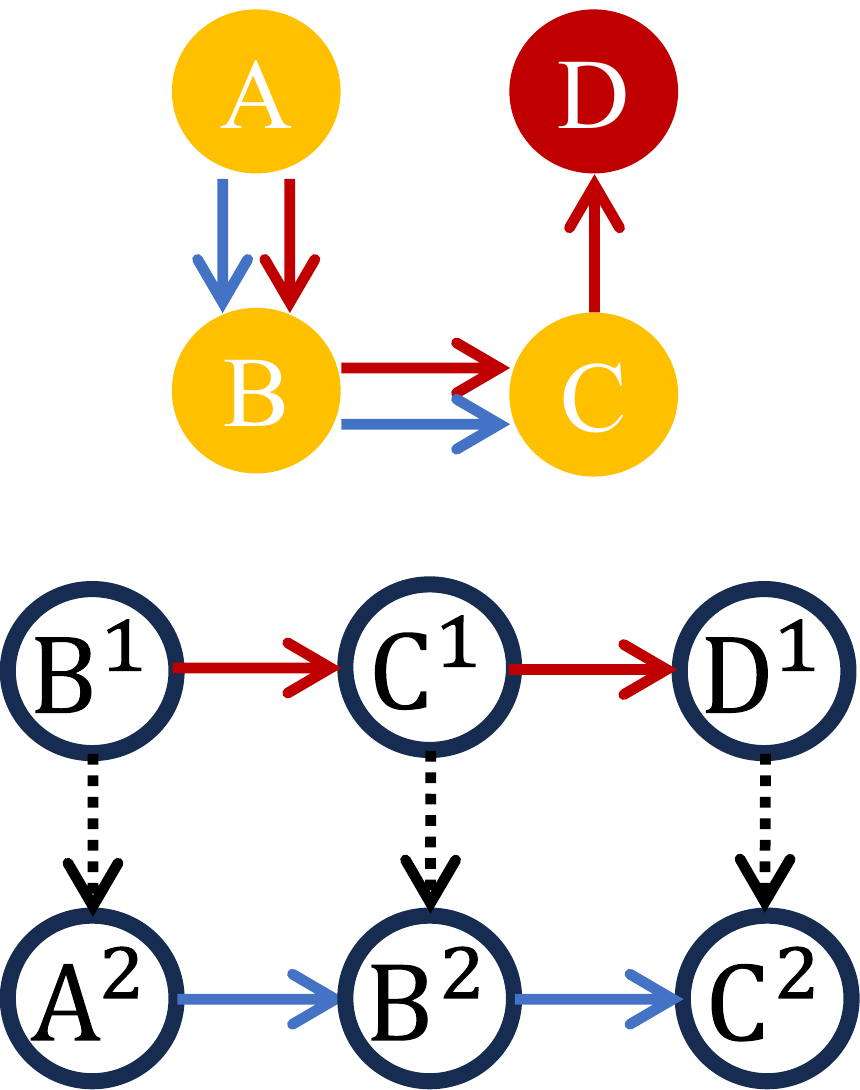}
      \caption{Parallel}
      \label{fig: parallel}
    \end{subfigure}%
    \hfill
    \begin{subfigure}[b]{0.15\textwidth}
      \centering
      \includegraphics[width=1\textwidth]{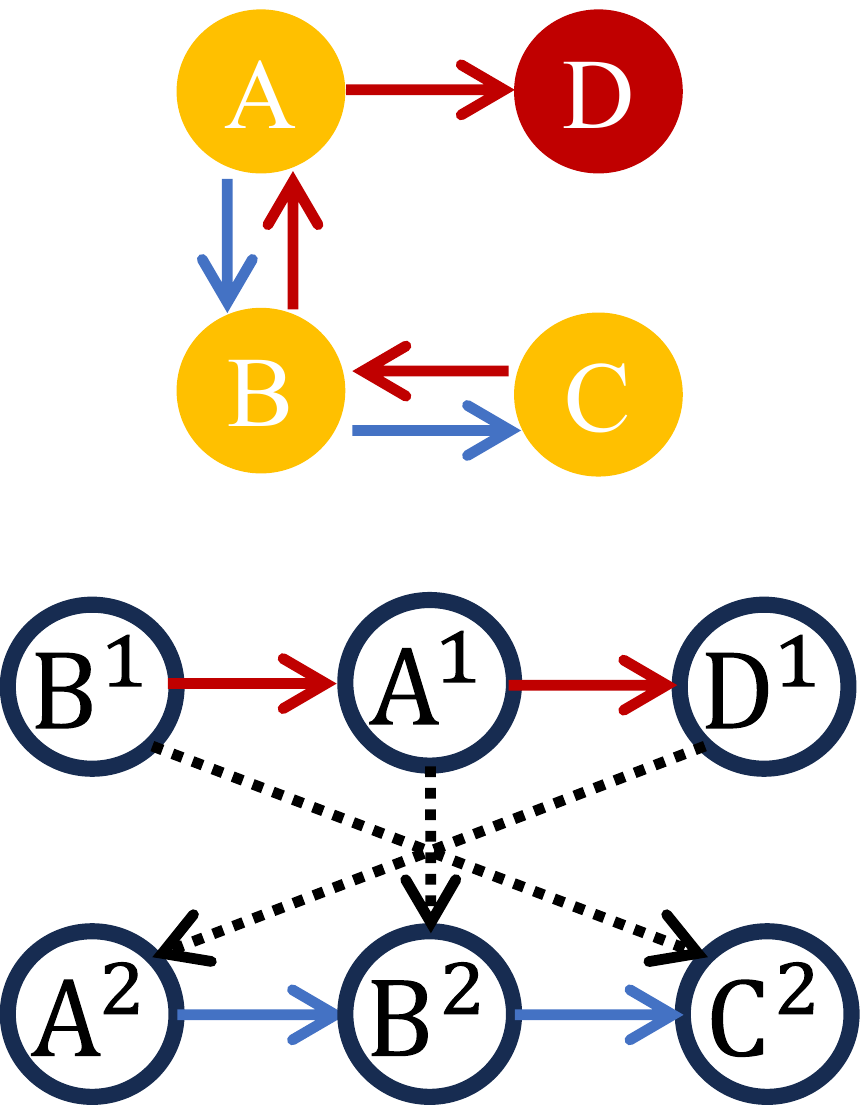}
      \caption{Crossing}
      \label{fig: crossing}
    \end{subfigure}%
    \hfill
    \begin{subfigure}[b]{0.15\textwidth}
      \centering
      \includegraphics[width=1\textwidth]{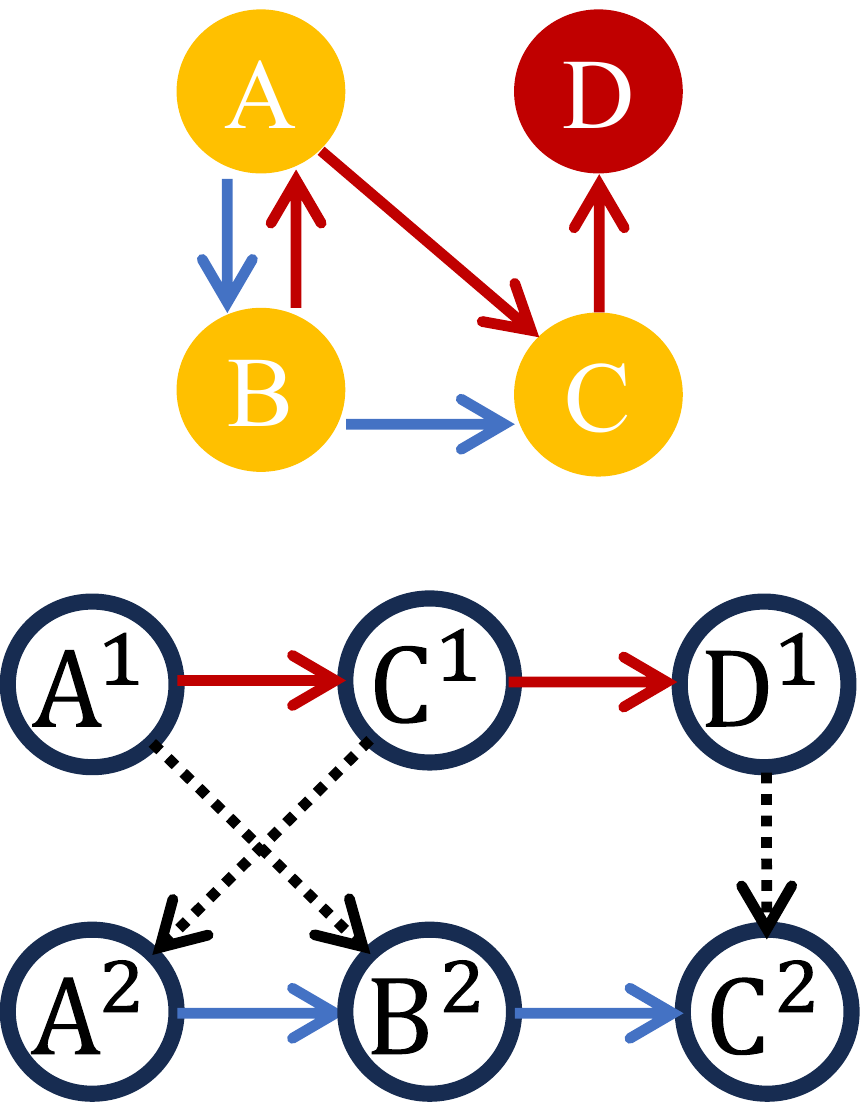}
      \caption{Mixed}
      \label{fig: mixed}
    \end{subfigure}
    
    \caption{Edge grouping examples. The upper graphs show the partial MAPF trajectories. The lower graphs show the partial STPG with switchable edges that can be grouped. In (a), Agent $2$ (blue arrows) visits the same sequence of locations, $A,B,C$, as Agent $1$ (red arrows). In a successful execution, one agent must visit each location before the other agent. Namely, the visiting orders of each location must be the same. Therefore, these edges must always have the same direction and are groupable. In (b), Agent $2$ traverses Agent $1$'s trajectory reversely, but the conclusion remains the same.
    (c) can be regarded as a mixture of (a) and (b).}
    \label{fig: grouping}
\end{figure}

In the MILP-based method, \citet{berndt2023receding} proposed a speedup method named dependency grouping, explicitly called edge grouping in this work. Edge grouping tries to find switchable edges whose directions can be decided together. If these edges do not follow the same direction, we can easily find a cycle within this group of edges. \citet{berndt2023receding} describe two obvious grouping patterns, parallel and crossing (\Cref{fig: parallel} and \Cref{fig: crossing}), which can be easily detected by a linear scan over all switchable edges.
However, they did not discuss whether there are other grouping patterns and how to find them all. Indeed, a counterexample is given in \Cref{fig: mixed}, which should be considered a single group but will be detected as two by their simple algorithm. 

In this work, we devise an algorithm for finding all the maximal groups among switchable edges from agent $i$ to agent $j$ to reduce the search tree size. We name our method \textbf{full grouping} and the old one \textbf{simple grouping} to differentiate them. This algorithm will be called before the best-first search as a preprocessing step to the initial STPG (\Cref{line: call edge grouping} of \Cref{algo: GSES}) so that we consider an edge group rather than a single edge at each branch during the search (\Cref{line: call get edge group} of \Cref{algo: GSES}). 
For the simplicity of the discussion, we assume the initial STPG has only switchable Type-2 edges since we can convert the non-switchable edges to switchable ones and re-settle their directions after the grouping. First, we formalize our definitions.

\begin{definition} [Ordered-Pairwise Subgraph $\mathcal{G}_{i,j}^S$]
    The ordered-pairwise subgraph $\mathcal{G}_{i,j}^S=(\mathcal{V}_{i,j},\mathcal{E}_{1\ i,j},(\mathcal{S}_{i,j}, \mathcal{N}_{i,j}))$ for two agents $i,j$ is a subgraph of an STPG $\mathcal{G}^S$ with only vertices of these two agents, their type-1 edges and all the Type-2 edges pointing from agent $i$ to agent $j$.
\end{definition}

\begin{definition} [Groupable]
    For a subgraph $\mathcal{G}_{i,j}^S$, two switchable edges $e_m,e_n\in \mathcal{S}_{i,j}$ are groupable, if, for all the two-agent acyclic TPGs that  $\mathcal{G}_{i,j}^S$ can produce, either both $e_m,e_n$ or both $Re(e_m),Re(e_n)$ are in them.
\end{definition}

Apparently, groupable is an equivalence relation that is reflexive, symmetric and transitive. It divides set $\mathcal{S}_{i,j}$ into a set of disjoint equivalence classes, which are called \textbf{maximal edge groups} in our setting. The following definition of the maximal edge group is a rephrase of the equivalence class defined by the groupable relation.

\begin{definition} [Maximal Edge Group]
    \label{definition: edge group}
    A maximal edge group $\mathcal{EG}$ of an edge $e$ is a subset of $\mathcal{S}_{i,j}$, which contains exactly all the edges groupable with $e$.
\end{definition}

Based on the property of the equivalence relation, the maximal edge group of an edge $e$ must be unique. Thus, we apply a simple framework to find all the maximal edge groups in \Cref{algo: EGF}. We initialize an edge set with all the switchable edges (\Cref{line: init edge set}). We select one edge from the edge set and find all the edges that are groupable with it (\Cref{line: call find group}), then remove them from the edge set (\Cref{line: remove edge group}). We repeat this process for the remaining edges until all the maximal edge groups are found and then removed. 

\begin{figure*}
    \centering
    \includegraphics[width=1\linewidth]{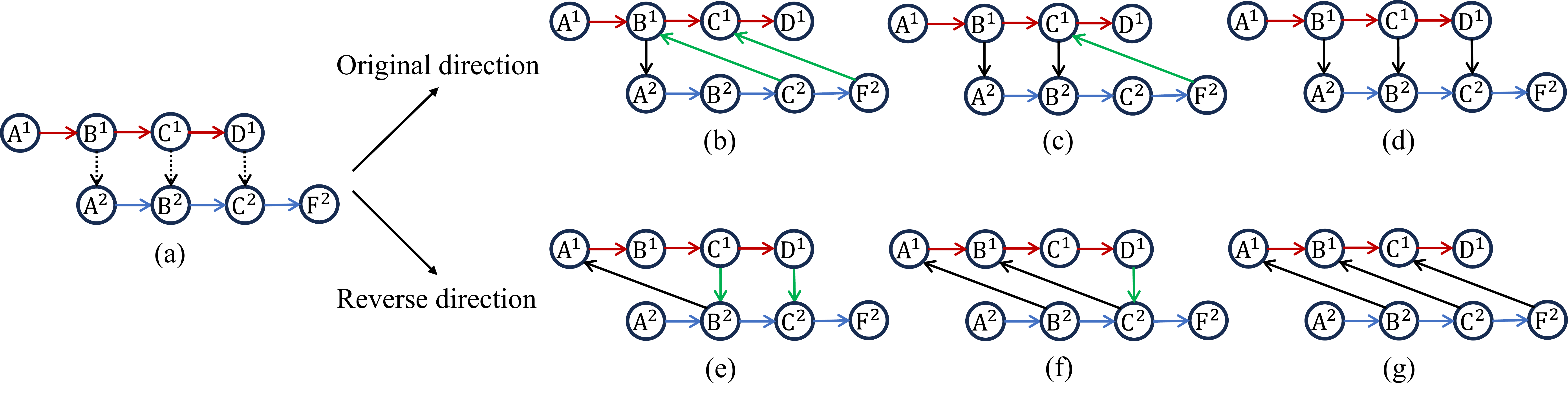}
    \caption{An example of finding all the switchable edges groupable with a certain switchable edge $e$. (a) is a complete STPG for the parallel pattern in \Cref{fig: parallel}. We want to find all the switchable edges groupable with $(B^1,A^2)$ in (a).
    We must consider two possible directions of it. (b), (c), (d) illustrate the reasoning for the case we fix it. Specifically, in (b), we fix $(B^1,A^2)$ (marked black) and temporarily reverse other switchable edges (marked green) to reason whether they must follow the same direction as $(B^1,A^2)$ based on the constraint that there should be no cycle in the graph. Since there is a cycle $C^2\rightarrow B^1\rightarrow A^2$ in (b), we know $(C^1,B^2)$ in (a) must be fixed instead of reversed. Then, we obtain (c). Similarly, there is a cycle $B^2 \rightarrow F^2 \rightarrow C^1$. So, $(D^1,C^2)$ must also be fixed. Then, we obtain (d). As a result, we find both edges $(C^1,B^2)$ and $(D^1,C^2)$ need to follow the same direction as $(B^1,A^2)$ if it is fixed. (e), (f), (g) illustrate the reasoning for the case that we reverse the switchable edge $(B^1,A^2)$. We confirm that both edges $(C^1,B^2)$ and $(D^1,C^2)$ need to follow the same direction as $(B^1,A^2)$ if it is reversed. Therefore, these three switchable edges must always select the same direction. Namely, they are groupable.
    }
    \label{fig: grouping_procedure}
\end{figure*}

\begin{algorithm}[tb]
\caption{Edge Grouping}
\label{algo: EGF}
\begin{algorithmic}[1]
\Function{EdgeGrouping}{$\mathcal{G}_{i,j}^S$}
    \State $Groups\gets \{\}, Edges\gets \mathcal{S}_{i,j}$ \label{line: init edge set}
    \While {$Edges$ is not emtpy}
        \State select any edge $e\in Edges$
        \State $\mathcal{EG}\gets \Call{FindGroupableEdges}{Edges, e}$ \label{line: call find group}
        \State $Groups \gets Groups \cup \{\mathcal{EG}\}$
        \State $Edges = Edges - \mathcal{EG}$ \label{line: remove edge group}
    \EndWhile
    \State \Return $Groups$
\EndFunction
\end{algorithmic}

\end{algorithm}
\begin{algorithm}[tb]
\caption{FindGroupableEdges}
\label{algo: FGE}
\begin{algorithmic}[1]
\Function{FindGroupableEdges}{$Edges, e$}
    \State // Reason with the reverse direction of $e$
    \State $S_1 \gets \Call{Propagate}{Edges, e}$
        \label{line: propagate original direction}
    \State // Reason with the original direction of $e$
    \State $E_r \gets \Call{Reverse}{Edges}$
        \label{line: propagate reverse direction start}
    \State $e_r \gets \Call{Reverse}{e}$
    \State $S_r \gets \Call{Propagate}{E_r, e_r}$
        \label{line: propagate reverse direction}
    \State $S_2 \gets \Call{Reverse}{S_r}$
        \label{line: propagate reverse direction end}
    \State // Obtain the edge group by intersection
    \State $\mathcal{EG} \gets S_1 \cap S_2$
        \label{line: intersection}
    \State \Return $\mathcal{EG}$
\EndFunction

\Function{Propagate}{$Edges,e$}
    \State $S \gets \{e\}, C\gets\{e\}, E\gets Edges-\{e\}$
        \label{line: initialize edge sets}
    \Repeat
        \State $C_r=\Call{Reverse}{C}$ 
            \label{line: reverse newly found edges}
        \State $C\gets \Call{FindCycleEdges}{E,C_r}$
            \label{line: call find cycle edges}
        \State $S\gets S\cup C$
        \State $E\gets E-C$
    \Until{$C$ is empty}
        \label{line: empty cycle edges}
    \State \Return S
\EndFunction

\Function{FindCycleEdges}{$E,C_r$}
    \State $C\gets \{\}$
    \ForAll{$e=(v^i_m,v^j_n) \in E$}
        \ForAll{$e_r=(v^j_q,v^i_p) \in C_r$}
            \If{$p\leq m \land n\leq q$}
                \State $C \gets C\cup \{e\}$
            \EndIf
        \EndFor
    \EndFor
    \State \Return $C$
\EndFunction

\end{algorithmic}
\end{algorithm}

Before discussing how to identify all groupable edges, we introduce a lemma that enables cycle detection in a two-agent TPG by examining only the vertex indexes of each opposite-direction edge pair.

\begin{lemma}
    \label{lemma: two-agent cycle detection}
    If a two-agent TPG is cyclic, we must be able to find a cycle that contains exactly two Type-2 edges, one edge pointing from agent $i$ to $j$ and another edge pointing reversely. For these two edges $e_1=(v_m^i,v_n^j)$ and $e_2=(v_q^j,v_p^i)$, we must have $p\leq m \land n\leq q$.
\end{lemma}

\begin{proof}
Given an arbitrary cycle on the two-agent TPG, let $v^i_p$ be the earliest vertex of agent $i$ in the cycle (i.e., the vertex of agent $i$ with the smallest $p$). The edge in the cycle that points to $v^i_p$ must be from agent $j$. Denote this edge as $(v_q^j,v_p^i)$. Similarly, let $v^j_n$ be the earliest vertex of agent $j$ in the cycle. The edge in the cycle that points to $v^j_n$ must be from agent $i$. Denote this edge as $(v_m^i,v_n^j)$. Since $v^i_p$ and $v^j_n$ are the earliest vertices of their respective agents, we know $p \le m$ and $n \le q$. Therefore, we find a cycle $v_q^j\rightarrow v_p^i \rightarrow \cdots \rightarrow v_m^i \rightarrow v_n^j \rightarrow \cdots \rightarrow v_q^j$, where the first ``$\cdots$'' denotes a sequence of Type-1 edges for agent $i$ and the second ``$\cdots$'' denotes a sequence of Type-1 edges for agent $j$. This cycle meets the properties described by the lemma. 
\end{proof}

Next, we describe the Function \textsc{FindGroupableEdges} to identify all switchable edges that can be grouped with a specific edge $e$ in \Cref{algo: FGE}. We provide an intuitive example in \Cref{fig: grouping_procedure} with detailed reasoning in its caption for the parallel pattern mentioned in the \Cref{fig: grouping}a.

To find all groupable edges, we must reason the two possible directions of $e$ (\Cref{line: propagate original direction} for the reverse direction and \Cref{line: propagate reverse direction start}-\Cref{line: propagate reverse direction end} for the original direction). If $e$ is reversed, we call Function \textsc{Propagate} to find the set of edges that must follow the same direction as $Re(e)$. If $e$ is fixed, because of the symmetry of the reasoning procedure, we need to reverse all the switchable edges first, then call \textsc{Propagate}, and finally reverse the result back (\Cref{line: propagate reverse direction start}-\Cref{line: propagate reverse direction end}). This way, we can find another set of edges that must follow the same direction as $e$ if $e$ is fixed. The intersection of these two edge sets contains all the edges that must always select the same direction as $e$, namely all the groupable edges (\Cref{line: intersection}).

Our core Function \textsc{Propagate} tries to figure out all the edges in $E$ that must also be reversed if we reverse the edge $e$. All the edges that must be reversed are recorded in $S$, the newly found ones are kept in $C$, and the remaining edges are maintained in $E$ (\Cref{line: init edge set}). $S$ and $C$ start with only the edge $e$ while $E$ is initialized to all the edges but $e$. \Cref{line: reverse newly found edges} reverses the direction of all newly found edges in $C$ and \Cref{line: call find cycle edges} calls Function \textsc{FindCycleEdges} to find the edges in $E$ that would form cycles with the reversed edges in $C_r$. Indeed, the returned edges from \textsc{FindCycleEdges} become the newly found edges that must be reversed as $e$. Otherwise, our TPG would be cyclic. The procedure of cycle detection is an efficient implementation based on the \Cref{lemma: two-agent cycle detection}.

If we cannot find any edges leading to cycles (\Cref{line: empty cycle edges}), it implies that we find a settlement of all switchable edges leading to a TPG with no cycle. This TPG is a certificate that all the remaining edges can select the opposite of $e$'s settled direction and, thus, are not groupable with $e$. Then, we can conclude that all the groupable edges must be a subset of $S_1$ in \Cref{line: propagate original direction}, similarly, of $S_2$ in \Cref{line: propagate reverse direction end}, and consequently, of their intersection. On the other hand, according to the result of Function \textsc{Propagate}, the edges in $S_1$ must select the same direction as $e$ if $e$ is reversed, and the edges in $S_2$ must select the same direction as $e$ if $e$ is fixed. Therefore, the edges in the intersection must follow the same direction as $e$. That is, we find the exact maximal edge group containing edge $e$ by \Cref{algo: FGE}, according to \Cref{definition: edge group}.

Regarding the time complexity of the edge grouping, in the worst case, the while loop in \Cref{algo: EGF} needs to iterate over all edges and has a complexity of $O(|\mathcal{S}_{i,j}|)$. For each call to Function \textsc{Propagate}, we at most need to check all pairs of edges, and it has a complexity of $O(|\mathcal{S}_{i,j}|^2)$. Therefore, the overall worst-case complexity is $O(|\mathcal{S}_{i,j}|^3)$. 

\subsection{Prioritized Branching}
\label{subsection: branching}

In \Cref{line: prioritizing branching}-\Cref{line: conflicting edge not found} of \Cref{algo: GSES}, we need to select a conflicting edge to branch. However, different prioritization in conflicting edge selection may influence the search speed. We experiment with the following four strategies:
\begin{enumerate}
    \item \textbf{Random}: randomly select a conflicting edge.
    \item \textbf{Earliest-First}: select the first conflicting edge $e=(v_p,v_q)$ with the smallest ordered-pair $(L(v_q),L(v_p))$.
    \item \textbf{Agent-First}: select the first conflicting edge with the smallest agent index. This is the strategy used by GSES.
    \item \textbf{Smallest-Edge-Slack-First}: select the first conflicting edge with the smallest edge slack $Sl(e)$.
\end{enumerate}

In \Cref{section: experiments}, we find that \textbf{Smallest-Edge-Slack-First} performs the best empirically. The intuition behind this design is that $Sl(e)$ reflects edge $e$'s degree of conflict with the current STPG, and we want to address the most conflicting one first because it may influence the overall cost the most. 

\subsection{Incremental Implementation}
\label{subsection: incremental}

In GSES, the computation of the longest path lengths  (\Cref{line: longest path lengths} of \Cref{algo: GSES}) takes the most time in a search node construction, as illustrated in Figure \ref{fig:time_profiling}. But each time we build a child node, we only add one edge (group) to the reduced TPG of the parent node. Therefore, We directly apply an existing algorithm for computing the longest path lengths incrementally in a directed acyclic graph~\citep{katriel2005maintaining} to speed up the search.

\section{Experiments}
\label{section: experiments}
Following the setting in the baseline GSES paper \citep{FengICAPS24}, we evaluate algorithms on four maps from the MAPF benchmark \citep{SternSoCS19}, with $5$ different numbers of agents per map. For each map and each number of agents, we generate $25$ different instances with start and goal locations evenly distributed. We then run 1-robust PBS on each instance to obtain the initial MAPF plans. We add the 1-robust requirement~\citep{atzmon2018robust} to PBS~\citep{ma2019searching} because the original PBS does not resolve the following conflicts. 
Each solved instance is tested with $6$ different delay scenarios. We obtain a scenario by simulating the initial MAPF plan until some delays happen. An agent will be independently delayed for 10-20 steps with a constant probability $p\in\{0.002,0.01,0.03\}$ at each step. 
We run $8$ times for each scenario and set a time limit of $16$ seconds for each run. Notably, since all the algorithms we compare with are optimal, we only focus on their search time. Due to the space limit, we only report the results of $0.01$ delay probability in the main text with others reported in Appendix A.4. The conclusions are consistent across different probabilities.

All the algorithms in the experiments are implemented in C++.\footnote{MILP also uses a Python wrapper, so we count only its C++ solver's time for a fair comparison.} All the experiments are conducted on a server with an Intel(R) Xeon(R) Platinum 8352V CPU 
and 120 GB RAM. Additional experimental details are covered in Appendix A.4. The code and benchmark are publicly available.\footnote{https://github.com/DiligentPanda/STPG.git}

\begin{figure}[tb]
    \setlength{\abovecaptionskip}{4pt}  
    \setlength{\belowcaptionskip}{8pt}  

    \centering
    \begin{subfigure}[b]{0.45\textwidth}
      \centering
      \includegraphics[width=1\textwidth]{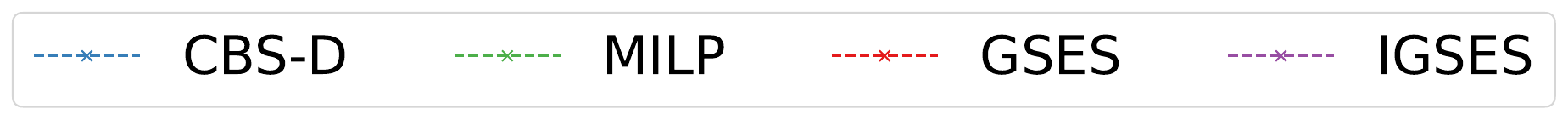}
    \end{subfigure}%
    \hfill

    \begin{subfigure}[b]{0.225\textwidth}
      \centering
      \includegraphics[width=1\textwidth]{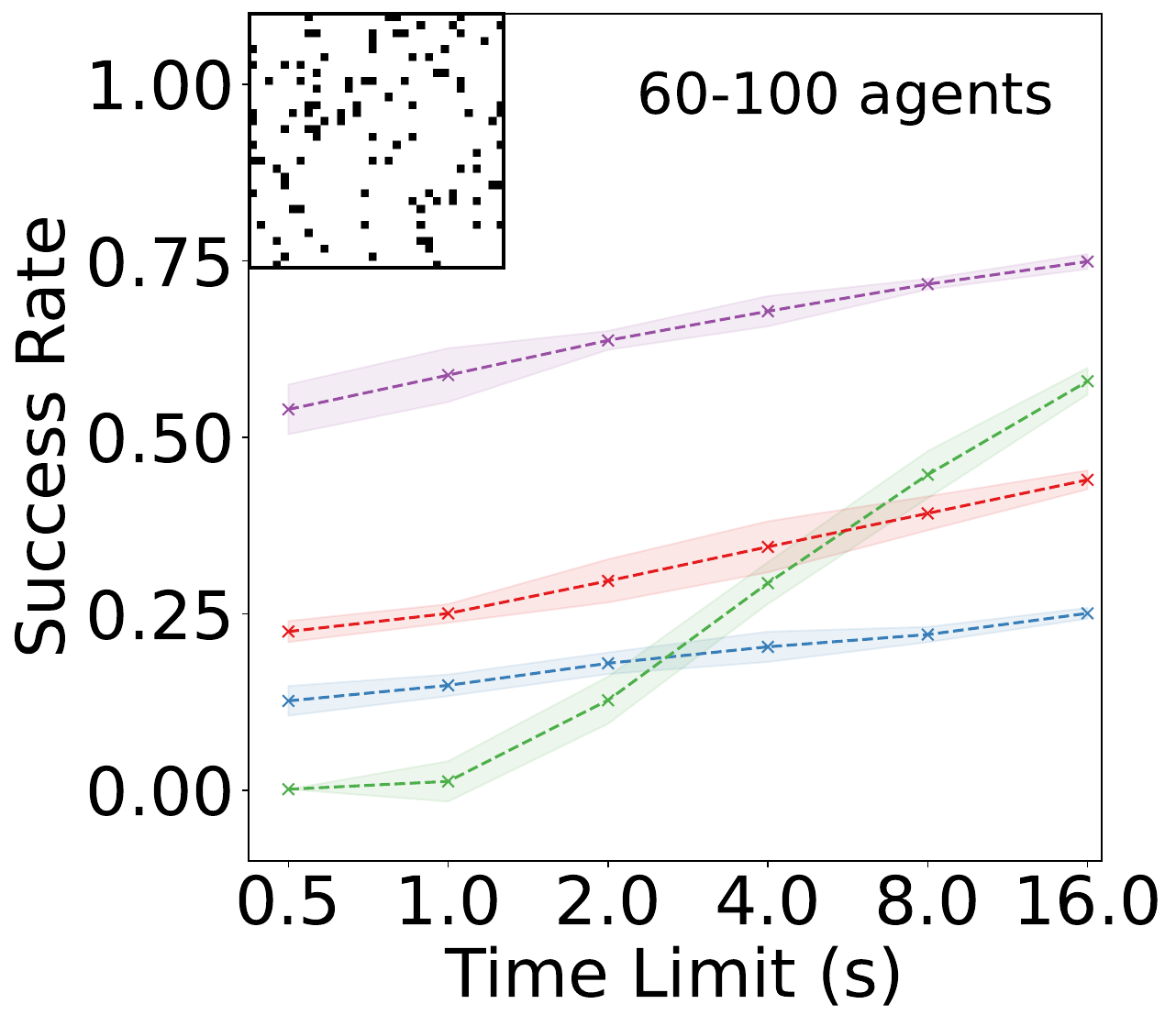}
      \caption{Random-32-32-10}
    \end{subfigure}%
    \hfill
    \begin{subfigure}[b]{0.225\textwidth}
      \centering
      \includegraphics[width=1\textwidth]{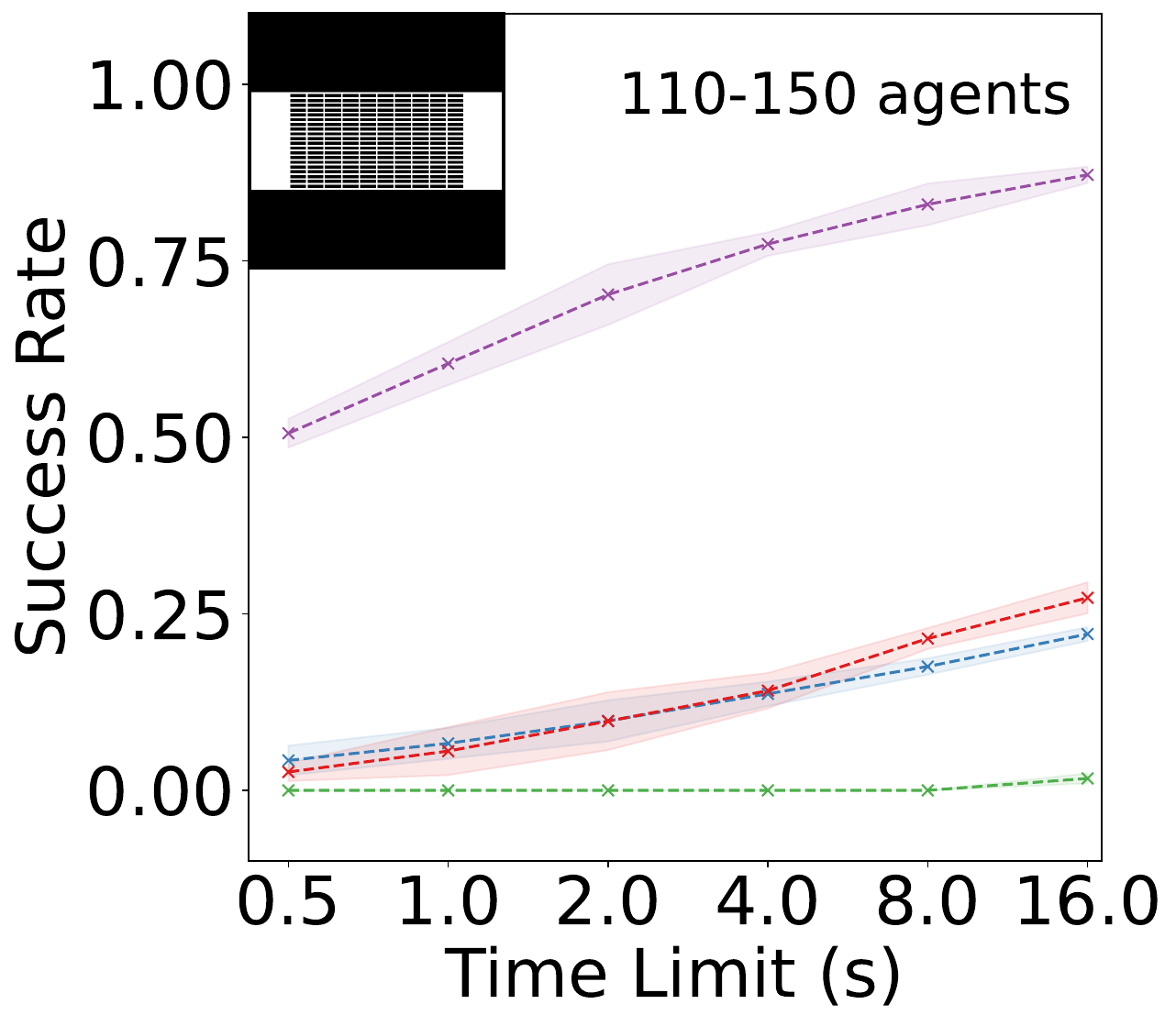}
      \caption{Warehouse-10-20-10-2-1}
    \end{subfigure}%
    \hfill

    \setlength{\belowcaptionskip}{4pt}  
    
    \begin{subfigure}[b]{0.225\textwidth}
      \centering
      \includegraphics[width=1\textwidth]{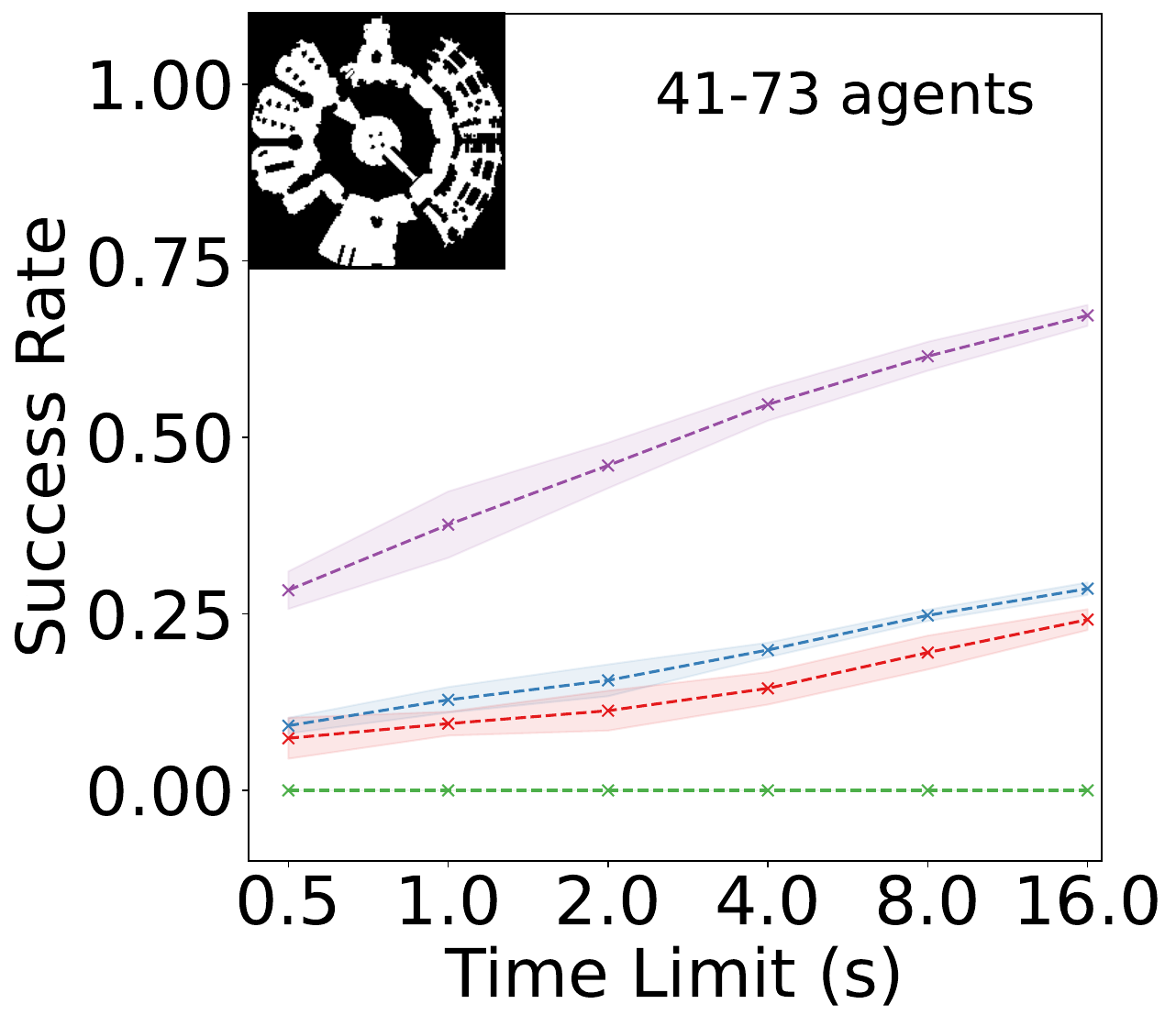}
      \caption{Lak303d}
    \end{subfigure}%
    \hfill
    \begin{subfigure}[b]{0.225\textwidth}
      \centering
      \includegraphics[width=1\textwidth]{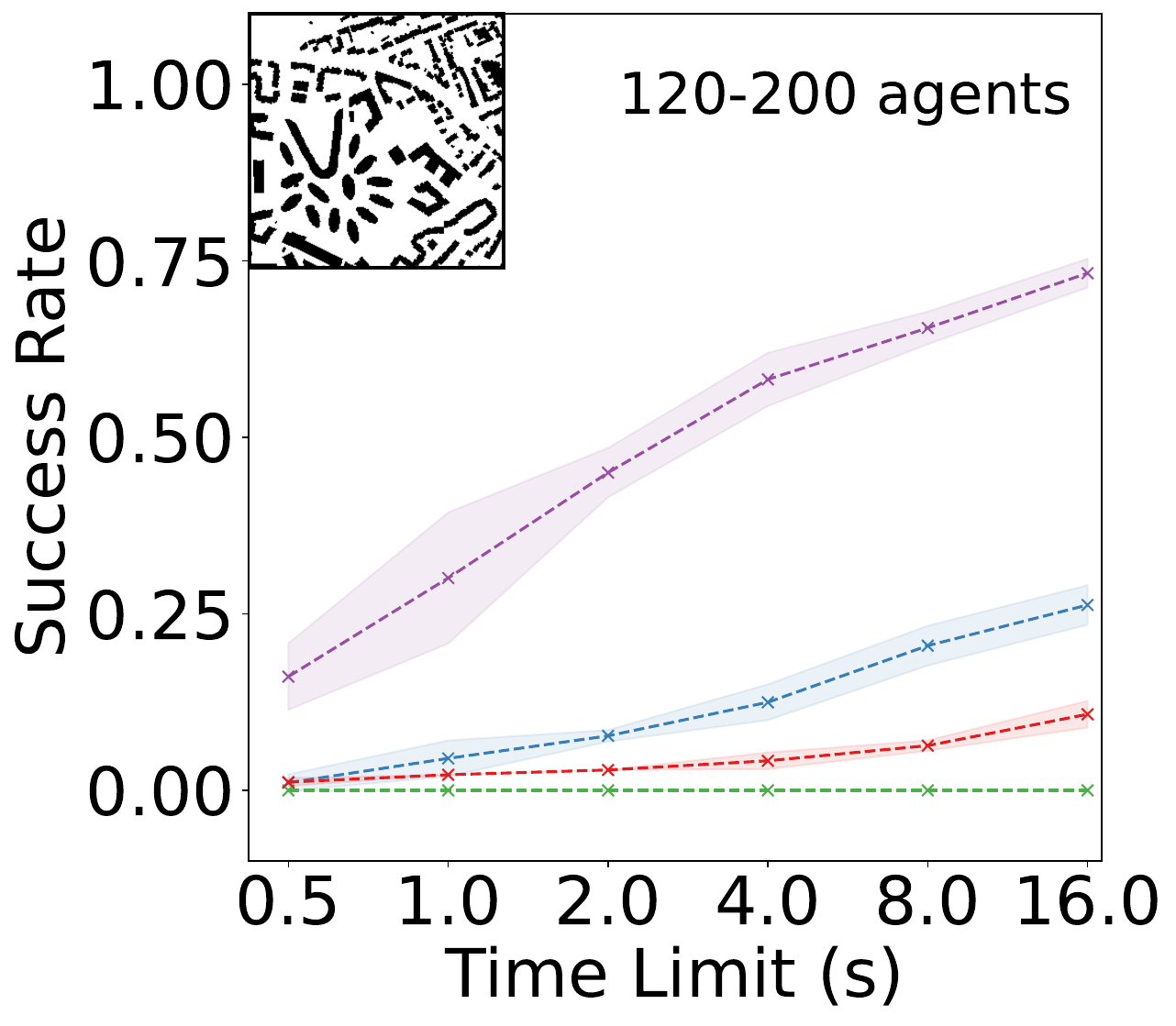}
      \caption{Paris\_1\_256}
    \end{subfigure}%
    \hfill
    \caption{Success rates of IGSES and other baselines on four maps. An instance is considered successfully solved if an optimal solution is returned within the time limit.
    The shading areas indicate the standard deviations of different runs. 
    They are multiplied by 10 for illustration. In each figure, the top-left corner shows the corresponding map, and the top-right corner shows the range of agent numbers.}
    \label{fig:succ_rates}
\end{figure}

\begin{table*}[!h]
    \caption{Incremental analysis on instances solved by all the settings. Rows 1,2,3 compare the effects of different grouping methods. GSES has no grouping. SG: simple grouping. FG: full grouping. Rows 3,4,5,6 compare the effects of different branching orders. Row 3 applies the default Agent-First branching in GSES. RB: Random branching. EB: Earliest-First branching. SB: Smallest-Edge-Slack-First branching. Rows 6,7 compare the effect of the stronger heuristics. SH: stronger heuristics. Rows 7,8 compare the effect of the incremental implementation. INC: incremental implementation. The last row is also our IGSES.}
    \label{tab:incremental}
    \centering
    \resizebox{0.98\textwidth}{!}{
    \begin{tabular}{c|c|rrr|rrr|rrr|rrr}
        \toprule
        \toprule
        \multirow{3}{*}{\makecell{Row}} & \multirow{3}{*}{Setting} & \multicolumn{3}{c|}{Random-32-32-10} &
        \multicolumn{3}{c|}{Warehouse-10-20-10-2-1} & 
        \multicolumn{3}{c|}{Lak303d} & 
        \multicolumn{3}{c}{Paris\_1\_256}\\
        & & \makecell{search\\time (s)} & \makecell{\#expanded\\nodes} & \makecell{\#edge\\ groups} & \makecell{search\\time (s)} & \makecell{\#expanded\\nodes} & \makecell{\#edge\\ groups} & 
        \makecell{search\\time (s)} & \makecell{\#expanded\\nodes} & \makecell{\#edge\\ groups} & 
        \makecell{search\\time (s)} & \makecell{\#expanded\\nodes} & \makecell{\#edge\\ groups}\\
        \midrule
    
1 &GSES & 1.421 & 3051.5 & 1637.9 & 3.805 & 948.3 & 15441.6 & 2.608 & 550.8 & 30320.6 & 5.221 & 253.1 & 40214.5\\
2 &+SG & 1.082 & 2025.7 & 1176.4 & 1.783 & 315.3 & 9134.2 & 1.242 & 185.7 & 19012.7 & 3.131 & 115.2 & 27065.5\\
3 &+FG & 0.945 & 1597.9 & 732.5 & 1.519 & 241.4 & 2367.0 & 1.021 & 130.8 & 5601.8 & 2.770 & 88.1 & 10706.9\\
4 &+FG +RB & 1.555 & 2495.2 & 732.5 & 2.685 & 447.4 & 2367.0 & 2.644 & 324.2 & 5601.8 & 3.553 & 114.5 & 10706.9\\
5 &+FG +EB & 1.185 & 2036.1 & 732.5 & 2.354 & 381.6 & 2367.0 & 1.481 & 172.9 & 5601.8 & 3.118 & 102.6 & 10706.9\\
6 &+FG +SB & 0.915 & 1480.5 & 732.5 & 1.472 & 229.5 & 2367.0 & 1.027 & 129.0 & 5601.8 & 2.472 & 78.7 & 10706.9\\
7 &+FG +SB +SH & 0.169 & 84.6 & 732.5 & 0.832 & 34.6 & 2367.0 & 0.871 & 33.1 & 5601.8 & 1.561 & 15.4 & 10706.9\\
8 &+FG +SB +SH +INC & 0.043 & 84.6 & 732.5 & 0.120 & 34.6 & 2367.0 & 0.221 & 33.1 & 5601.8 & 0.321 & 15.4 & 10706.9\\

        \bottomrule
        \bottomrule
    \end{tabular}
    }
\end{table*}

\subsection{Comparison with Other Algorithms}
We compare our method IGSES with the baseline GSES algorithm and other two optimal algorithms, MILP \citep{berndt2023receding} and CBS-D\footnote{We adapt the original CBS-D for our MAPF definition, as we forbid not only edge conflicts but all following conflicts.} \citep{kottinger2024introducing}. The success rates on four maps with increasing sizes are shown in \Cref{fig:succ_rates}. 
Compared to GSES, IGSES at least doubles the success rates in most cases. Due to the inefficiency of a general branch-and-bound solver, MILP can only solve instances on the small Random map. An interesting observation is that CBS-D is worse than GSES on small maps but generally becomes better than GSES as the map size increases.
However, directly applying existing MAPF algorithms without tailored adaptations like CBS-D shows much worse performance than our IGSES, which better exploits the problem's structure. In Appendix A.4, we also plot the mean search time for different maps and agent numbers to support our conclusions.

\begin{figure}[!tb]
    \begin{subfigure}[b]{0.22\textwidth}
      \centering
      \includegraphics[width=1\textwidth]{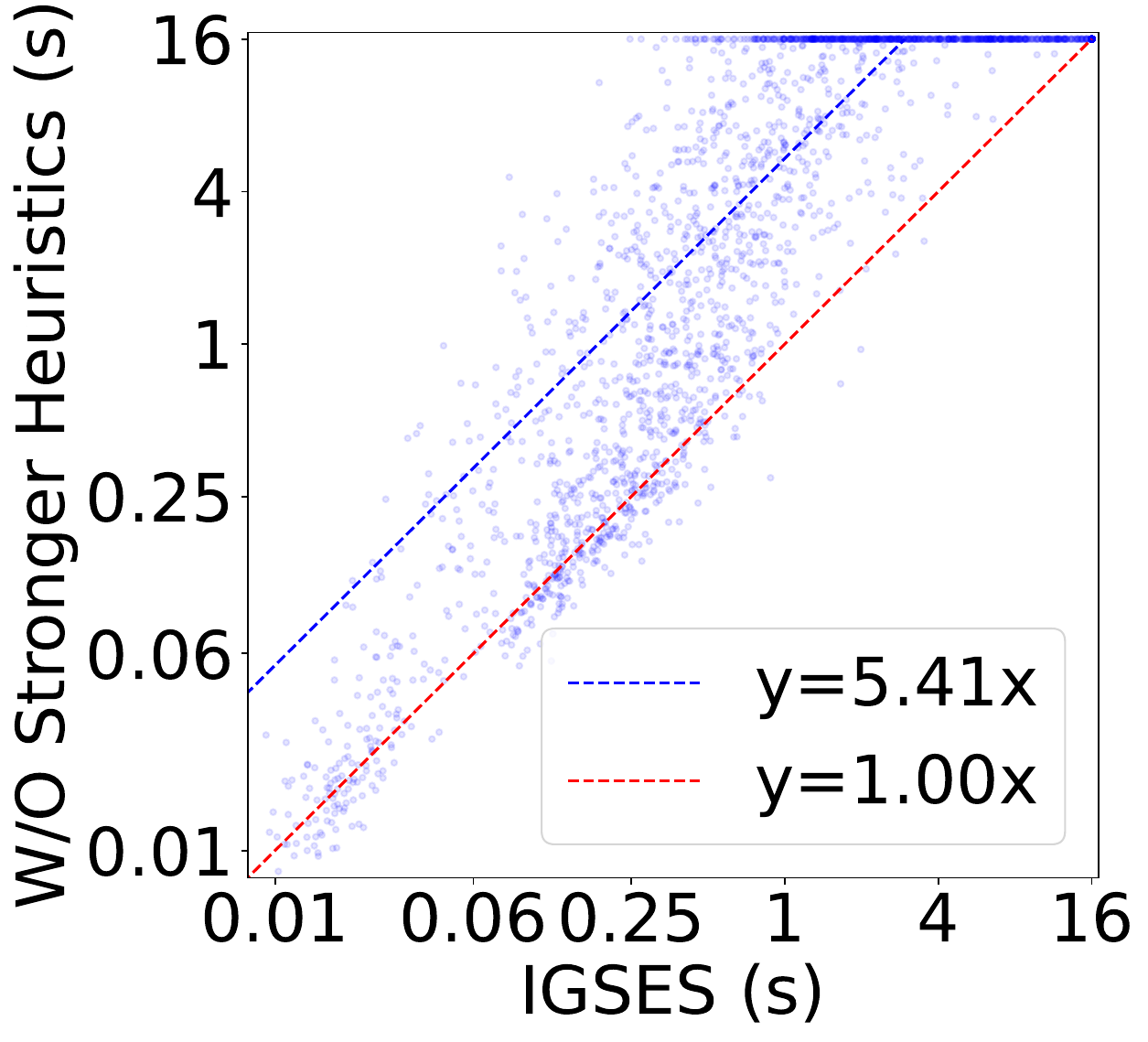}
      \caption{Heuristics}
    \end{subfigure}%
    \hfill
    \begin{subfigure}[b]{0.22\textwidth}
      \centering
      \includegraphics[width=1\textwidth]{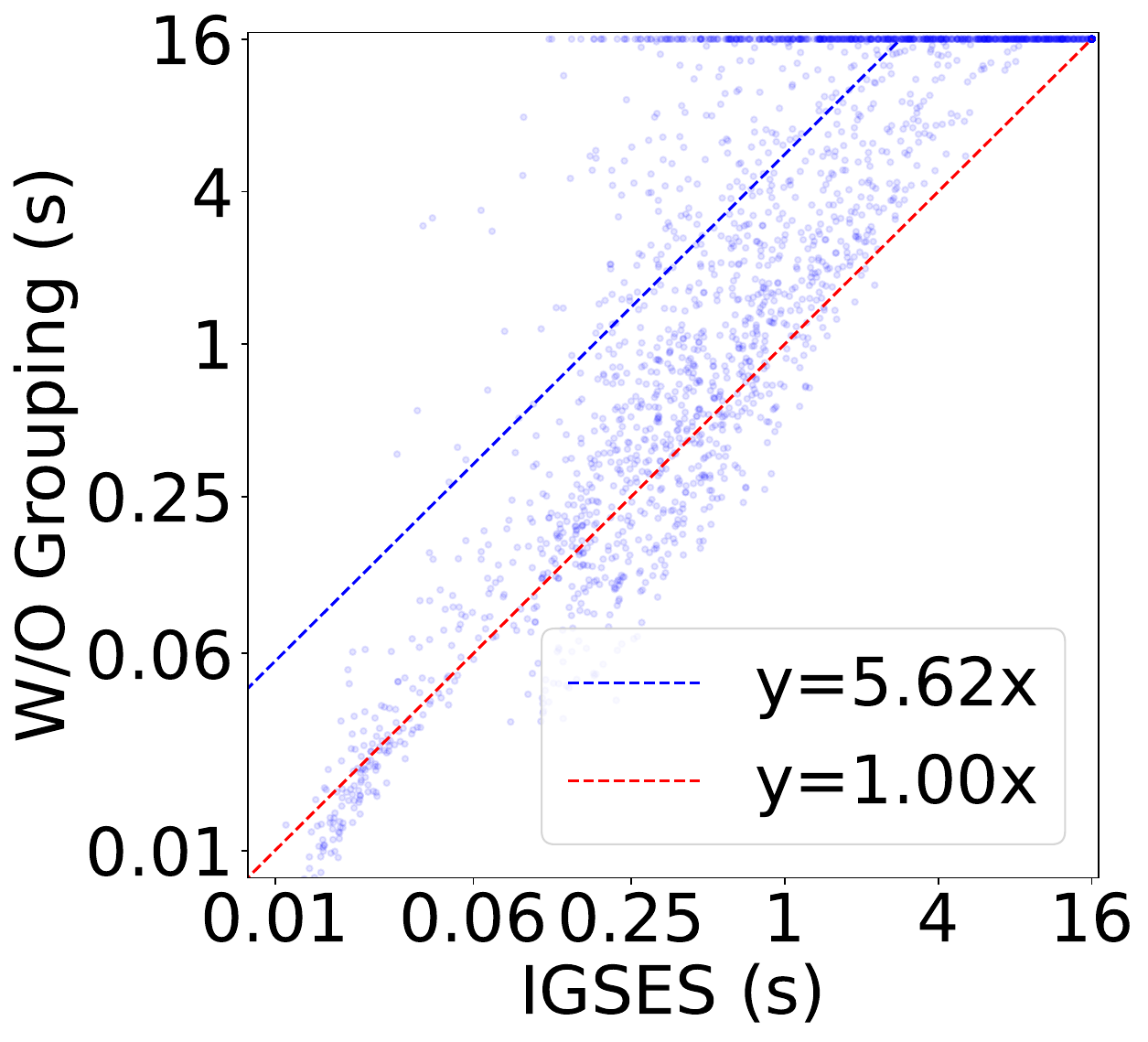}
      \caption{Grouping}
    \end{subfigure}%
    \hfill    
    
    \begin{subfigure}[b]{0.22\textwidth}
      \centering
      \includegraphics[width=1\textwidth]{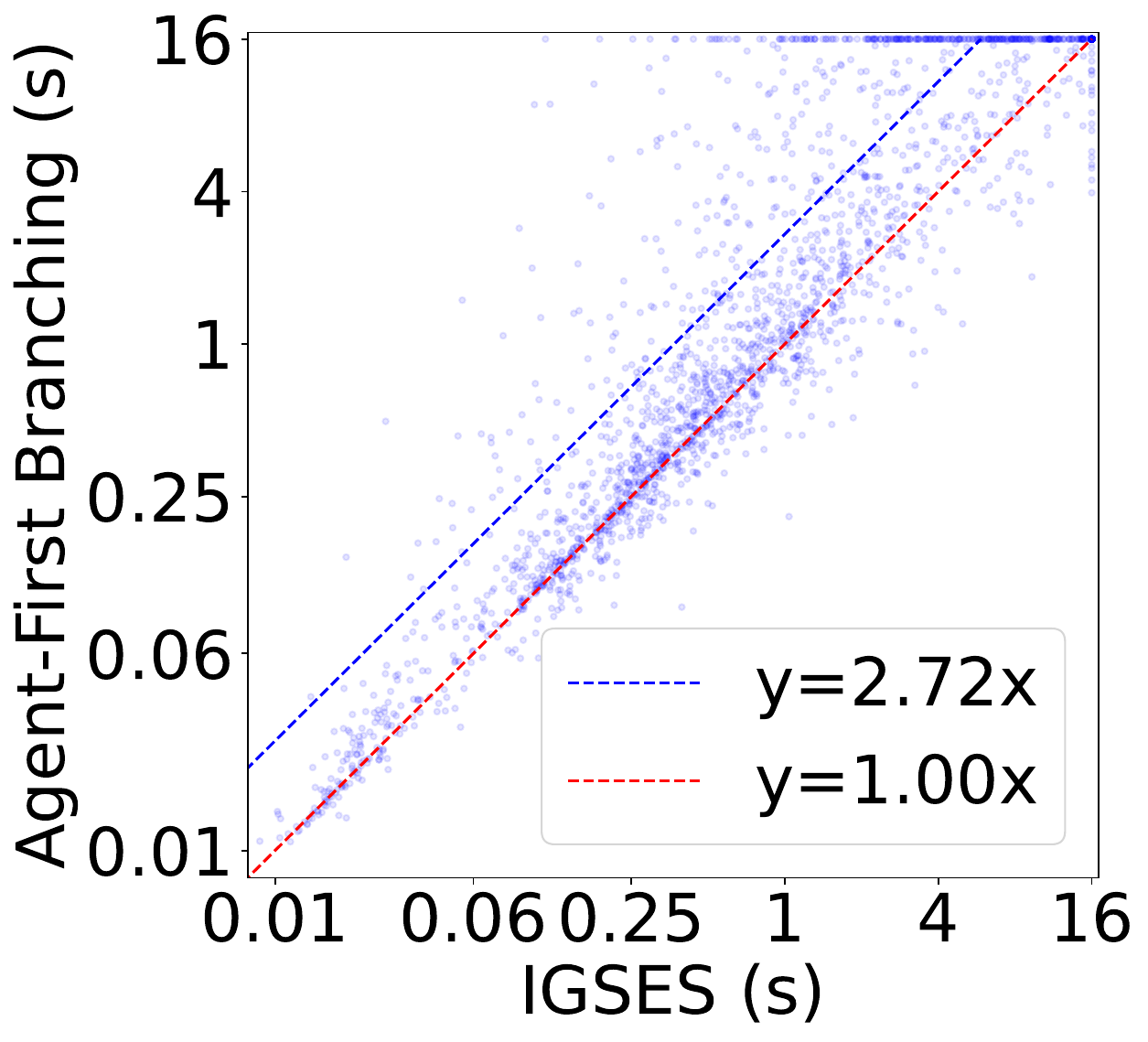}
      \caption{Branching}
    \end{subfigure}%
    \hfill
    \begin{subfigure}[b]{0.22\textwidth}
      \centering
      \includegraphics[width=1\textwidth]{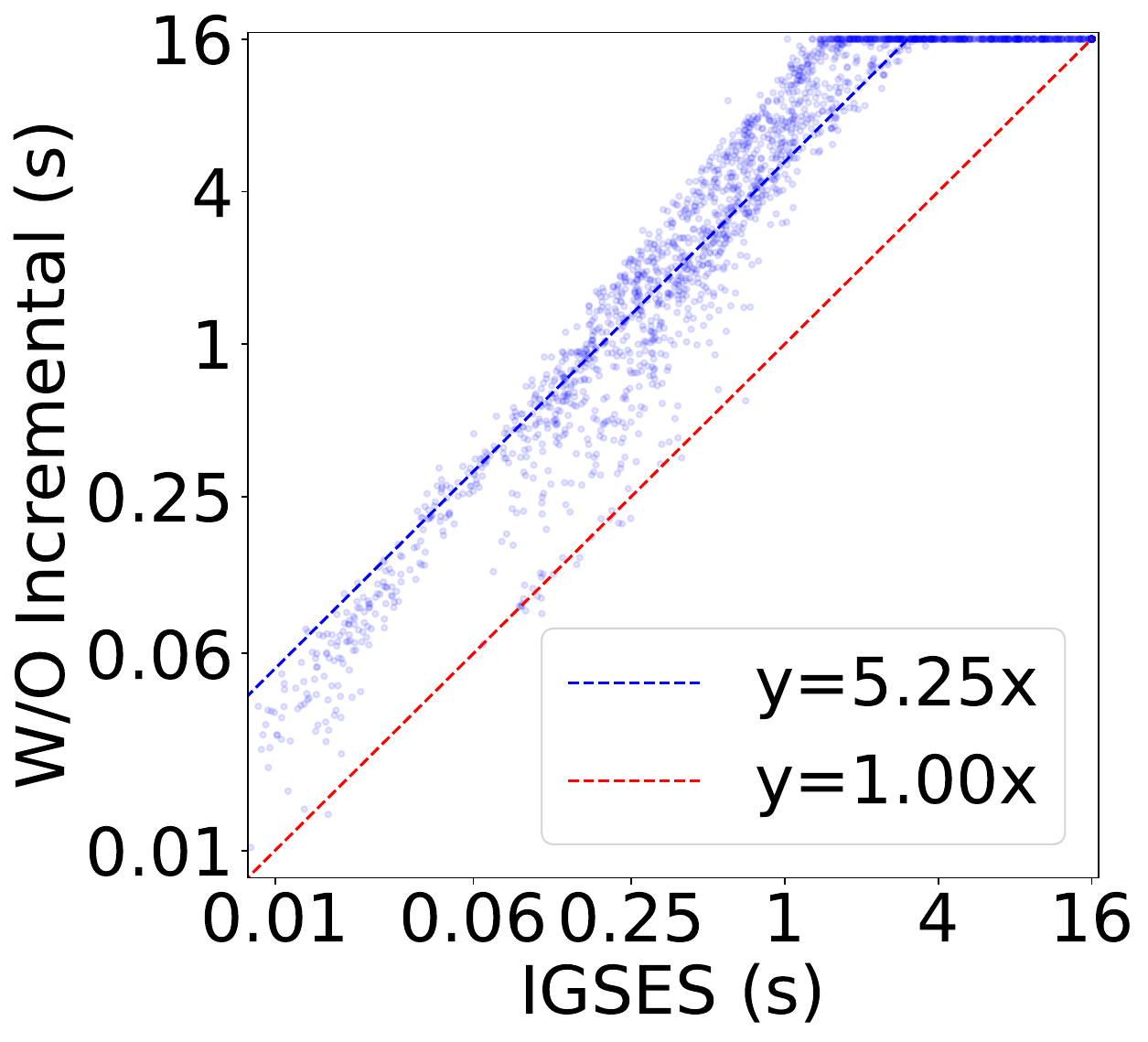}
      \caption{Incremental}
    \end{subfigure}%
    \hfill
    \caption{Ablation on four speedup techniques on all instances. In each figure, we compare IGSES to the ablated setting that replaces one of its techniques by the GSES's choice. Each point in the figure represents an instance with its $x,y$ coordinates being the search time of the two settings. The search time is set to 16 seconds if an instance is not solved. The blue line is fitted on the instances solved by at least one setting. Its slope indicates the average speedup.}
    \label{fig: ablation}
\end{figure}

\subsection{Ablation Studies}
\label{subsection: ablation studies}

First, we provide an incremental analysis in \Cref{tab:incremental}. We begin with GSES in the first row, gradually compare different choices in each technique, and add the most effective one. Finally, we obtain IGSES in the last row, which achieves a 10- to 30-fold speedup and over a 90\% reduction in node expansion on instances solved by all the settings, compared to GSES. Notably, comparing rows 1,2,3, we can find that our full grouping has the smallest number of edge groups, thus empirically less branching and node expansion. Comparing rows 7,8, we can find that the incremental implementation does not affect the node expansion but significantly reduces the computation time as expected. 

Further, we do an ablation study on the four techniques with all the instances, including unsolved ones, through pairwise comparison in \Cref{fig: ablation}. All the techniques are critical to IGSES, as they all illustrate an obvious speedup.

\begin{figure}[tb]
\centering
\includegraphics[width=1.0\columnwidth]{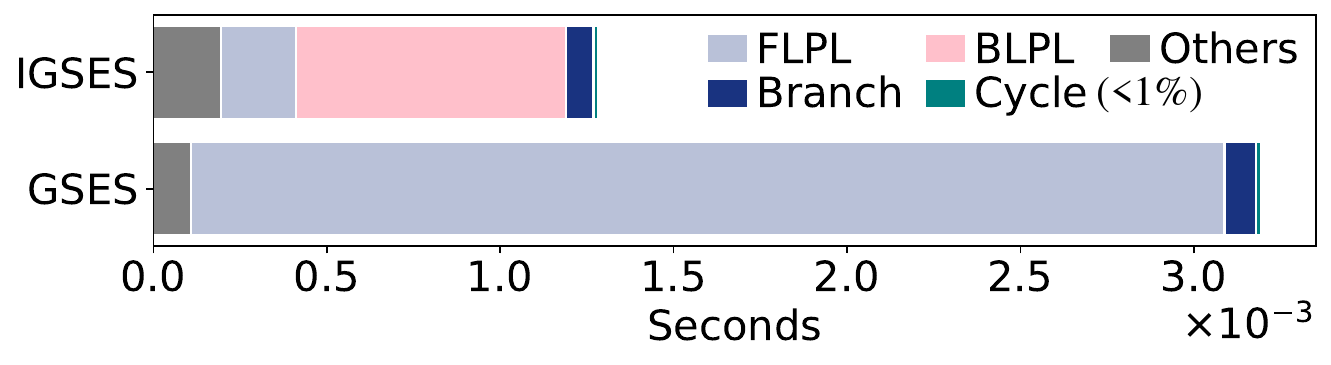} 
\caption{Runtime profile for each search node in the warehouse experiments. FLPL: forward longest path lengths. BLPL: backward longest path lengths. Branch: branch prioritization. Cycle: cycle detection. Others: copy and free data structures, termination check, and others.}
\label{fig:time_profiling}
\end{figure}

Finally, we profile the average runtime on each node for GSES and IGSES with the warehouse experiment data in \Cref{fig:time_profiling}. The computation of the longest path lengths takes the most runtime in both cases. IGSES's runtime on the forward longest path lengths is largely reduced mainly because of the incremental implementation. The backward longest path lengths take more time to compute than the forward, but IGSES's overall runtime is much smaller. The edge grouping time is not included here. Edge grouping can be computed only once before executing the MAPF plan and used whenever delays occur. It takes $0.037$ seconds on average, which is negligible compared to the MAPF's planning time.

\section{Conclusion and Future Work}
In this paper, we study the STPG optimization problem. We analyze the weakness of the optimal GSES algorithm and propose four speedup techniques. Their effectiveness is validated by both theoretical proof and experimental data. To scale to larger instances, a potential future direction is to devise sub-optimal algorithms for the STPG optimization problem to trade-off between solution quality and speed. 

\section*{Acknowledgments}
The research was supported by the National Science Foundation (NSF) under grant number \#2328671 and a gift from Amazon. The views and conclusions contained in this document are those of the authors and should not be interpreted as representing the official policies, either expressed or implied, of the sponsoring organizations, agencies, or the U.S. government.

\bibliographystyle{aaai25}
\bibliography{aaai25}

\clearpage

\appendix
\section{Appendix}
\label{section: appendix}

\subsection{Computing the Longest Path Lengths}
\label{subsection: appendix_longest_path_lengths}
This section describes the algorithms for computing the forward longest path lengths (FLPL) and backward longest path lengths (BLPL) on the reduced TPG. They are simple, so we mainly present the pseudocode.

We first describe the non-incremental version of these two algorithms in \Cref{algo: FLPL} and \Cref{algo: BLPL}. The input graph $\mathcal{G}=(\mathcal{V},\mathcal{E})$ is the reduced TPG of the current STPG, where $\mathcal{E}$ contains both Type-1 edges and non-switchable Type-2 edges. The topological ordering, $Ordered(\mathcal{V})$, obtained in the computation of FLPL will also be input to the calculation of BLPL. The returned $L(v)$ and $L(v,g)$ represent the functions of FLPL and BLPL, respectively.

When computing the FLPL in \Cref{algo: FLPL}, we update $L(v)$ by the predecessors of $v$ (\Cref{line: update FL}) in the topological ordering. Similarly, when computing the BLPL in \Cref{algo: BLPL}, we update $L(v,g^i)$ by the successors of $v$ that are connected to $g^i$ (\Cref{line: update BL}) in the reverse topological ordering. 

\begin{algorithm}[htb]
\caption{Forward Longest Path Lengths}
\label{algo: FLPL}
\begin{algorithmic}[1]
\Function{FLPL}{$\mathcal{G}=(\mathcal{V},\mathcal{E})$}
    \State Obtain a topological ordering of $G$, denoted as $Ordered(\mathcal{V})$
    \State $L(v)\leftarrow 0$ for $\forall v\in \mathcal{V}$.
    \For{$v\in Ordered(\mathcal{V})$}
        \State Get the predecessors of $v$, denoted as $Pred(v)$
        \State $L(v)\leftarrow \max{\{L(u)+1 \mid u\in Pred(v)\}}$
            \label{line: update FL}
    \EndFor
    \State \Return $L(v)$,$Ordered(\mathcal{V})$
\EndFunction
\end{algorithmic}
\end{algorithm}

\begin{algorithm}[htb]
\caption{Backward Longest Path Lengths}
\label{algo: BLPL}
\begin{algorithmic}[1]
\Function{BLPL}{$\mathcal{G}=(\mathcal{V},\mathcal{E})$,$Ordered(\mathcal{V})$}
    \State Obtain the reverse ordering of $Ordered(\mathcal{V})$, denoted as $ROrdered(\mathcal{V})$
    \State $L(v,g^i)\leftarrow \infty$ for $\forall v\in \mathcal{V}, i\in I$
    \State $L(g^i,g^i)\leftarrow 0$ for $\forall i\in I$
    \For{$v\in ROrdred(\mathcal{V})$}
        \State Get the successors of $v$, denoted as $Succ(v)$
        \For{$i\in I$}
            \State $L(v,g^i)\leftarrow \max{}\{L(u,g^i)+1 \mid L(u,g^i)\neq \infty,u\in S(v) \}$
                \label{line: update BL}
        \EndFor
    \EndFor
    \State \Return $L(v,g)$
\EndFunction
\end{algorithmic}
\end{algorithm}

Here, we present the incremental versions of \Cref{algo: FLPL} and \Cref{algo: BLPL}, shown in \Cref{algo: FLPL_inc} and \Cref{algo: BLPL_inc} respectively, adapted from the paper \citep{katriel2005maintaining}.

For both algorithms, in addition to the graph $G$, we also input $NewEdges$, the set of new edges added to the graph, and $L(v)$, the old FLPL function, which specifies an update order. Notably, $L$ actually specifies a topological order on the original graph $G$ if we sort all the vertices by $L(v)$ in the ascending order. This order or its reverse is maintained by a heap in both algorithms. The returned $L'(v)$ and $L'(v,g)$ represent the functions of new FLPL and BLPL, respectively.  

When computing the FLPL in \Cref{algo: FLPL_inc}, we update the $L'(v)$ by the predecessors of $v$ (\Cref{line: update FL'}). If the length changes, we push into the heap all the successors of $v$ that have not been visited for future updating (\Cref{line: push into FL' heap}). Similarly, when computing the BLPL in \Cref{algo: BLPL_inc}, we update the $L'(v,g_i)$ by the successors of $v$ that are connected to $g^i$ (\Cref{line: update BL'}). If any length changes, we push into the heap all the predecessors of $v$ that have not been visited for future updating (\Cref{line: push into BL' heap}). 

Since the incremental implementation only updates the longest path lengths for a small portion of vertices affected by the newly added edge (group) in the reduced TPG, the computational complexity is significantly reduced.

\begin{algorithm}[htb]
\caption{Forward Longest Path Lengths (Incremental)}
\label{algo: FLPL_inc}
\begin{algorithmic}[1]
\Function{FLPL\_inc}{$\mathcal{G}=(\mathcal{V},\mathcal{E})$, $L(v)$, $NewEdges$}
    \State Initialize a min-heap $H$ 
    \State $L'(v)\leftarrow L(v),\forall v \in \mathcal{V}$
    \State $visited(v)\leftarrow False,\forall v\in \mathcal{V}$
    \For {$e=(v_1,v_2)\in NewEdges$}
        \State Push $(L(v_2),v_2)$ into $H$
        \State $visited(v_2)\leftarrow True$
    \EndFor
    \While{$H$ not empty}
        \State Pop $(L(v),v)$ from $H$
        \State Get the predecessors of $v$, denoted as $Pred(v)$
        \State $L'(v)\leftarrow \max{\{L'(u)+1 \mid u\in Pred(v)\}}$
            \label{line: update FL'}
        \State $visited(v)\leftarrow True$
        \If {$L'(v)\neq L(v)$}
            \State Get the successors of $v$, denoted as $Succ(v)$
            \For{$u\in Succ(v)$}
                \If {not $visited(u)$}
                    \State Push $(L(u),u)$ into $H$
                        \label{line: push into FL' heap}
                \EndIf
            \EndFor
        \EndIf
    \EndWhile
    \State \Return $L'(v)$
\EndFunction
\end{algorithmic}
\end{algorithm}

\begin{algorithm}[htb]
\caption{Backward Longest Path Lengths (Incremental)}
\label{algo: BLPL_inc}
\begin{algorithmic}[1]
\Function{BLPL\_inc}{$\mathcal{G}=(\mathcal{V},\mathcal{E})$, $L(v)$, $NewEdges$}
    \State Initialize a max-heap $H$ 
    \State $L'(v,g_i)\leftarrow L(v,g_i),\forall v \in \mathcal{V},i\in I$
    \State $visited(v)\leftarrow False,\forall v\in \mathcal{V}$
    \For {$e=(v_1,v_2)\in NewEdges$}
        \State Push $(L(v_1),v_1)$ into $H$
        \State $visited(v_1)\leftarrow True$
    \EndFor
    \While{$H$ not empty}
        \State Pop $(L(v),v)$ from $H$
        \State Get the successors of $v$, denoted as $Succ(v)$
        \For{$i\in I$}
            \State $L'(v,g^i)\leftarrow \max{}\{L'(u,g^i)+1 \mid L'(u,g^i)\neq \infty,u\in S(v) \}$
                \label{line: update BL'}
        \EndFor
        \State $visited(v)\leftarrow True$
        \If {$\exists i\in I, L'(v,g_i)\neq L(v,g_i)$}
            \State Get the predecessors of $v$, denoted as $Pred(v)$
            \For{$u\in Pred(v)$}
                \If {not $visited(u)$}
                    \State Push $(L(u),u)$ into $H$
                        \label{line: push into BL' heap}
                \EndIf
            \EndFor
        \EndIf
    \EndWhile
    \State \Return $L'(v,g)$
\EndFunction
\end{algorithmic}
\end{algorithm}

\subsection{The Termination Condition of GSES}

This section proves the termination condition of GSES.

\begin{proposition}
\label{prop: termination_of_GSES}
If there is no conflicting switchable edge found in \Cref{line: solution found end} of \Cref{algo: GSES}, GSES can be terminated by fixing all the switchable edges. It returns an acyclic TPG with the same execution cost as the current reduced TPG.
\end{proposition}

If all the switchable edges are not conflicting, then their edge slacks $Sl(e)\geq 0, \forall e\in \mathcal{S}$, where $\mathcal{S}$ is the set of switchable edges in the current STPG $\mathcal{G}^S$. To prove that fixing all the switchable edges will not introduce any cycle and extra cost, we prove \Cref{lemma: adding_edge}, which considers the case of adding only one edge to an acyclic TPG.

\begin{lemma}
\label{lemma: adding_edge}
    Given an acyclic TPG $\mathcal{G}=(\mathcal{V},\mathcal{E}_1,\mathcal{E}_2)$ with the forward longest path lengths as $L(v),v\in V$. If we add to $\mathcal{G}$ a new Type-2 edge $e'=(v_1,v_2)$ with its slack $Sl(e')=L(v_2)-L(v_1)-1\geq 0$, then the new graph $\mathcal{G}'=(\mathcal{V},\mathcal{E}_1,\mathcal{E}_2')$ remains acyclic, where $\mathcal{E}_2'=\mathcal{E}_2\bigcup\{e'\}$. Further, it has the same forward longest path lengths as $G$ at each vertex $v$. Namely, $L'(v)=L(v),v\in V$. 
\end{lemma}

\begin{proof}
By definition, $L(e)\geq 0$ for all $e\in \mathcal{E}_2$. Since $L(e')\geq 0$ as well, so we can state that $L(e)\geq 0$ for all $e\in \mathcal{E}_2'$. 

First, we prove that $\mathcal{G}'$ is acyclic by contradiction. We assume that there is a cycle $v_1 \to v_2 \to \cdots \to v_n \to v_1$. Since for $e_i=(v_{i},v_{i+1}), Sl(e)=L(v_{i+1})-L(v_{i})-1\geq 0$, we have $L(v_{i+1})>L(v_i)$. So, $L(v_n)>L(v_1)$. But there is also an edge from $v_n$ to $v_1$, so we can also get $L(v_1)>L(v_n)$, leading to the contradiction. Therefore, $\mathcal{G}'$ must be acyclic.

Since $\mathcal{G}'$ is also a direct acyclic graph, we assume we obtain a topological ordering of it in the \Cref{algo: FLPL}, as $Ordered(\mathcal{G}’)$. Since $\mathcal{G}'$ has one more edge than $\mathcal{G}$, $Ordered(\mathcal{G}')$ is also a topological ordering of $\mathcal{G}$ when this edge is removed. 

Now, if we run \Cref{algo: FLPL} for both $\mathcal{G}'$ and $\mathcal{G}$ with this ordering, every vertex $v$ ordered before $v_2$ in $Ordered(\mathcal{G}’)$ must have $L(v)=L'(v)$ because its predecessors and its predecessors' forward longest path lengths are not changed. 

Then, we validate $L'(v_2)=L(v_2)$. First, there exists some predecessor $u$ of $v_2$ such that $L(v_2)=L(u)+1$. Notably, $u$ and $v_1$ must be vertices ordered before $v_2$ in $Ordered(\mathcal{G}’)$ because they both have an edge pointing to $v_2$. Thus, $L'(u)=L(u)$ and $L'(v_1)=L(v_1)$. Since $Sl(e')=L(v_2)-L(v_1)-1\geq 0$, $L'(u)+1=L(u)+1=L(v_2)\geq L(v_1)+1=L'(v_1)+1$. Therefore, $L'(u)+1$ is still the maximum in \Cref{line: update FL} of \Cref{algo: FLPL}, even though a new edge $e'$ is added. So, $L'(v_2)=L'(u)+1=L(u)+1=L(v_2)$.

Finally, we can deduce every vertex $v$ ordered after $v_2$ in $Ordered(\mathcal{G}’)$ must also have $L(v)=L'(v)$ because its predecessors and its predecessors' forward longest path lengths remain the same. Therefore, $L'(v)=L(v),\forall v\in V$.
\end{proof}

\begin{corollary}
The execution cost of $\mathcal{G}'$ is the same as $\mathcal{G}$.
\end{corollary}

\begin{proof}
Based on \Cref{theorem: TPG cost}, the execution cost of a TPG is the sum of the longest path lengths at all agents' goals. Since the longest path lengths remain the same, the execution costs are also the same.
\end{proof}

Then, we can prove \Cref{prop: termination_of_GSES} easily by fixing switchable edges one by one. Each fixing adds a new edge to the acyclic TPG but does not introduce any cycle or extra cost.

\label{subsection: appendix_adding_edge}
\subsection{Greedy Matching Algorithm For EWMVC}
\label{subsection: appendix_ewmvc}

This section describes the greedy matching algorithm for the Edge-Weighted Minimum Vertex Cover (EWMVC) problem in \Cref{subsection: heuristics}. We have a weighted fully-connected undirected graph $G_{D}=(V_{D},E_{D}, W_{D})$, where each vertex $u_i\in V_{D}$ represents an agent, $E_D$ are edges and $W_D$ are edges' weights. Specifically, we set the weight of each edge $(u_i,u_j)\in E_{D}$ to be the pairwise cost increase, $\Delta cost(g^{i},g^{j})$ obtained from our stronger heuristic reasoning. Our target is to assign a cost increase $x_i$ to each vertex $u_i$ such that $x_i+x_j\geq \Delta cost(g^{i},g^{j})$ so that the overall cost increase $\sum_i x_i$ is minimized.

The pseudocode is illustrated in \Cref{algo: GM}. Briefly speaking, the algorithm always selects the max-weighted edge at each iteration, whose two endpoints have never been matched. Then, the weight is added to the overall weight, and the two endpoints are marked as matched. Since each iteration checks all the edge weights and we assume it is a fully connected graph, each iteration takes $O(|V_D|^2)$ time. There could be, at most, $|V_D|$ iterations. The worst-case complexity of this algorithm is $O(|V_D|^3)$.

\begin{algorithm}[htb]
\caption{Greedy Matching for EWMVC}
\label{algo: GM}
\begin{algorithmic}[1]
\Function{GreedyMatching}{$G_D=(V_D,E_D,W_D)$}
    \State $matched[i]\leftarrow False$ for $\forall i\in I$
    \State $sumW \leftarrow 0$
    \While{True}
        \State $maxW\leftarrow 0,p\leftarrow-1,q\leftarrow-1$
        \For {$(u_i,u_j)\in E_D$}
            \If{$matched[i] \lor matched[j]$}
                continue
            \EndIf
            \If{$W_D[i,j]>maxW$}
                \State $maxW\leftarrow W[i,j],p\leftarrow i,q\leftarrow j$
            \EndIf
        \EndFor
        \If{$maxW=0$}
            break
        \EndIf
        \State $sumW\leftarrow sumW+maxW$
        \State $matched[p]\leftarrow True,matched[q]\leftarrow True$
    \EndWhile
    \State \Return $sumW$
\EndFunction
\end{algorithmic}
\end{algorithm}

\subsection{Experiments}
\label{subsection: appendix_experiments}

This section explains more details about our benchmark, code implementation, and how to reproduce the experiments. We also illustrate experiment results with different delay probabilities. To be more self-contained, this section may overlap with the main text.

\subsubsection{Benchmark}
The benchmark generation process is described in the main text. Notably, our benchmark evaluates with as twice many agents as the benchmark used in the GSES paper \citep{FengICAPS24}. Therefore, We use $k$-robust PBS rather than $k$-robust CBS to obtain the initial MAPF plan since the former is much faster than the latter when solving instances with more agents. We run $k$-robust PBS on 25 instances with evenly distributed starts and goals from the MovingAI benchmark \footnote{https://www.movingai.com/benchmarks/mapf/index.html} and set the time limit of $k$-robust PBS to be $6$ minutes for each MAPF instance. Each solved instance will generate $6$ delay scenarios by simulation. We report the number of solved instances for each map and each agent number in \Cref{tab: solved_instances}.


\begin{table}[tb]
    \caption{The number of solved Instances for each map and each agent number.}
    \centering
    \resizebox{0.41\textwidth}{!}{
    \begin{tabular}{cccccc}
        \toprule
        \toprule
        \multicolumn{6}{c}{Random-32-32-10}\\
        \midrule
        \#agents  & 60 & 70 & 80 & 90  & 100 \\
        \#solved instances & 24 & 21 & 25 & 22 & 16 \\
        \midrule
        \midrule
        \multicolumn{6}{c}{Warehouse-10-20-10-2-1}\\
        \midrule
        \#agents & 110 & 120 & 130 & 140 & 150 \\
        \#solved instances & 24 & 23 & 22 & 23 & 20 \\
        \midrule
        \midrule
        \multicolumn{6}{c}{Lack303d}\\
        \midrule
        \#agents & 41 & 49 & 57 & 65 & 73 \\
        \#solved instances & 25 & 25 & 25 & 25 & 25 \\
        \midrule
        \midrule
        \multicolumn{6}{c}{Paris\_1\_256}\\
        \midrule
        \#agents & 120 & 140 & 160 & 180 & 200 \\
        \#solved instances & 22 & 20 & 20 & 18 & 18 \\
        \bottomrule
        \bottomrule
    \end{tabular}
    }
    \label{tab: solved_instances}
\end{table}

\subsubsection{Code Implementation}
In this paper, we compare with other three optimal STPG optimization algorithms. We incorporate their open-source implementation into our code.

\begin{enumerate}
    \item GSES: \url{https://github.com/YinggggFeng/Multi-Agent-via-Switchable-ADG}
    \item MILP: \url{https://github.com/alexberndt/sadg-controller}
    \item CBS-D: \url{https:/github.com/aria-systems-group/Delay-Robust-MAPF}
\end{enumerate}

All codes are implemented in C++ except MILP, which uses a Python interface to an open-sourced brand-and-bound C++ solver, the COIN-OR Branch-and-Cut solver, at \url{https://github.com/coin-or/Cbc}.

Notably, the code can handle non-uniform edge costs, even though we assume all the costs are $1$ in the main text.

\subsubsection{Reproducibility}
Our code and data are publicly available at https://github.com/DiligentPanda/STPG.git. Results can be reproduced by running the experiment scripts in the Linux system (e.g., Ubuntu). The readme file in the code repository explains the details.

\subsubsection{Experiment Results}
We organize all the experiment data in \Cref{tab:all_experiments}. Success Rate, Incremental Analysis, and Ablation Study with delay probability $p=0.01$ have been discussed in the main text. The conclusions for other delay probabilities are consistent with the ones in the main text. 

We further include plots of the mean search time for different maps and agent numbers, which reflects the superiority of IGSES over other optimal algorithms from another view. Of course, with more agents and larger delay probabilities, the search clearly takes more time.

\begin{table}[tb]
    \caption{Reference to all experiment data.}
    \label{tab:all_experiments}
    \centering
    \resizebox{0.48\textwidth}{!}{
        \begin{tabular}{cccc}
            \toprule
            \toprule
           Delay Probability  & $p=0.002$ & $p=0.01$  & $p=0.03$ \\
           \midrule
            Success Rate  & \Cref{fig:succ_rates_p002} &  \Cref{fig:succ_rates_p01} & \Cref{fig:succ_rates_p03} \\
            Search Time & \Cref{fig:search_time_p002} & \Cref{fig:search_time_p01} & \Cref{fig:search_time_p03} \\
            Incremental Analysis & \Cref{tab:incremental_p002} & \Cref{tab:incremental_p01}  &  \Cref{tab:incremental_p03} \\
             Ablation Study &
             \Cref{fig: ablation_p002} & 
             \Cref{fig: ablation_p01} &
             \Cref{fig: ablation_p03} \\
            \bottomrule
            \bottomrule
        \end{tabular}
    }
\end{table}

\begin{figure}[tb]
    \setlength{\abovecaptionskip}{4pt}  
    \setlength{\belowcaptionskip}{8pt}  

    \centering
    \begin{subfigure}[b]{0.45\textwidth}
      \centering
      \includegraphics[width=1\textwidth]{figures/comparison/success_rates_legend.pdf}
    \end{subfigure}%
    \hfill

    \begin{subfigure}[b]{0.225\textwidth}
      \centering
      \includegraphics[width=1\textwidth]{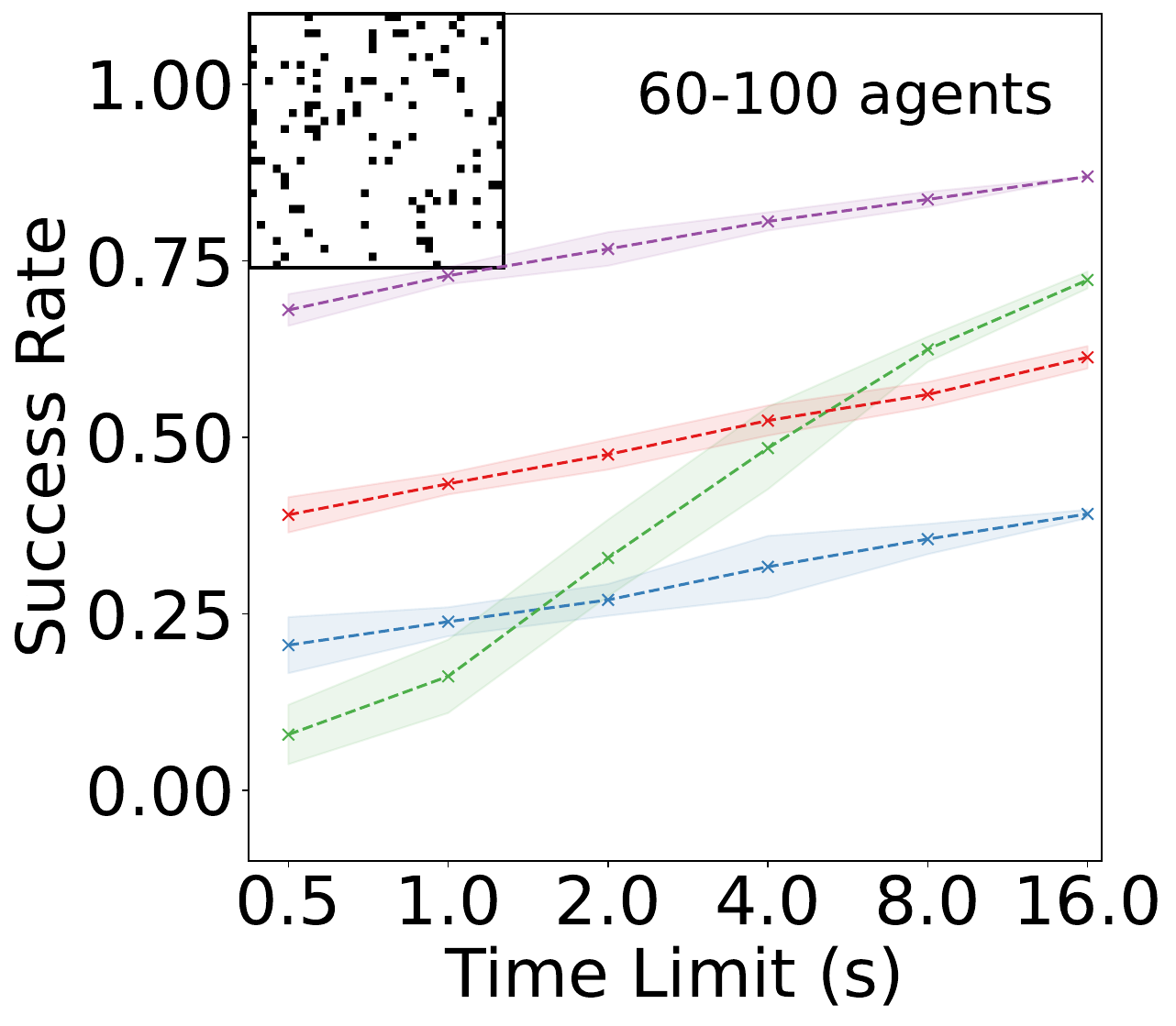}
      \caption{Random-32-32-10}
    \end{subfigure}%
    \hfill
    \begin{subfigure}[b]{0.225\textwidth}
      \centering
      \includegraphics[width=1\textwidth]{figures/comparison/success_rates_Warehouse_p01.pdf}
      \caption{Warehouse-10-20-10-2-1}
    \end{subfigure}%
    \hfill

    \setlength{\belowcaptionskip}{4pt}  
    
    \begin{subfigure}[b]{0.225\textwidth}
      \centering
      \includegraphics[width=1\textwidth]{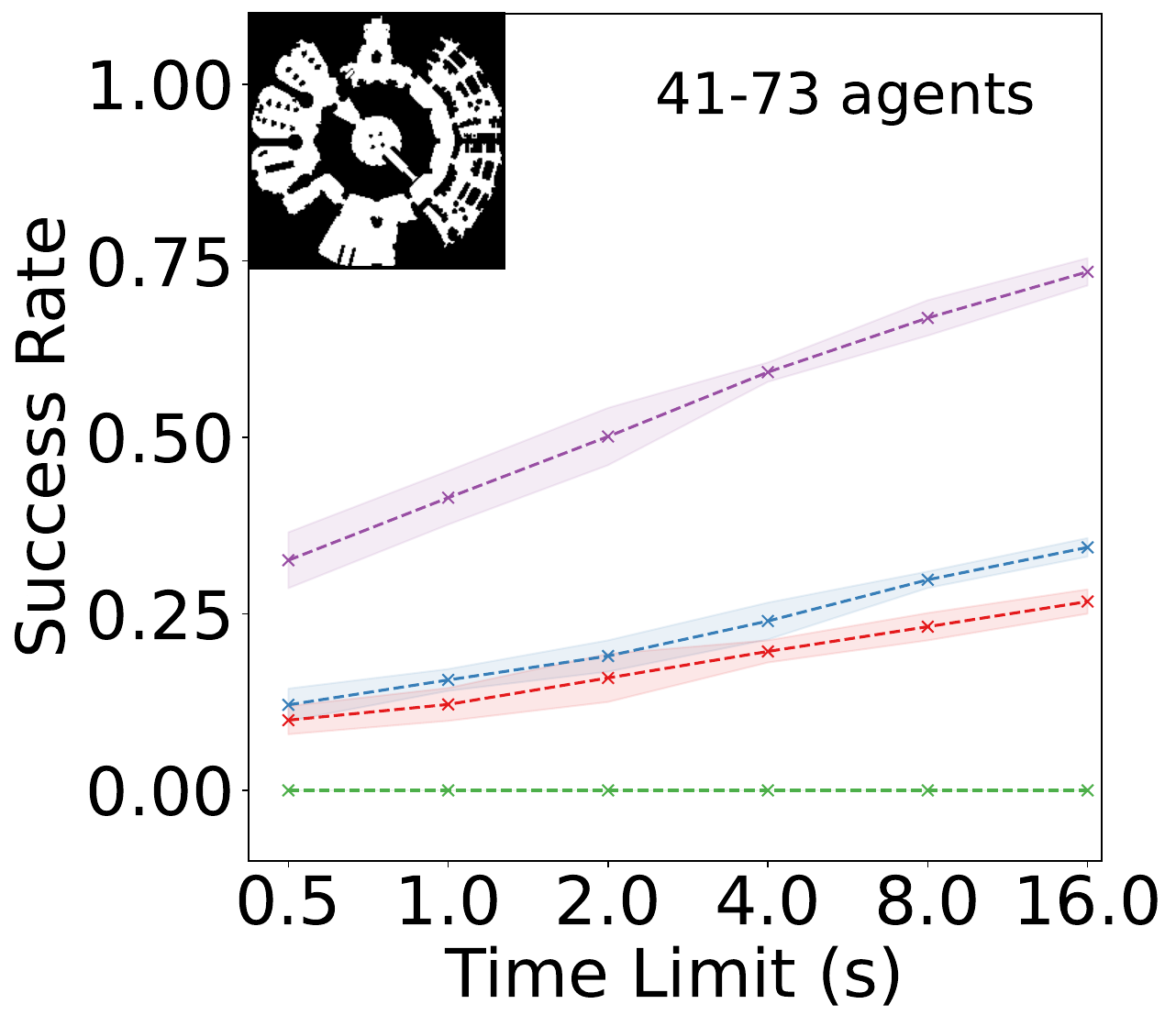}
      \caption{Lak303d}
    \end{subfigure}%
    \hfill
    \begin{subfigure}[b]{0.225\textwidth}
      \centering
      \includegraphics[width=1\textwidth]{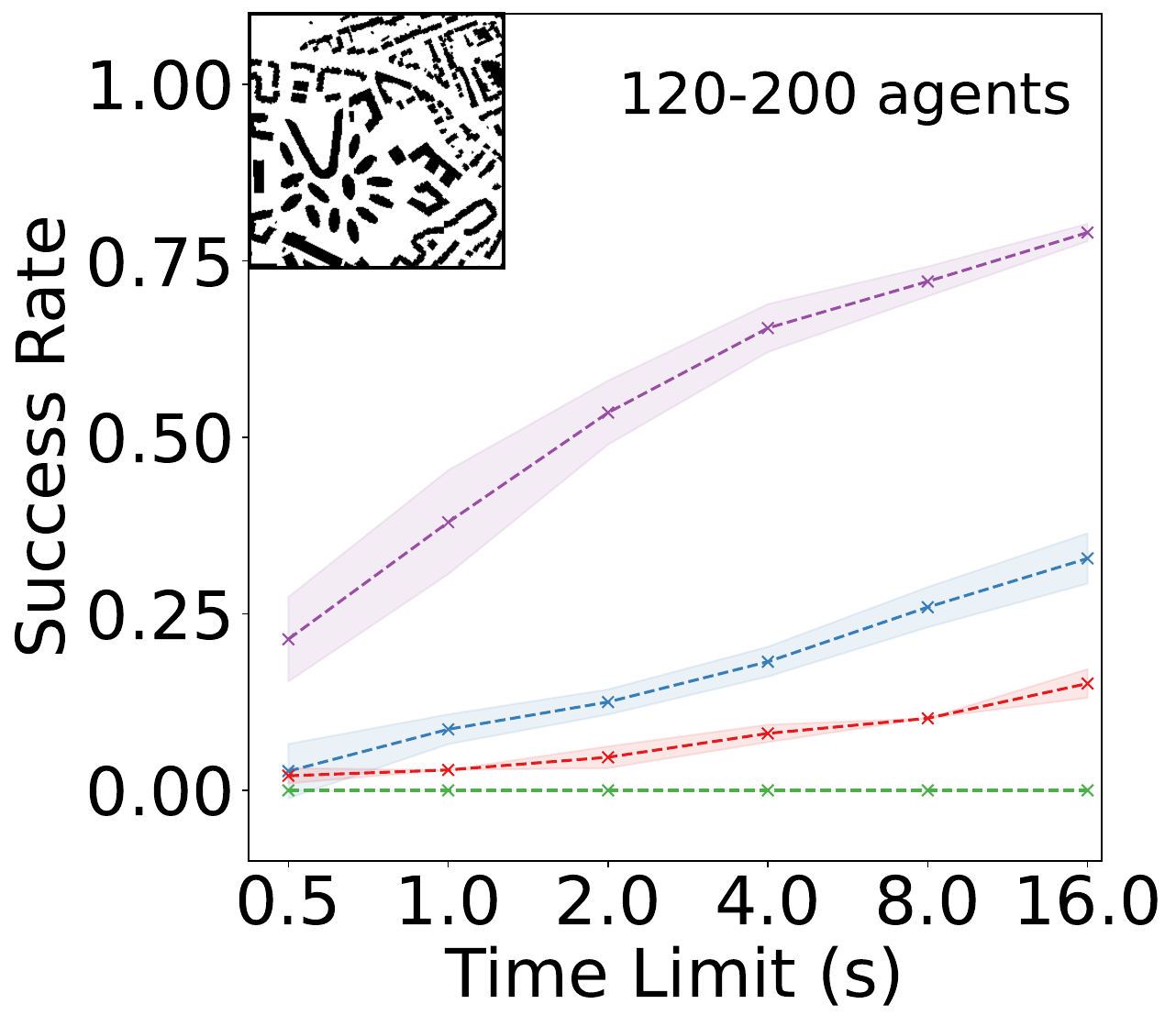}
      \caption{Paris\_1\_256}
    \end{subfigure}%
    \hfill

    \caption{Success rates of IGSES and other baselines on four maps with delay probability $p=0.002$. The shading areas indicate the standard deviations of different runs. They are multiplied by 10 for illustration. In each figure, the top-left shows the corresponding map, and the top-right corner shows the range of agent numbers.}
    \label{fig:succ_rates_p002}

    \centering
    \begin{subfigure}[b]{0.45\textwidth}
      \centering
      \includegraphics[width=1\textwidth]{figures/comparison/success_rates_legend.pdf}
    \end{subfigure}%
    \hfill


    

    \begin{subfigure}[b]{0.205\textwidth}
      \centering
      \includegraphics[width=1\textwidth]{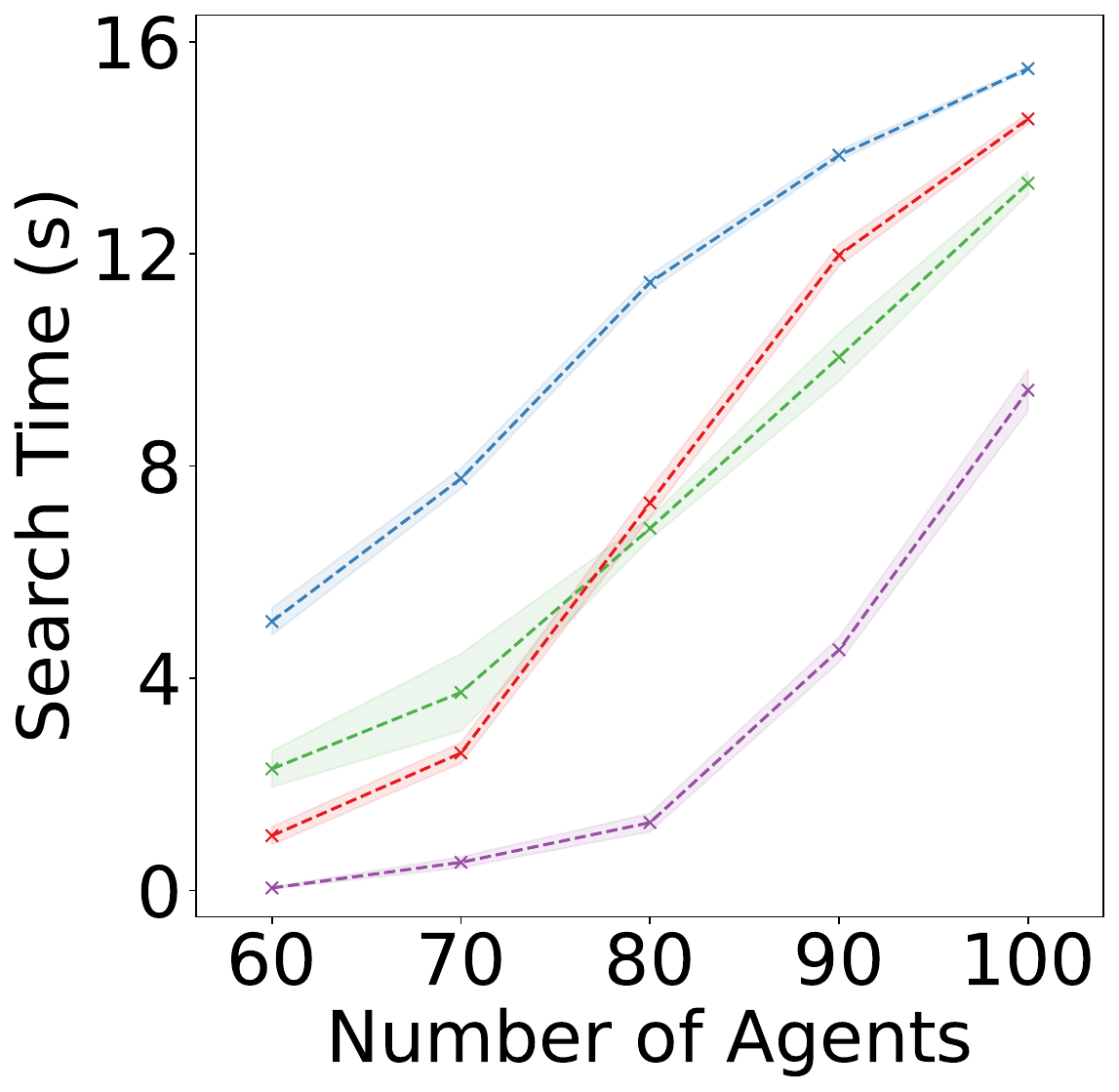}
      \caption{Random-32-32-10}
    \end{subfigure}%
    \hfill
    \begin{subfigure}[b]{0.205\textwidth}
      \centering
      \includegraphics[width=1\textwidth]{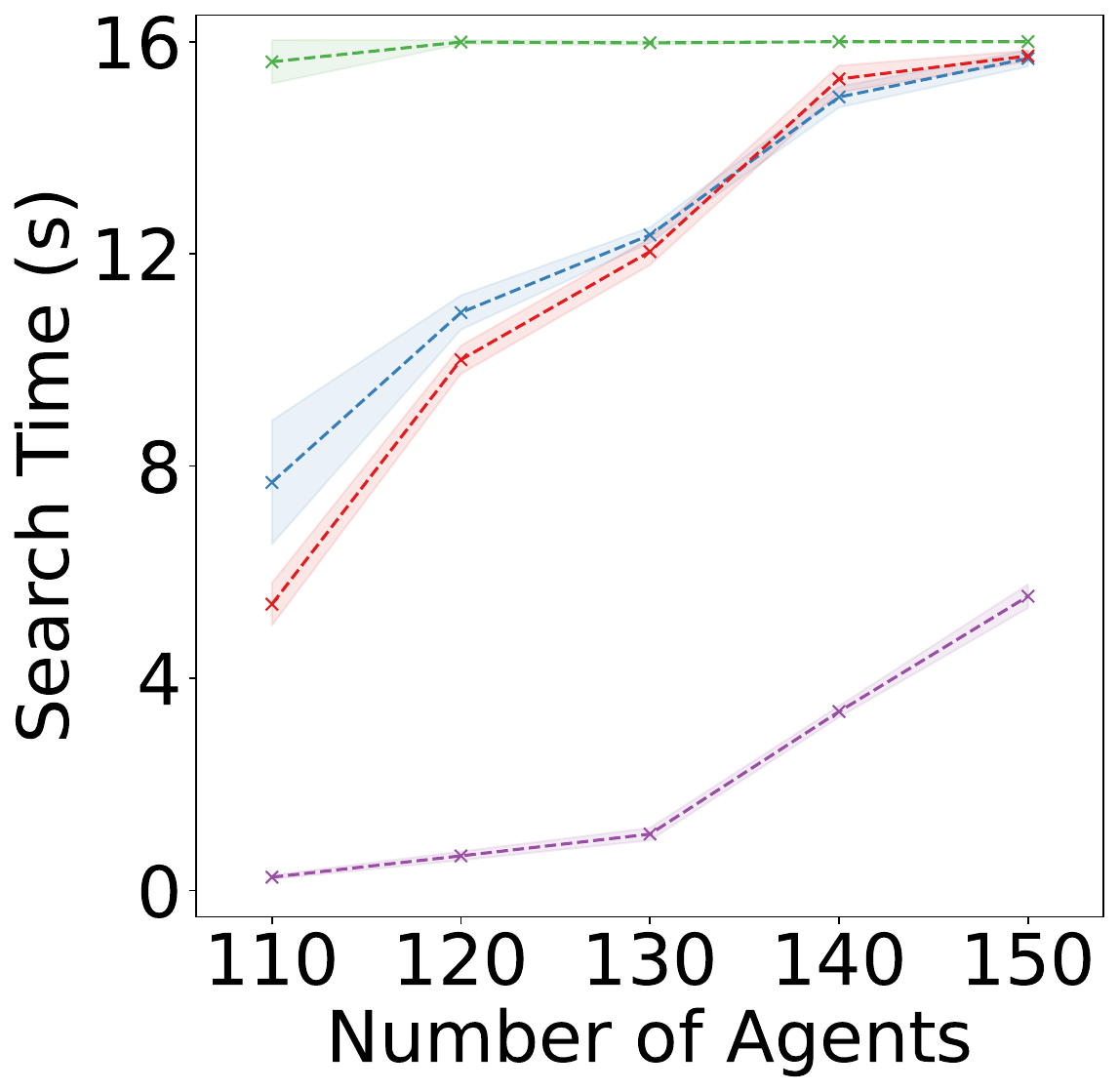}
      \caption{Warehouse-10-20-10-2-1}
    \end{subfigure}%
    \hfill
    \begin{subfigure}[b]{0.205\textwidth}
      \centering
      \includegraphics[width=1\textwidth]{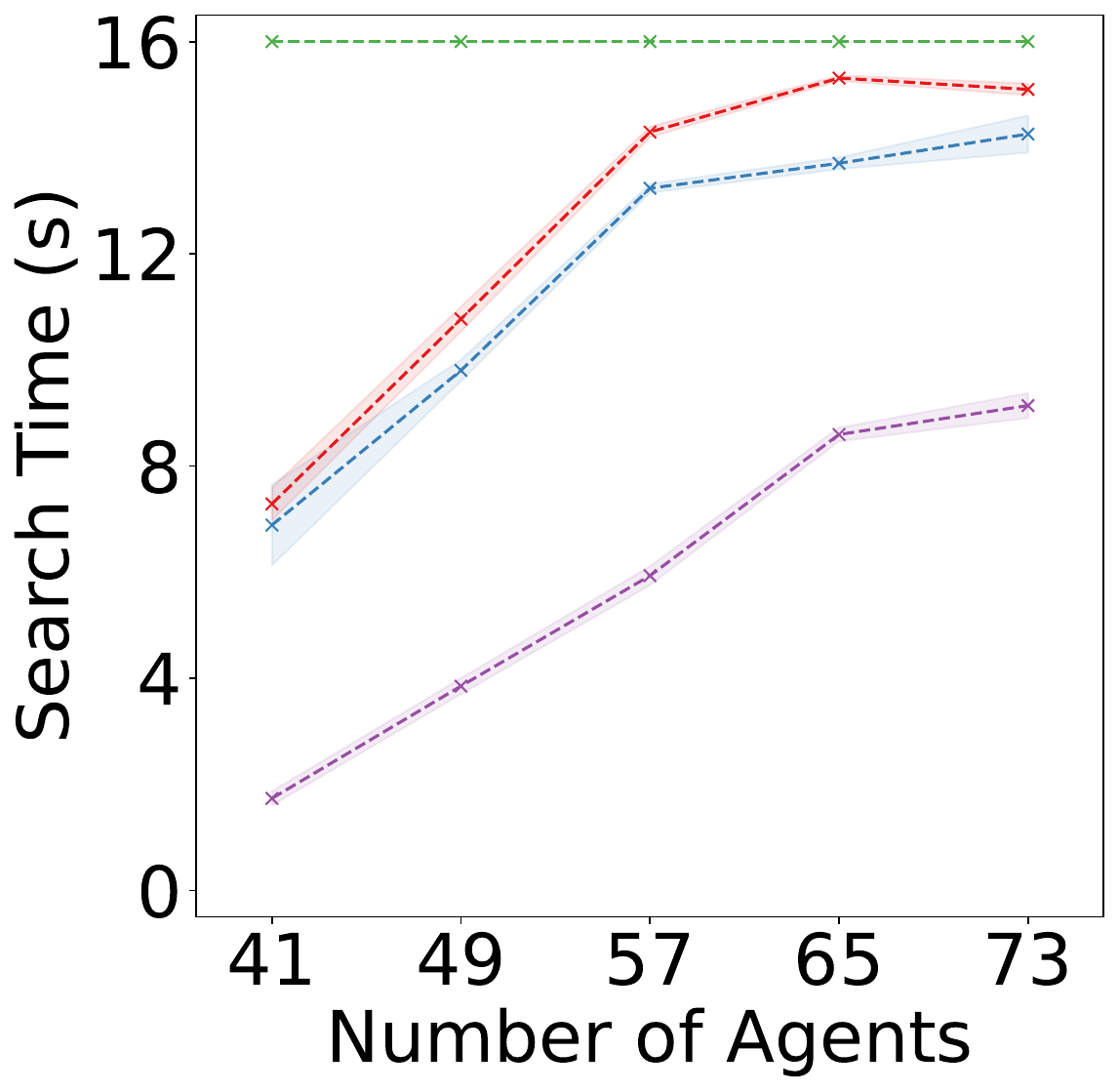}
      \caption{Lak303d}
    \end{subfigure}%
    \hfill
    \begin{subfigure}[b]{0.205\textwidth}
      \centering
      \includegraphics[width=1\textwidth]{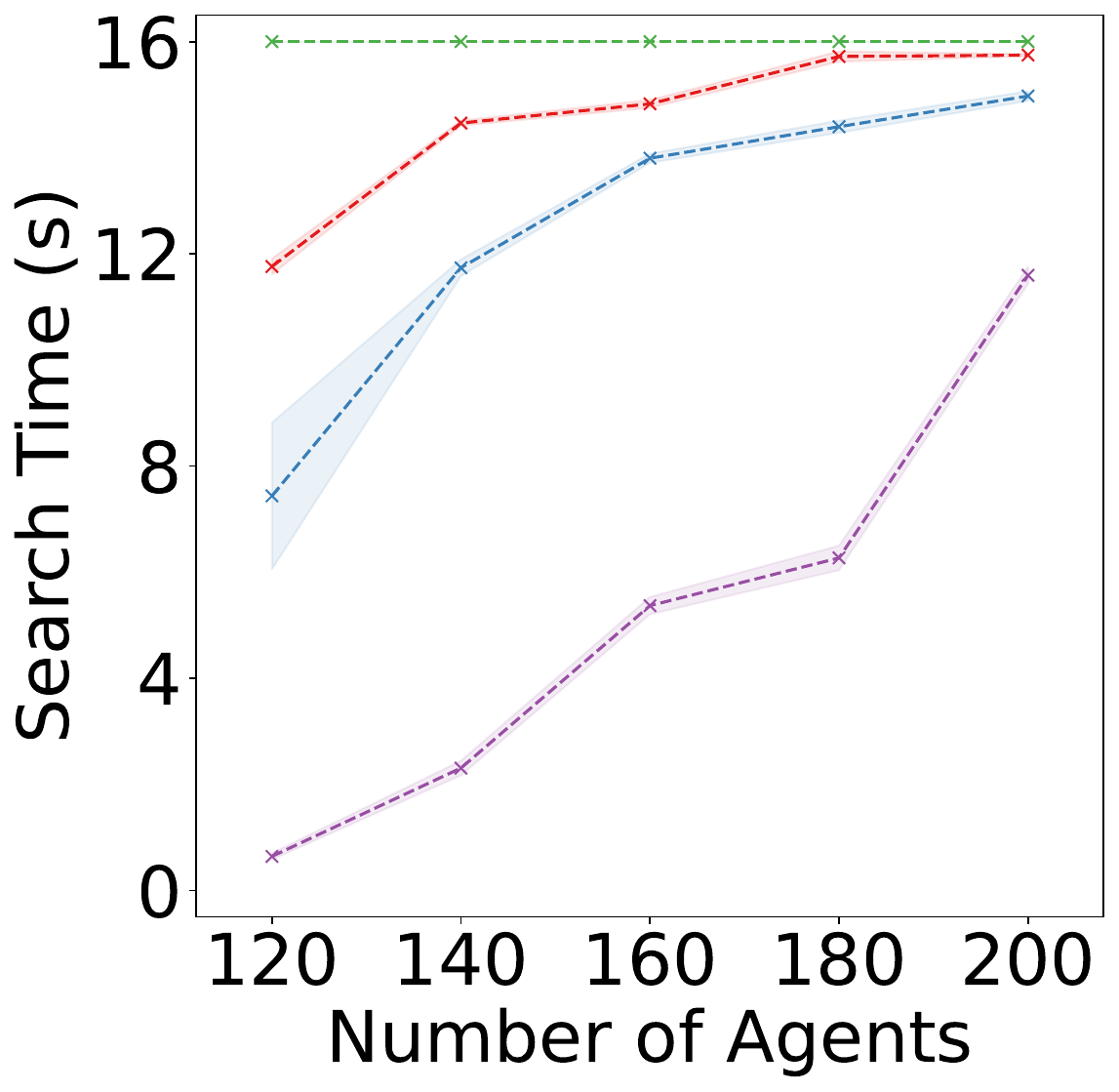}
      \caption{Paris\_1\_256}
    \end{subfigure}%
    \hfill
    \caption{Search Time of IGSES and other baselines on four maps with delay probability $p=0.002$. The search time of unsolved instances is set to the time limit of $16$ seconds. The shading areas indicate the standard deviations of different runs. They are multiplied by 10 for illustration.}
    \label{fig:search_time_p002}
\end{figure}

\begin{figure}[tbh]
    \setlength{\abovecaptionskip}{4pt}  
    \setlength{\belowcaptionskip}{8pt}  

    \centering
    \begin{subfigure}[b]{0.45\textwidth}
      \centering
      \includegraphics[width=1\textwidth]{figures/comparison/success_rates_legend.pdf}
    \end{subfigure}%
    \hfill

    \begin{subfigure}[b]{0.225\textwidth}
      \centering
      \includegraphics[width=1\textwidth]{figures/comparison/success_rates_Random_p01.pdf}
      \caption{Random-32-32-10}
    \end{subfigure}%
    \hfill
    \begin{subfigure}[b]{0.225\textwidth}
      \centering
      \includegraphics[width=1\textwidth]{figures/comparison/success_rates_Warehouse_p01.pdf}
      \caption{Warehouse-10-20-10-2-1}
    \end{subfigure}%
    \hfill

    \setlength{\belowcaptionskip}{4pt}  
    
    \begin{subfigure}[b]{0.225\textwidth}
      \centering
      \includegraphics[width=1\textwidth]{figures/comparison/success_rates_Game_p01.pdf}
      \caption{Lak303d}
    \end{subfigure}%
    \hfill
    \begin{subfigure}[b]{0.225\textwidth}
      \centering
      \includegraphics[width=1\textwidth]{figures/comparison/success_rates_City_p01.pdf}
      \caption{Paris\_1\_256}
    \end{subfigure}%
    \hfill

    \caption{Success rates of IGSES and other baselines on four maps with delay probability $p=0.01$. The shading areas indicate the standard deviations of different runs. They are multiplied by 10 for illustration. In each figure, the top-left shows the corresponding map, and the top-right corner shows the range of agent numbers.}
    \label{fig:succ_rates_p01}

    \centering
    \begin{subfigure}[b]{0.45\textwidth}
      \centering
      \includegraphics[width=1\textwidth]{figures/comparison/success_rates_legend.pdf}
    \end{subfigure}%
    \hfill


    

    \begin{subfigure}[b]{0.205\textwidth}
      \centering
      \includegraphics[width=1\textwidth]{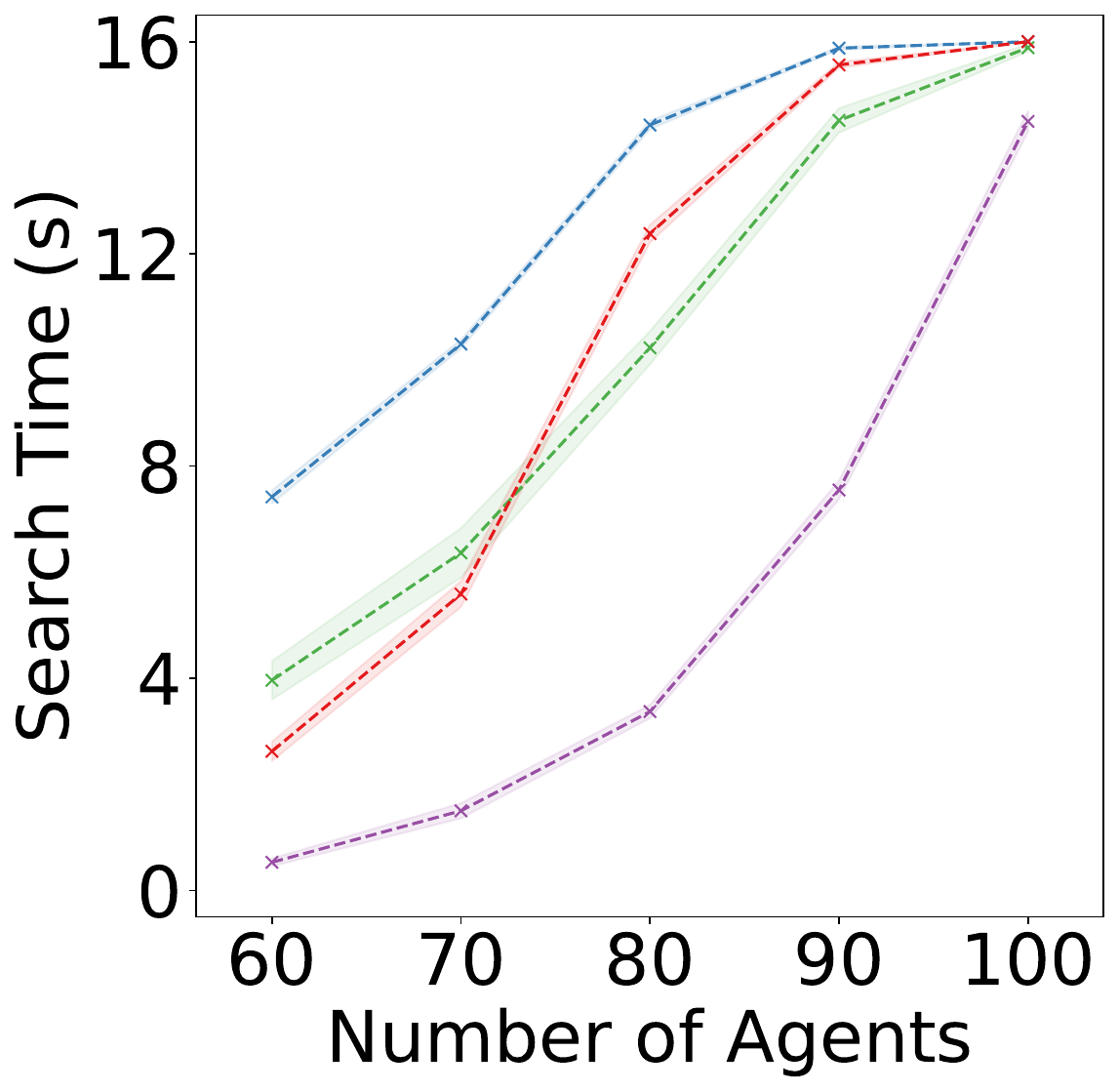}
      \caption{Random-32-32-10}
    \end{subfigure}%
    \hfill
    \begin{subfigure}[b]{0.205\textwidth}
      \centering
      \includegraphics[width=1\textwidth]{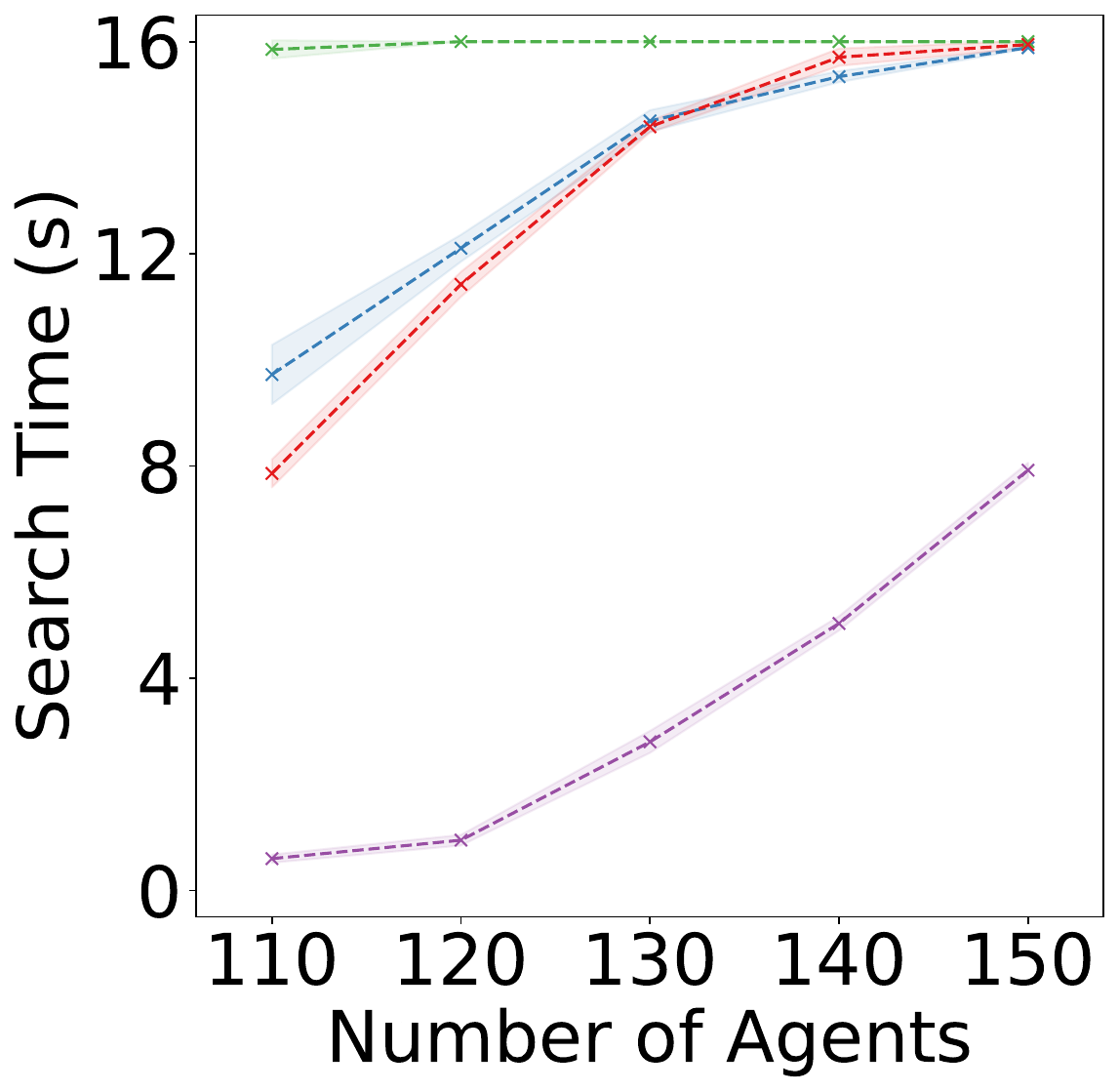}
      \caption{Warehouse-10-20-10-2-1}
    \end{subfigure}%
    \hfill
    \begin{subfigure}[b]{0.205\textwidth}
      \centering
      \includegraphics[width=1\textwidth]{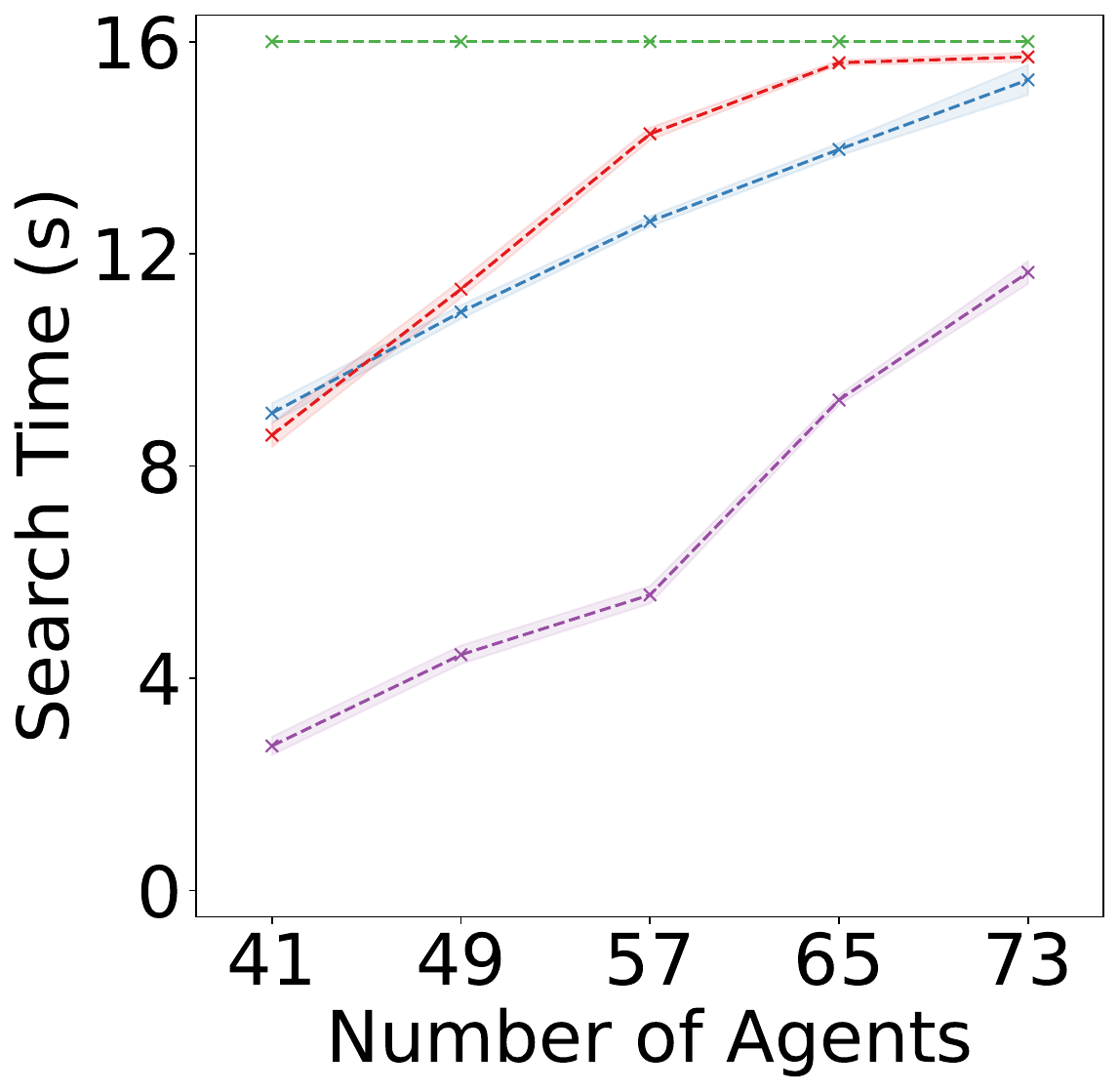}
      \caption{Lak303d}
    \end{subfigure}%
    \hfill
    \begin{subfigure}[b]{0.205\textwidth}
      \centering
      \includegraphics[width=1\textwidth]{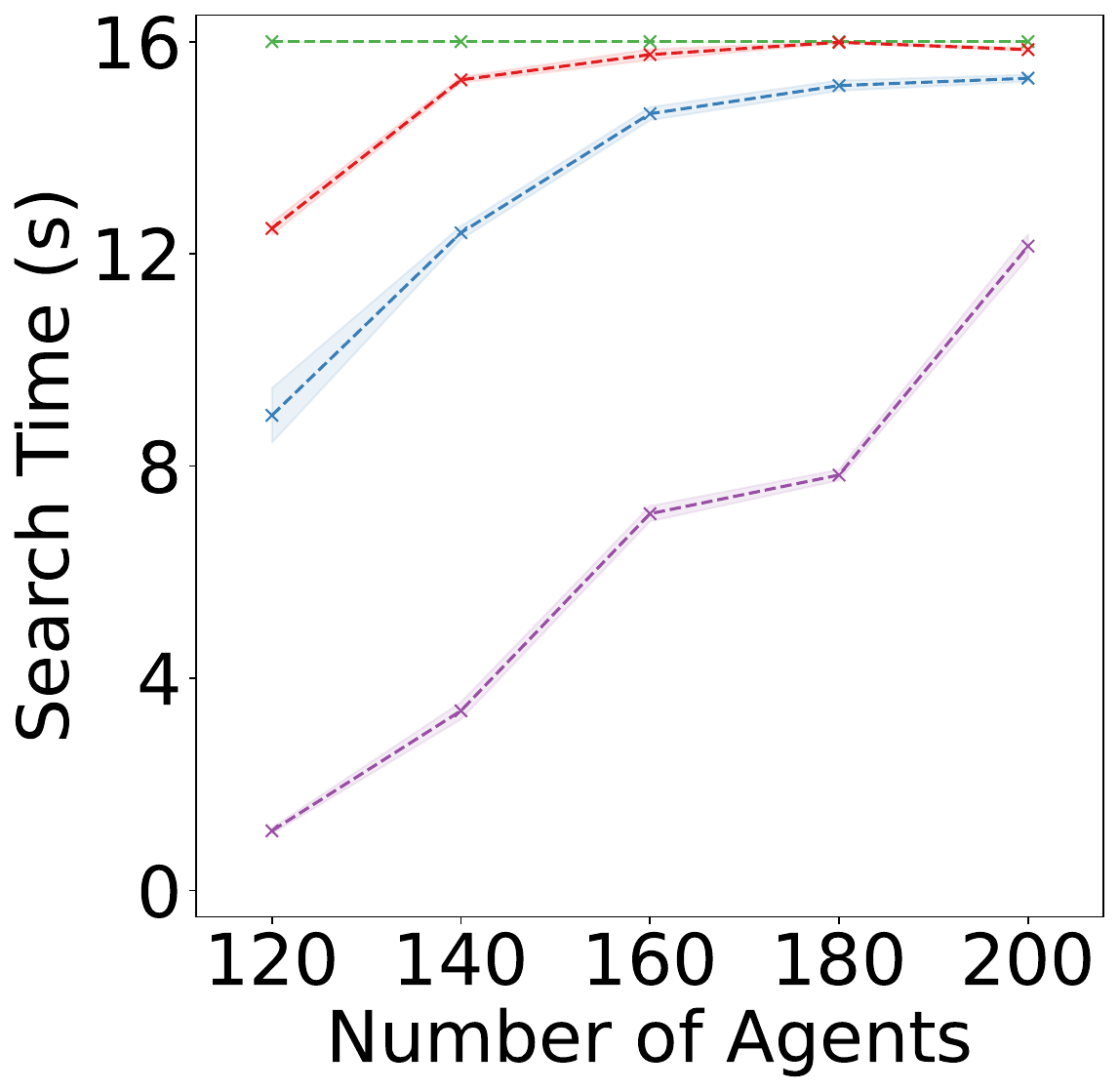}
      \caption{Paris\_1\_256}
    \end{subfigure}%
    \hfill
    \caption{Search Time of IGSES and other baselines on four maps with delay probability $p=0.01$. The search time of unsolved instances is set to the time limit of $16$ seconds. The shading areas indicate the standard deviations of different runs. They are multiplied by 10 for illustration.}
    \label{fig:search_time_p01}
\end{figure}

\begin{figure}[tb]
    \setlength{\abovecaptionskip}{4pt}  
    \setlength{\belowcaptionskip}{8pt}  

    \centering
    \begin{subfigure}[b]{0.45\textwidth}
      \centering
      \includegraphics[width=1\textwidth]{figures/comparison/success_rates_legend.pdf}
    \end{subfigure}%
    \hfill

    \begin{subfigure}[b]{0.225\textwidth}
      \centering
      \includegraphics[width=1\textwidth]{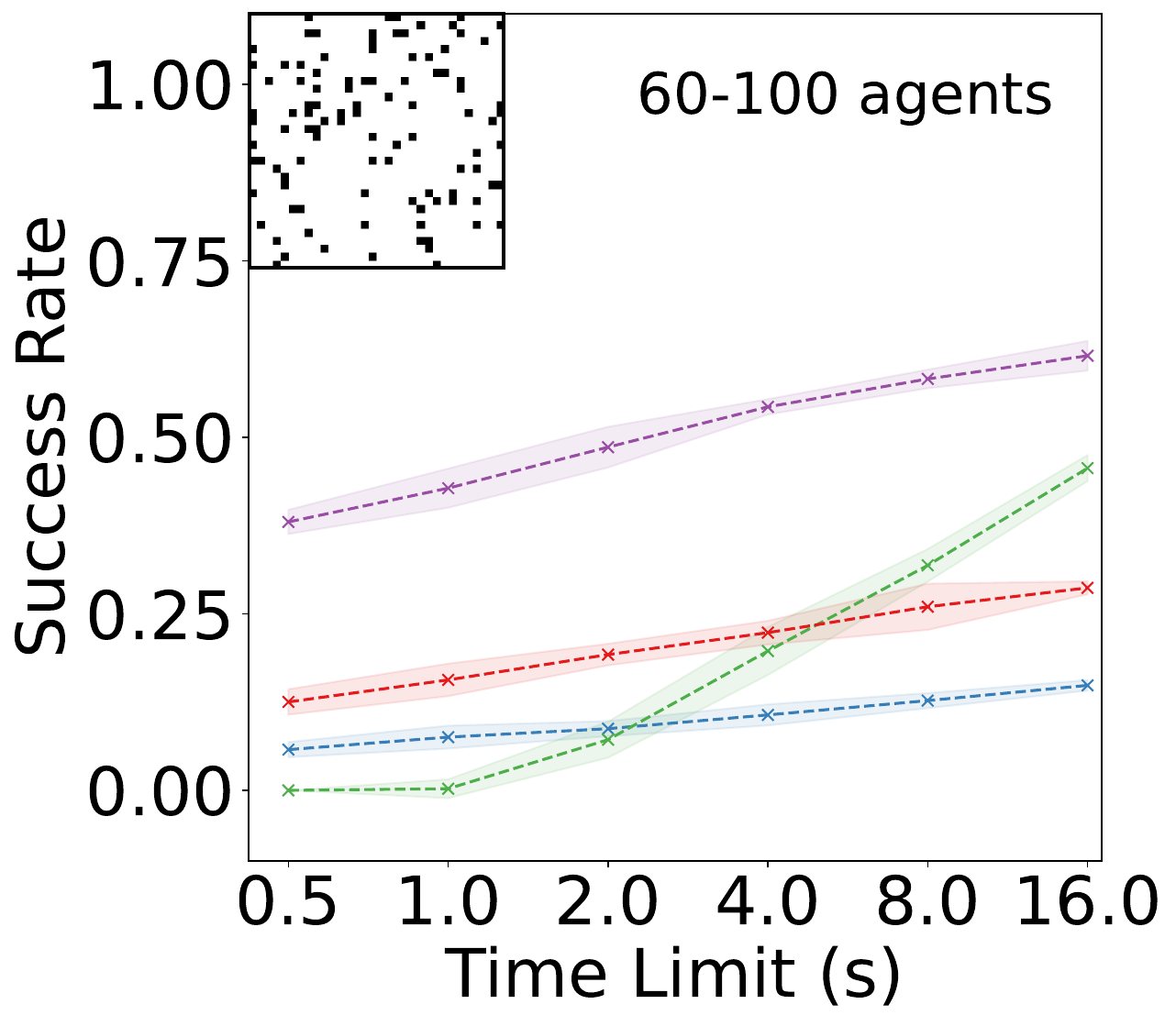}
      \caption{Random-32-32-10}
    \end{subfigure}%
    \hfill
    \begin{subfigure}[b]{0.225\textwidth}
      \centering
      \includegraphics[width=1\textwidth]{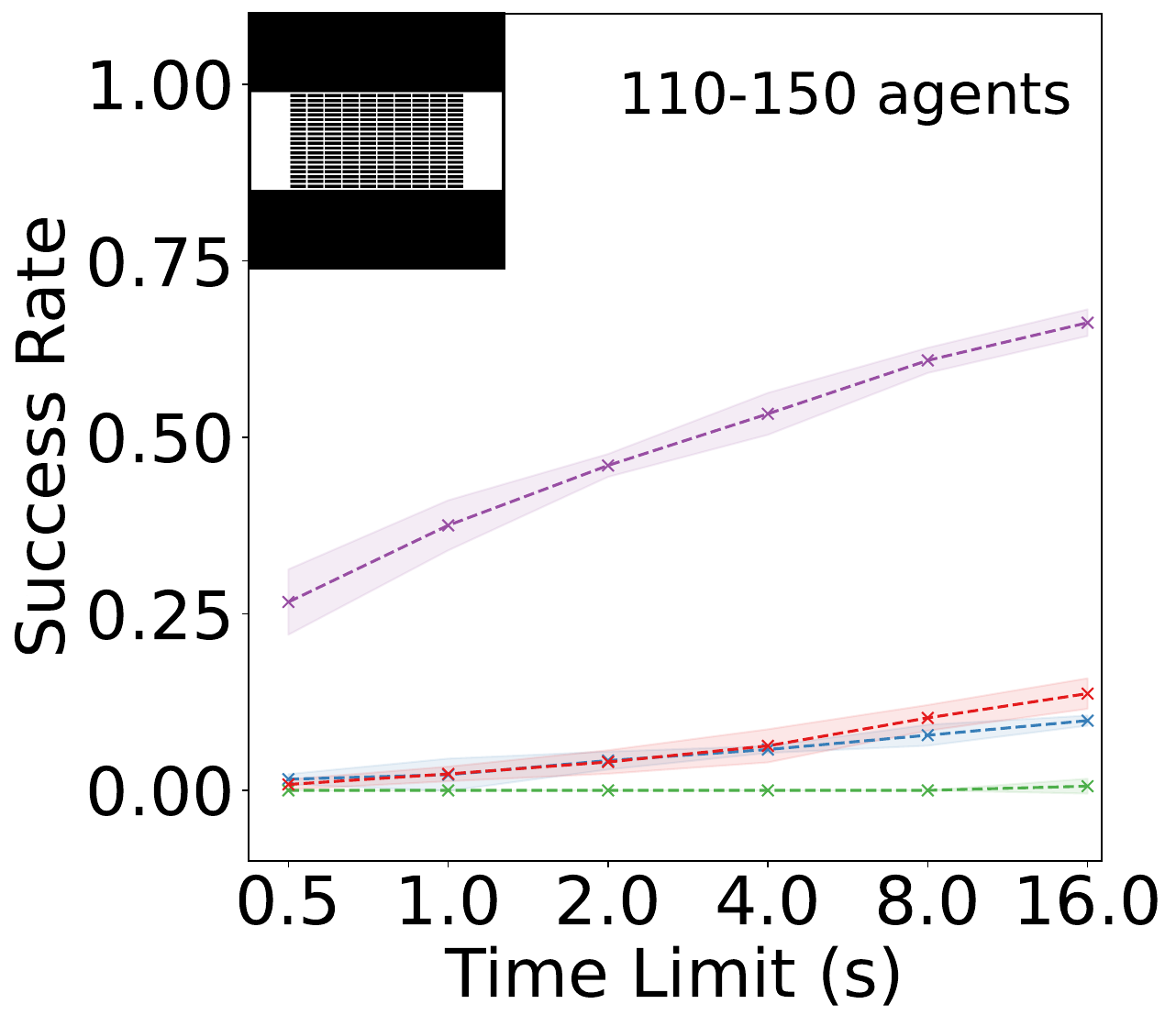}
      \caption{Warehouse-10-20-10-2-1}
    \end{subfigure}%
    \hfill

    \setlength{\belowcaptionskip}{4pt}  
    
    \begin{subfigure}[b]{0.225\textwidth}
      \centering
      \includegraphics[width=1\textwidth]{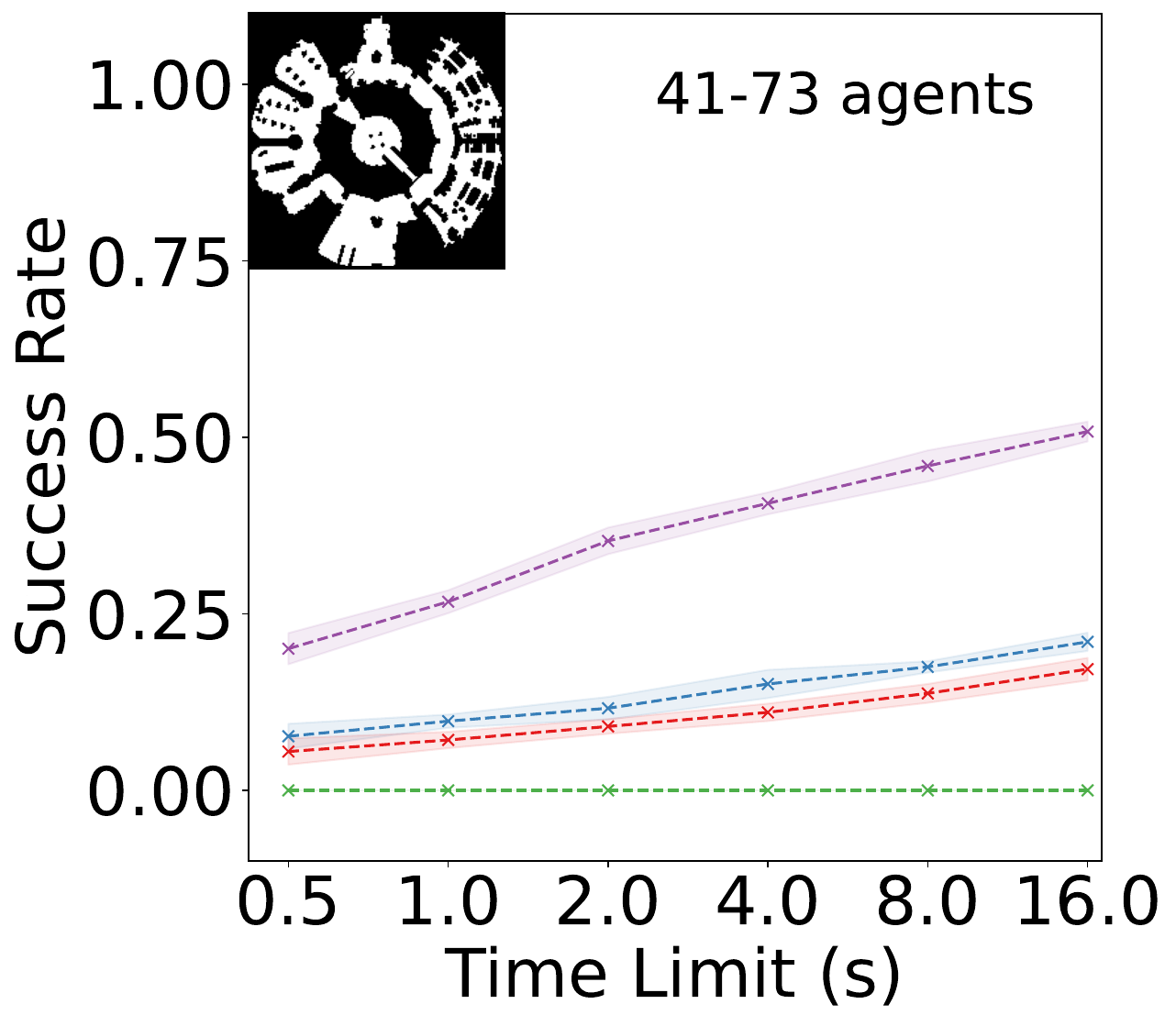}
      \caption{Lak303d}
    \end{subfigure}%
    \hfill
    \begin{subfigure}[b]{0.225\textwidth}
      \centering
      \includegraphics[width=1\textwidth]{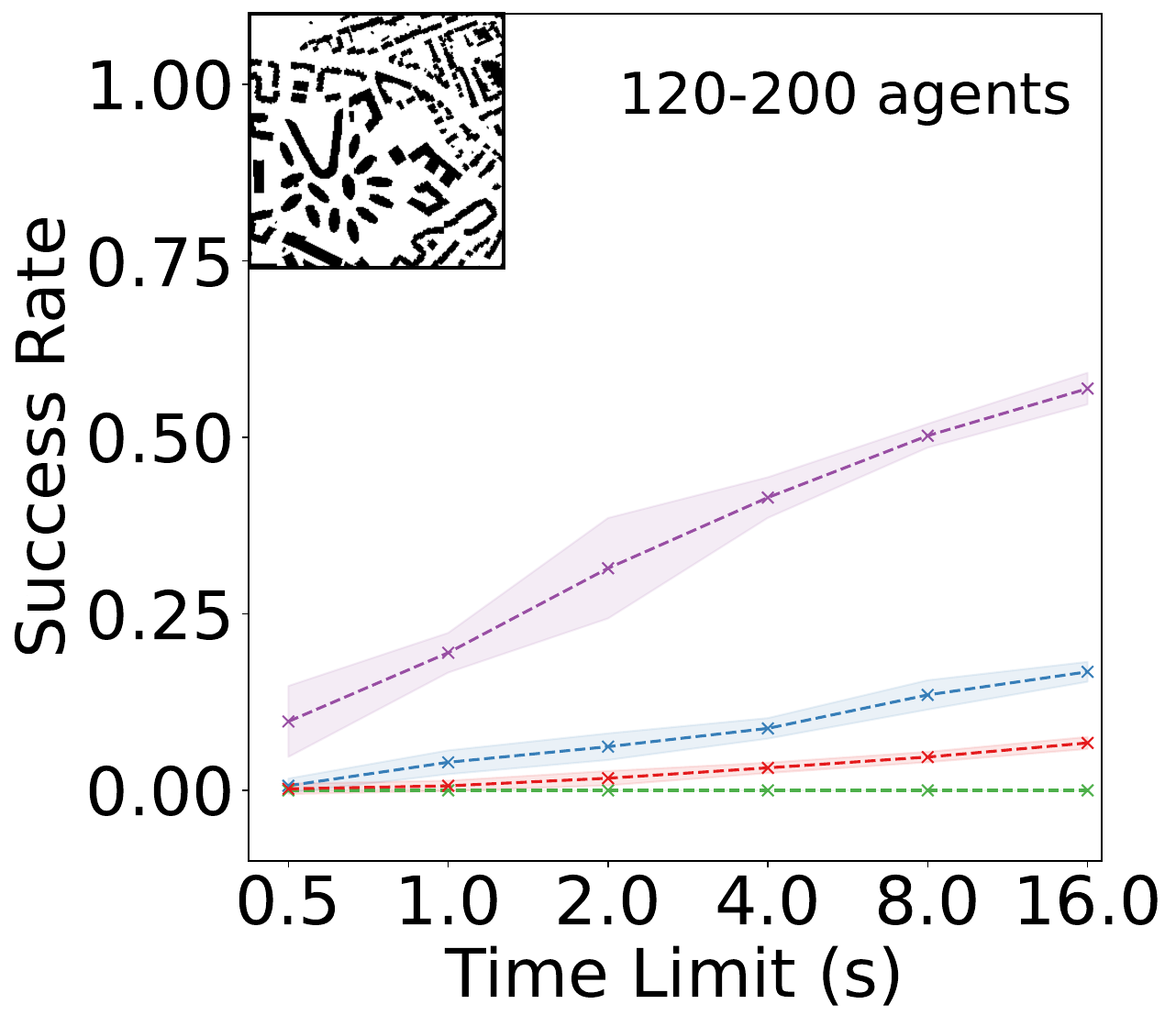}
      \caption{Paris\_1\_256}
    \end{subfigure}%
    \hfill

    \caption{Success rates of IGSES and other baselines on four maps with delay probability $p=0.03$. The shading areas indicate the standard deviations of different runs. They are multiplied by 10 for illustration. In each figure, the top-left shows the corresponding map, and the top-right corner shows the range of agent numbers.}
    \label{fig:succ_rates_p03}

    \centering
    \begin{subfigure}[b]{0.45\textwidth}
      \centering
      \includegraphics[width=1\textwidth]{figures/comparison/success_rates_legend.pdf}
    \end{subfigure}%
    \hfill


    

    \begin{subfigure}[b]{0.205\textwidth}
      \centering
      \includegraphics[width=1\textwidth]{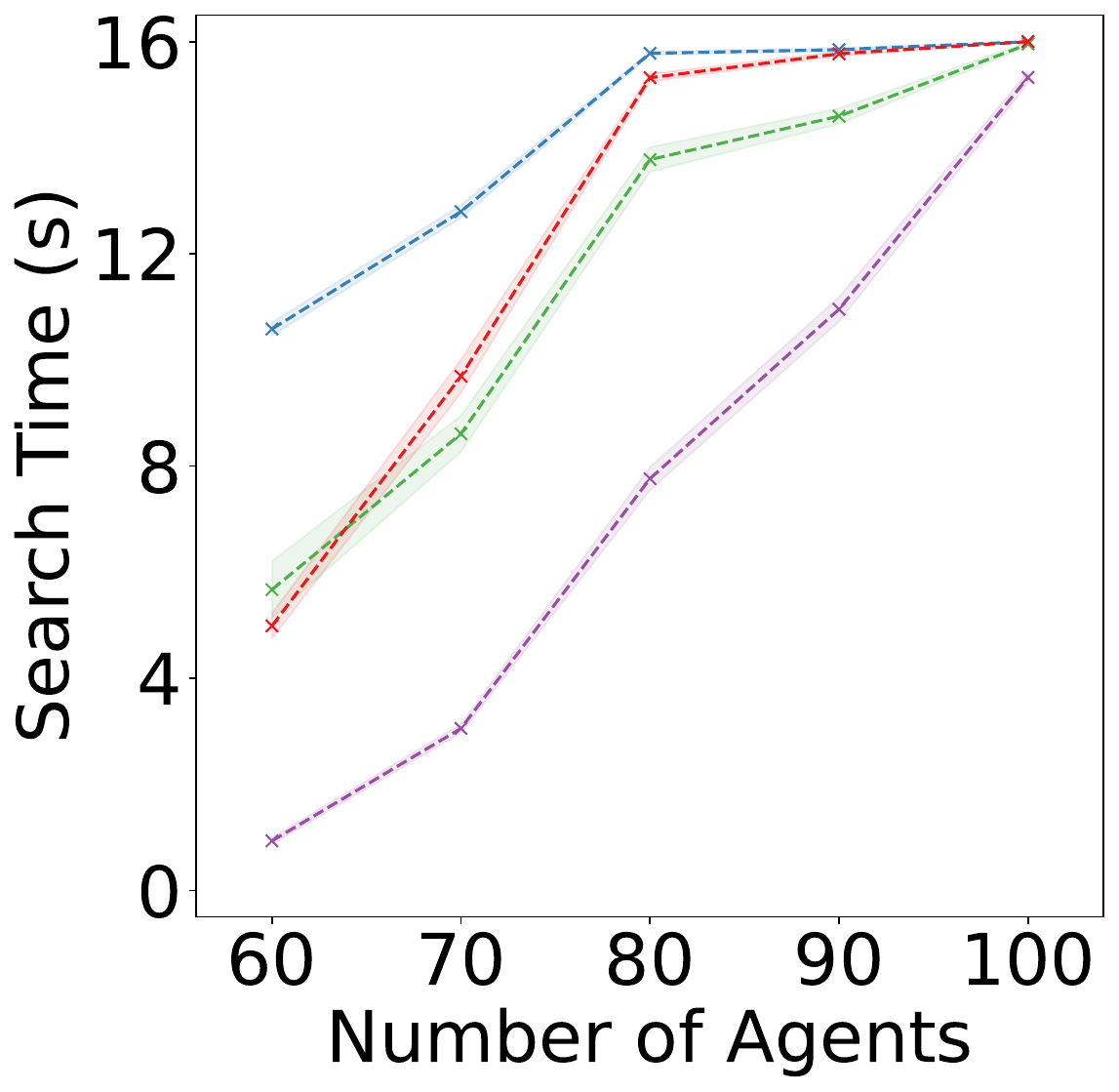}
      \caption{Random-32-32-10}
    \end{subfigure}%
    \hfill
    \begin{subfigure}[b]{0.205\textwidth}
      \centering
      \includegraphics[width=1\textwidth]{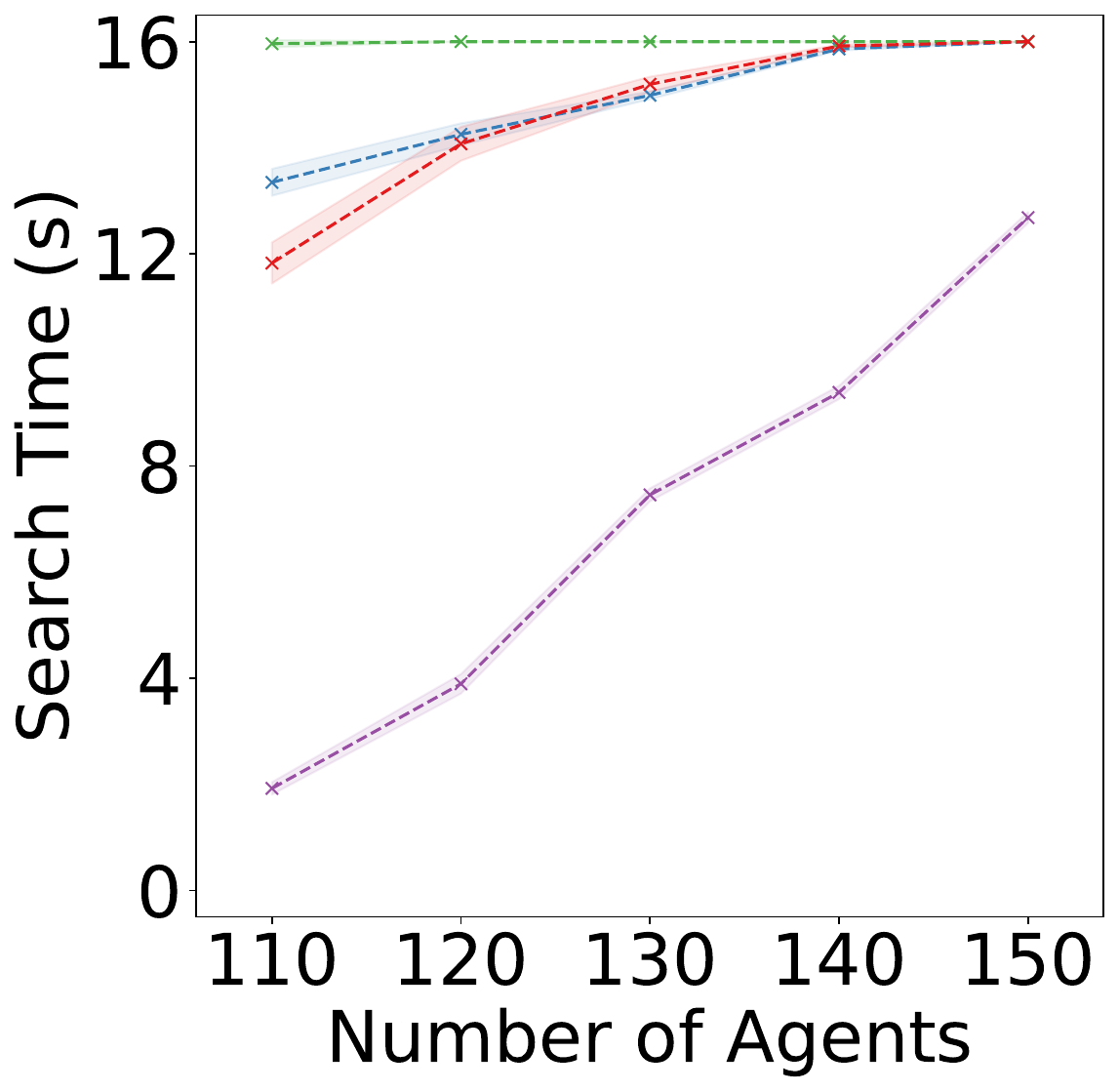}
      \caption{Warehouse-10-20-10-2-1}
    \end{subfigure}%
    \hfill
    \begin{subfigure}[b]{0.205\textwidth}
      \centering
      \includegraphics[width=1\textwidth]{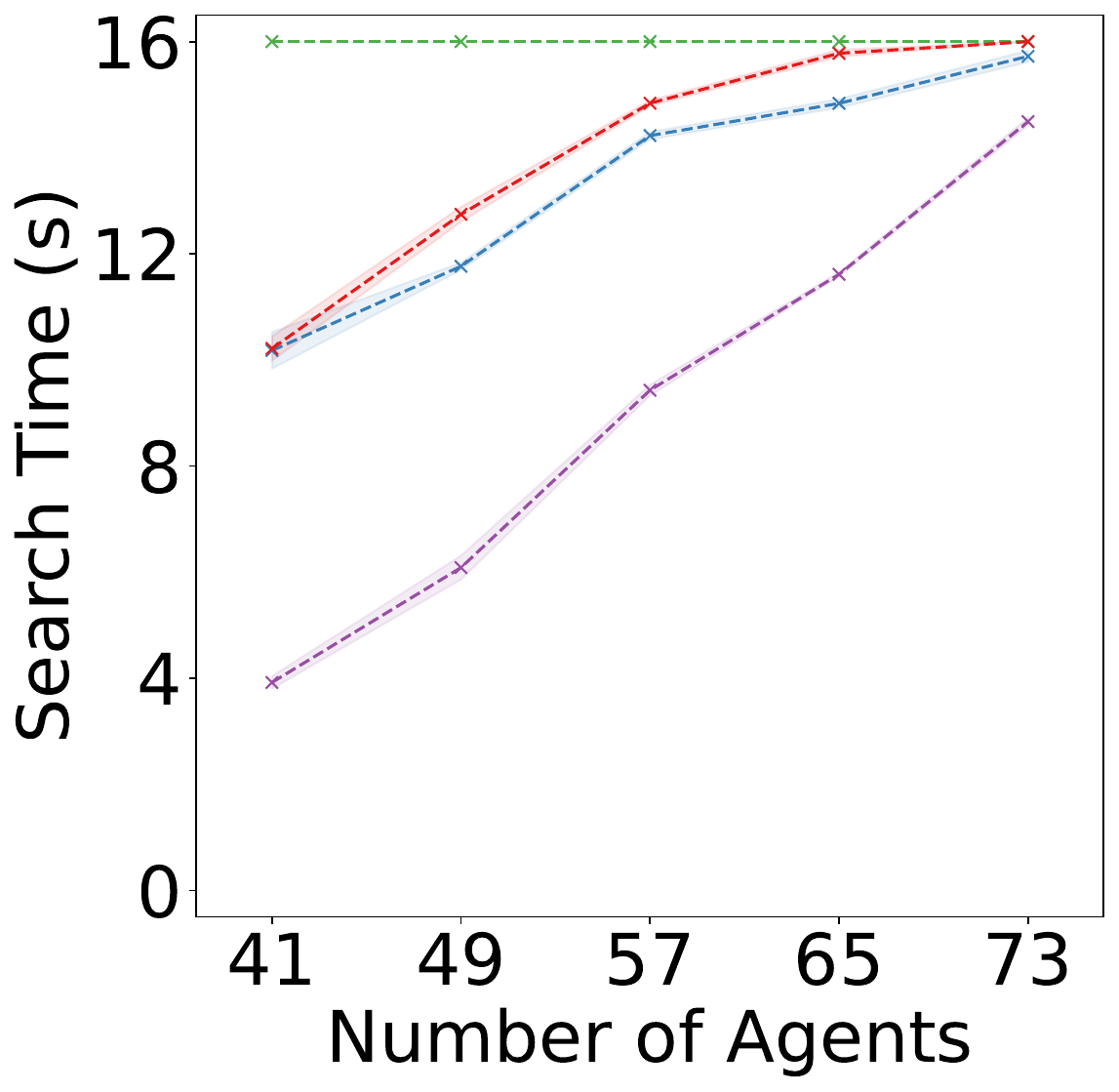}
      \caption{Lak303d}
    \end{subfigure}%
    \hfill
    \begin{subfigure}[b]{0.205\textwidth}
      \centering
      \includegraphics[width=1\textwidth]{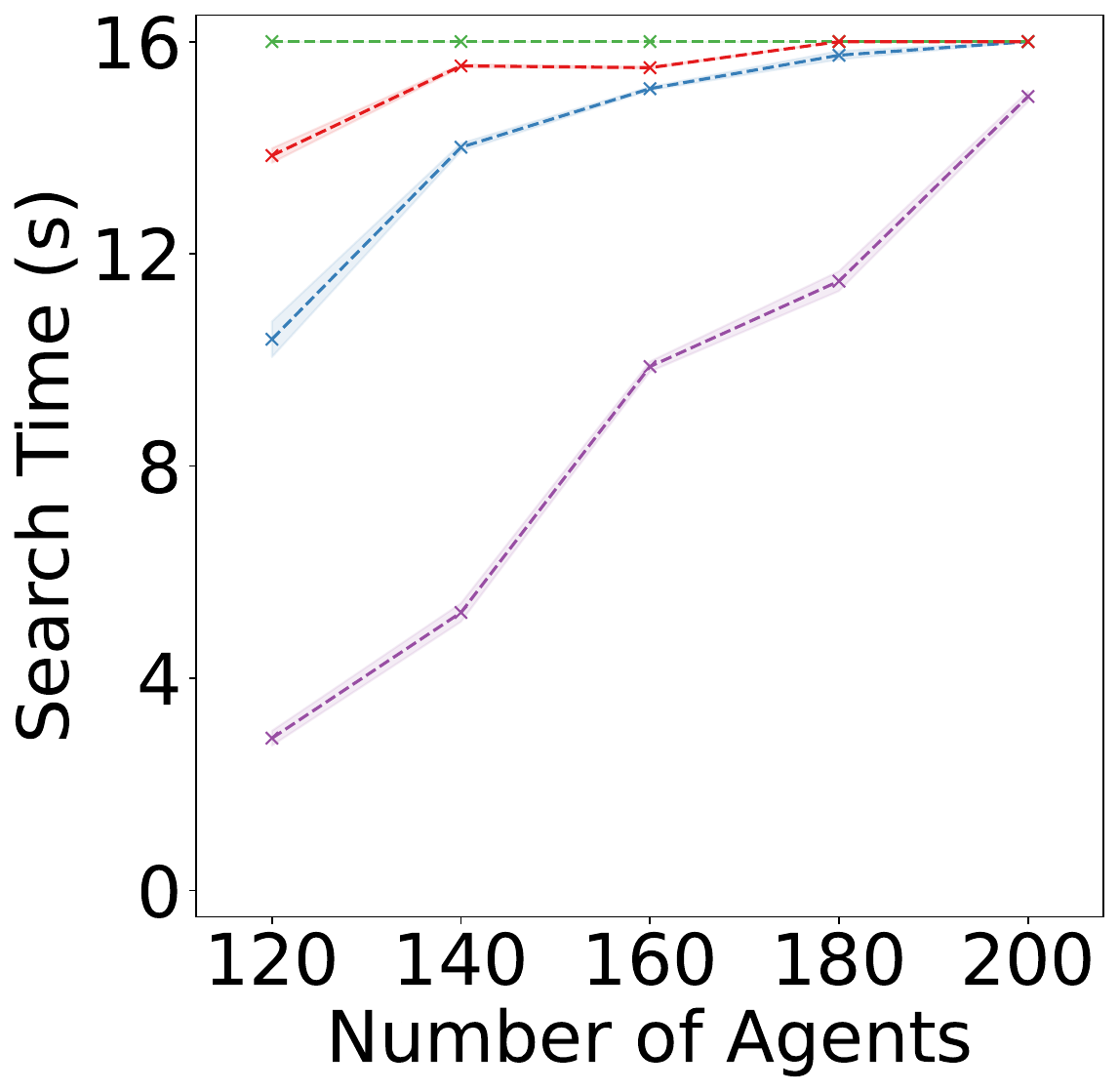}
      \caption{Paris\_1\_256}
    \end{subfigure}%
    \hfill
    \caption{Search Time of IGSES and other baselines on four maps with delay probability $p=0.03$. The search time of unsolved instances is set to the time limit of $16$ seconds. The shading areas indicate the standard deviations of different runs. They are multiplied by 10 for illustration.}
    \label{fig:search_time_p03}
\end{figure}

\begin{table*}[tbh]
    \caption{Incremental analysis on instances solved by all the settings with delay probability $p=0.002$. Rows 1,2,3 compare the effects of different grouping methods. GSES has no grouping. SG: simple grouping. FG: full grouping. Rows 3,4,5,6 compare the effects of different branching orders. Row 3 applies the default Agent-First branching in GSES. RB: Random branching. EB: Earliest-First branching. SB: Smallest-Edge-Slack-First branching. Rows 6,7 compare the effect of the stronger heuristics. SH: stronger heuristics. Rows 7,8 compare the effect of the incremental implementation. INC: incremental implementation. The last row is also our IGSES.}
    \label{tab:incremental_p002}
    \centering
    \resizebox{0.98\textwidth}{!}{
    \begin{tabular}{c|c|rrr|rrr|rrr|rrr}
        \toprule
        \toprule
        \multirow{3}{*}{\makecell{Row}} & \multirow{3}{*}{Setting} & \multicolumn{3}{c|}{Random-32-32-10} &
        \multicolumn{3}{c|}{Warehouse-10-20-10-2-1} & 
        \multicolumn{3}{c|}{Lak303d} & 
        \multicolumn{3}{c}{Paris\_1\_256}\\
        & & \makecell{search\\time (s)} & \makecell{\#expanded\\nodes} & \makecell{\#edge\\ groups} & \makecell{search\\time (s)} & \makecell{\#expanded\\nodes} & \makecell{\#edge\\ groups} & 
        \makecell{search\\time (s)} & \makecell{\#expanded\\nodes} & \makecell{\#edge\\ groups} & 
        \makecell{search\\time (s)} & \makecell{\#expanded\\nodes} & \makecell{\#edge\\ groups}\\
        \midrule
    
1 &GSES & 1.093 & 2249.8 & 1164.0 & 3.013 & 757.5 & 14686.1 & 1.648 & 333.0 & 30230.5 & 4.516 & 190.5 & 44884.0\\
2 &+SG & 0.840 & 1475.6 & 826.7 & 1.580 & 285.7 & 8668.4 & 0.848 & 112.2 & 18938.0 & 3.292 & 104.2 & 30191.6\\
3 &+FG & 0.747 & 1213.1 & 518.8 & 1.355 & 218.4 & 2200.0 & 0.704 & 80.5 & 5581.3 & 2.949 & 83.5 & 11873.3\\
4 &+FG +RB & 1.259 & 1996.9 & 518.8 & 2.226 & 358.1 & 2200.0 & 2.168 & 246.3 & 5581.3 & 4.393 & 122.1 & 11873.3\\
5 &+FG +EB & 1.054 & 1755.0 & 518.8 & 1.791 & 291.9 & 2200.0 & 1.120 & 115.2 & 5581.3 & 2.896 & 79.5 & 11873.3\\
6 &+FG +SB & 0.713 & 1111.7 & 518.8 & 1.312 & 205.0 & 2200.0 & 0.726 & 83.3 & 5581.3 & 3.119 & 88.4 & 11873.3\\
7 &+FG +SB +SH & 0.127 & 59.9 & 518.8 & 0.743 & 31.6 & 2200.0 & 0.701 & 25.1 & 5581.3 & 1.537 & 13.9 & 11873.3\\
8 &+FG +SB +SH +INC & 0.035 & 59.9 & 518.8 & 0.116 & 31.6 & 2200.0 & 0.207 & 25.1 & 5581.3 & 0.349 & 13.9 & 11873.3\\

        \bottomrule
        \bottomrule
    \end{tabular}
    }
\end{table*}

\begin{table*}[tbh]
    \caption{Incremental analysis on instances solved by all the settings with delay probability $p=0.01$. Rows 1,2,3 compare the effects of different grouping methods. GSES has no grouping. SG: simple grouping. FG: full grouping. Rows 3,4,5,6 compare the effects of different branching orders. Row 3 applies the default Agent-First branching in GSES. RB: Random branching. EB: Earliest-First branching. SB: Smallest-Edge-Slack-First branching. Rows 6,7 compare the effect of the stronger heuristics. SH: stronger heuristics. Rows 7,8 compare the effect of the incremental implementation. INC: incremental implementation. The last row is also our IGSES.}
    \label{tab:incremental_p01}
    \centering
    \resizebox{0.98\textwidth}{!}{
    \begin{tabular}{c|c|rrr|rrr|rrr|rrr}
        \toprule
        \toprule
        \multirow{3}{*}{\makecell{Row}} & \multirow{3}{*}{Setting} & \multicolumn{3}{c|}{Random-32-32-10} &
        \multicolumn{3}{c|}{Warehouse-10-20-10-2-1} & 
        \multicolumn{3}{c|}{Lak303d} & 
        \multicolumn{3}{c}{Paris\_1\_256}\\
        & & \makecell{search\\time (s)} & \makecell{\#expanded\\nodes} & \makecell{\#edge\\ groups} & \makecell{search\\time (s)} & \makecell{\#expanded\\nodes} & \makecell{\#edge\\ groups} & 
        \makecell{search\\time (s)} & \makecell{\#expanded\\nodes} & \makecell{\#edge\\ groups} & 
        \makecell{search\\time (s)} & \makecell{\#expanded\\nodes} & \makecell{\#edge\\ groups}\\
        \midrule
    
1 &GSES & 1.421 & 3051.5 & 1637.9 & 3.805 & 948.3 & 15441.6 & 2.608 & 550.8 & 30320.6 & 5.221 & 253.1 & 40214.5\\
2 &+SG & 1.082 & 2025.7 & 1176.4 & 1.783 & 315.3 & 9134.2 & 1.242 & 185.7 & 19012.7 & 3.131 & 115.2 & 27065.5\\
3 &+FG & 0.945 & 1597.9 & 732.5 & 1.519 & 241.4 & 2367.0 & 1.021 & 130.8 & 5601.8 & 2.770 & 88.1 & 10706.9\\
4 &+FG +RB & 1.555 & 2495.2 & 732.5 & 2.685 & 447.4 & 2367.0 & 2.644 & 324.2 & 5601.8 & 3.553 & 114.5 & 10706.9\\
5 &+FG +EB & 1.185 & 2036.1 & 732.5 & 2.354 & 381.6 & 2367.0 & 1.481 & 172.9 & 5601.8 & 3.118 & 102.6 & 10706.9\\
6 &+FG +SB & 0.915 & 1480.5 & 732.5 & 1.472 & 229.5 & 2367.0 & 1.027 & 129.0 & 5601.8 & 2.472 & 78.7 & 10706.9\\
7 &+FG +SB +SH & 0.169 & 84.6 & 732.5 & 0.832 & 34.6 & 2367.0 & 0.871 & 33.1 & 5601.8 & 1.561 & 15.4 & 10706.9\\
8 &+FG +SB +SH +INC & 0.043 & 84.6 & 732.5 & 0.120 & 34.6 & 2367.0 & 0.221 & 33.1 & 5601.8 & 0.321 & 15.4 & 10706.9\\

        \bottomrule
        \bottomrule
    \end{tabular}
    }
\end{table*}

\begin{table*}[tbh]
    \caption{Incremental analysis on instances solved by all the settings with delay probability $p=0.03$. Rows 1,2,3 compare the effects of different grouping methods. GSES has no grouping. SG: simple grouping. FG: full grouping. Rows 3,4,5,6 compare the effects of different branching orders. Row 3 applies the default Agent-First branching in GSES. RB: Random branching. EB: Earliest-First branching. SB: Smallest-Edge-Slack-First branching. Rows 6,7 compare the effect of the stronger heuristics. SH: stronger heuristics. Rows 7,8 compare the effect of the incremental implementation. INC: incremental implementation. The last row is also our IGSES.}
    \label{tab:incremental_p03}
    \centering
    \resizebox{0.98\textwidth}{!}{
    \begin{tabular}{c|c|rrr|rrr|rrr|rrr}
        \toprule
        \toprule
        \multirow{3}{*}{\makecell{Row}} & \multirow{3}{*}{Setting} & \multicolumn{3}{c|}{Random-32-32-10} &
        \multicolumn{3}{c|}{Warehouse-10-20-10-2-1} & 
        \multicolumn{3}{c|}{Lak303d} & 
        \multicolumn{3}{c}{Paris\_1\_256}\\
        & & \makecell{search\\time (s)} & \makecell{\#expanded\\nodes} & \makecell{\#edge\\ groups} & \makecell{search\\time (s)} & \makecell{\#expanded\\nodes} & \makecell{\#edge\\ groups} & 
        \makecell{search\\time (s)} & \makecell{\#expanded\\nodes} & \makecell{\#edge\\ groups} & 
        \makecell{search\\time (s)} & \makecell{\#expanded\\nodes} & \makecell{\#edge\\ groups}\\
        \midrule
    
1 &GSES & 1.580 & 3699.3 & 1564.2 & 4.210 & 1080.9 & 14763.0 & 2.262 & 443.3 & 29775.0 & 5.457 & 273.6 & 39522.9\\
2 &+SG & 1.173 & 2383.1 & 1124.2 & 1.959 & 364.7 & 8744.3 & 1.115 & 153.5 & 18628.9 & 3.528 & 133.5 & 26525.9\\
3 &+FG & 0.996 & 1824.5 & 700.0 & 1.685 & 277.8 & 2283.7 & 0.907 & 105.7 & 5492.3 & 3.011 & 97.1 & 10396.3\\
4 &+FG +RB & 1.605 & 2970.8 & 700.0 & 2.968 & 496.8 & 2283.7 & 2.397 & 277.5 & 5492.3 & 4.120 & 133.8 & 10396.3\\
5 &+FG +EB & 1.716 & 3276.8 & 700.0 & 2.534 & 435.8 & 2283.7 & 1.426 & 170.6 & 5492.3 & 3.067 & 100.2 & 10396.3\\
6 &+FG +SB & 0.736 & 1317.0 & 700.0 & 1.564 & 253.6 & 2283.7 & 0.792 & 84.2 & 5492.3 & 2.680 & 86.2 & 10396.3\\
7 &+FG +SB +SH & 0.187 & 96.8 & 700.0 & 0.924 & 40.6 & 2283.7 & 0.752 & 26.3 & 5492.3 & 2.080 & 21.9 & 10396.3\\
8 &+FG +SB +SH +INC & 0.045 & 96.8 & 700.0 & 0.124 & 40.6 & 2283.7 & 0.205 & 26.3 & 5492.3 & 0.337 & 21.9 & 10396.3\\

        \bottomrule
        \bottomrule
    \end{tabular}
    }
\end{table*}

\begin{figure}[!tbh]
    \begin{subfigure}[b]{0.225\textwidth}
      \centering
      \includegraphics[width=1\textwidth]{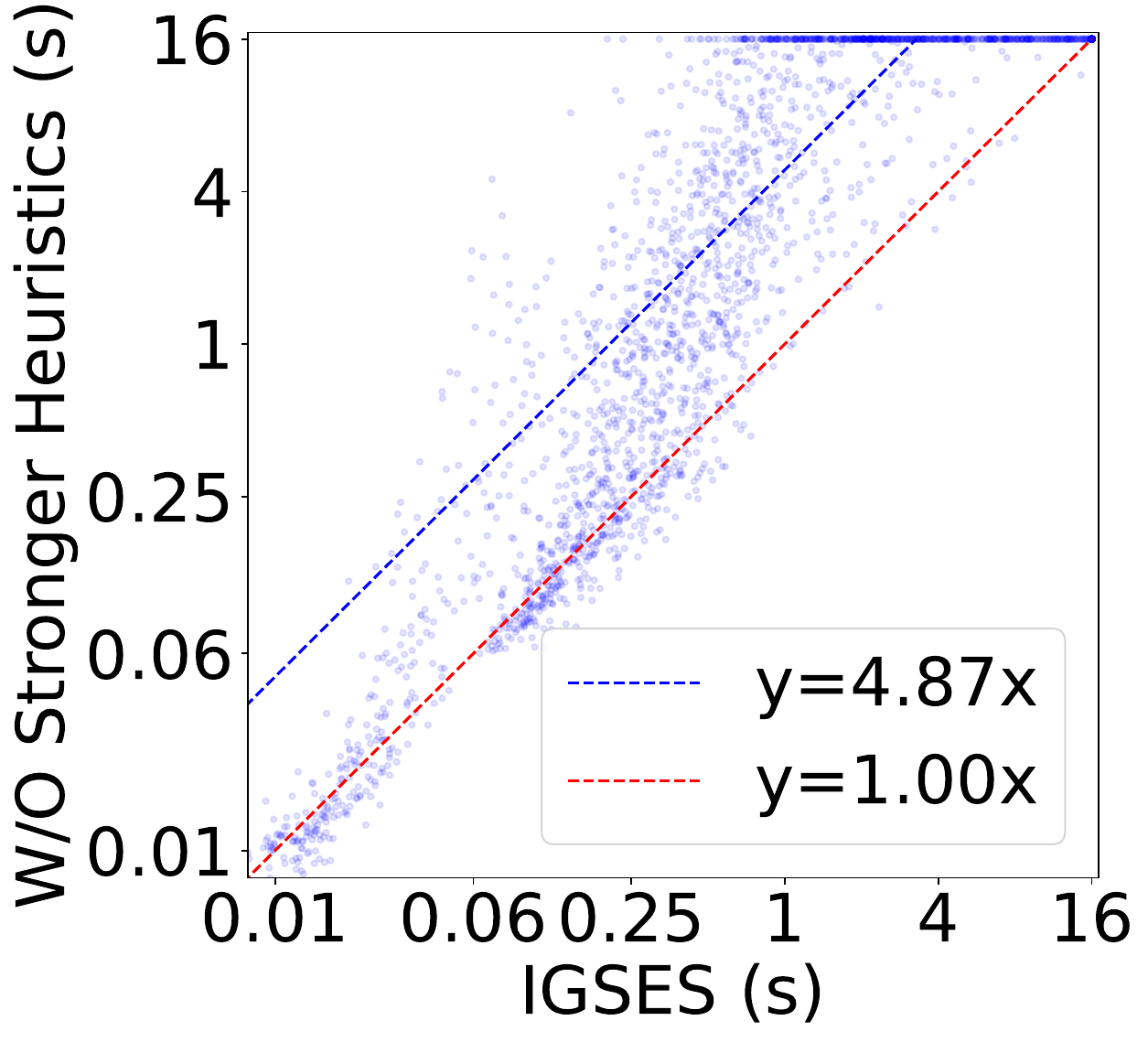}
      \caption{Heuristics}
    \end{subfigure}%
    \hfill
    \begin{subfigure}[b]{0.225\textwidth}
      \centering
      \includegraphics[width=1\textwidth]{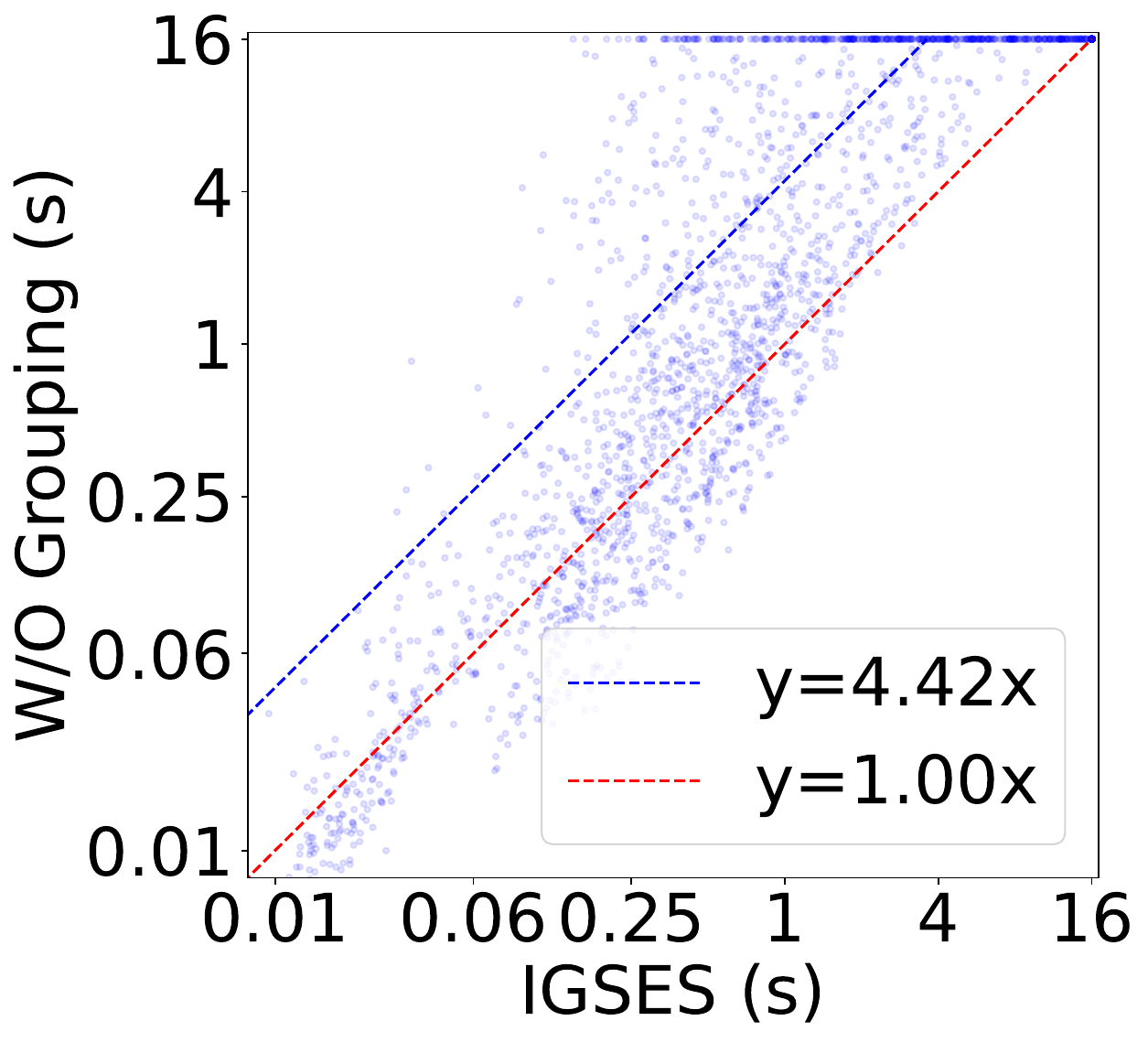}
      \caption{Grouping}
    \end{subfigure}%
    \hfill    
    
    \begin{subfigure}[b]{0.225\textwidth}
      \centering
      \includegraphics[width=1\textwidth]{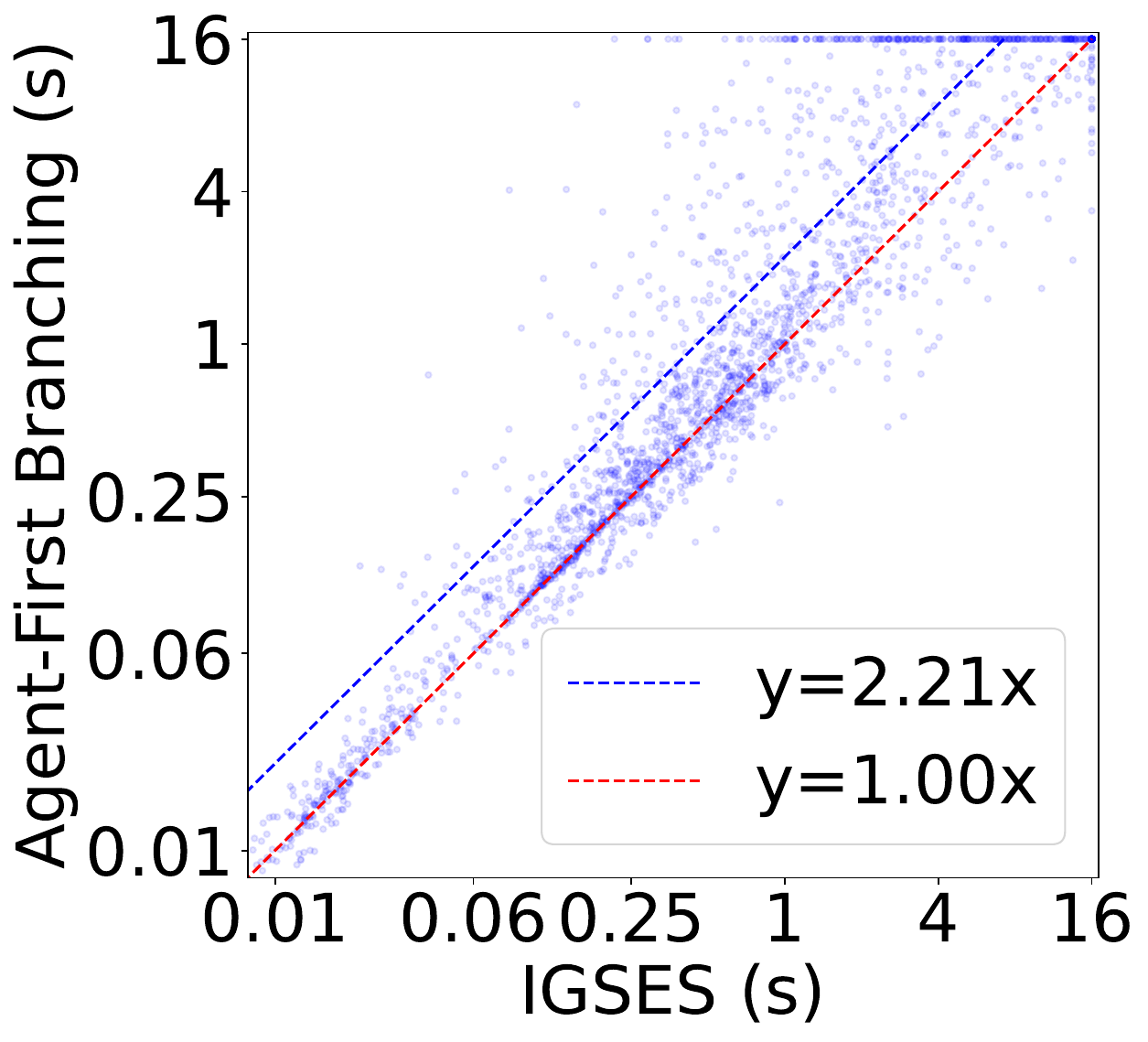}
      \caption{Branching}
    \end{subfigure}%
    \hfill
    \begin{subfigure}[b]{0.225\textwidth}
      \centering
      \includegraphics[width=1\textwidth]{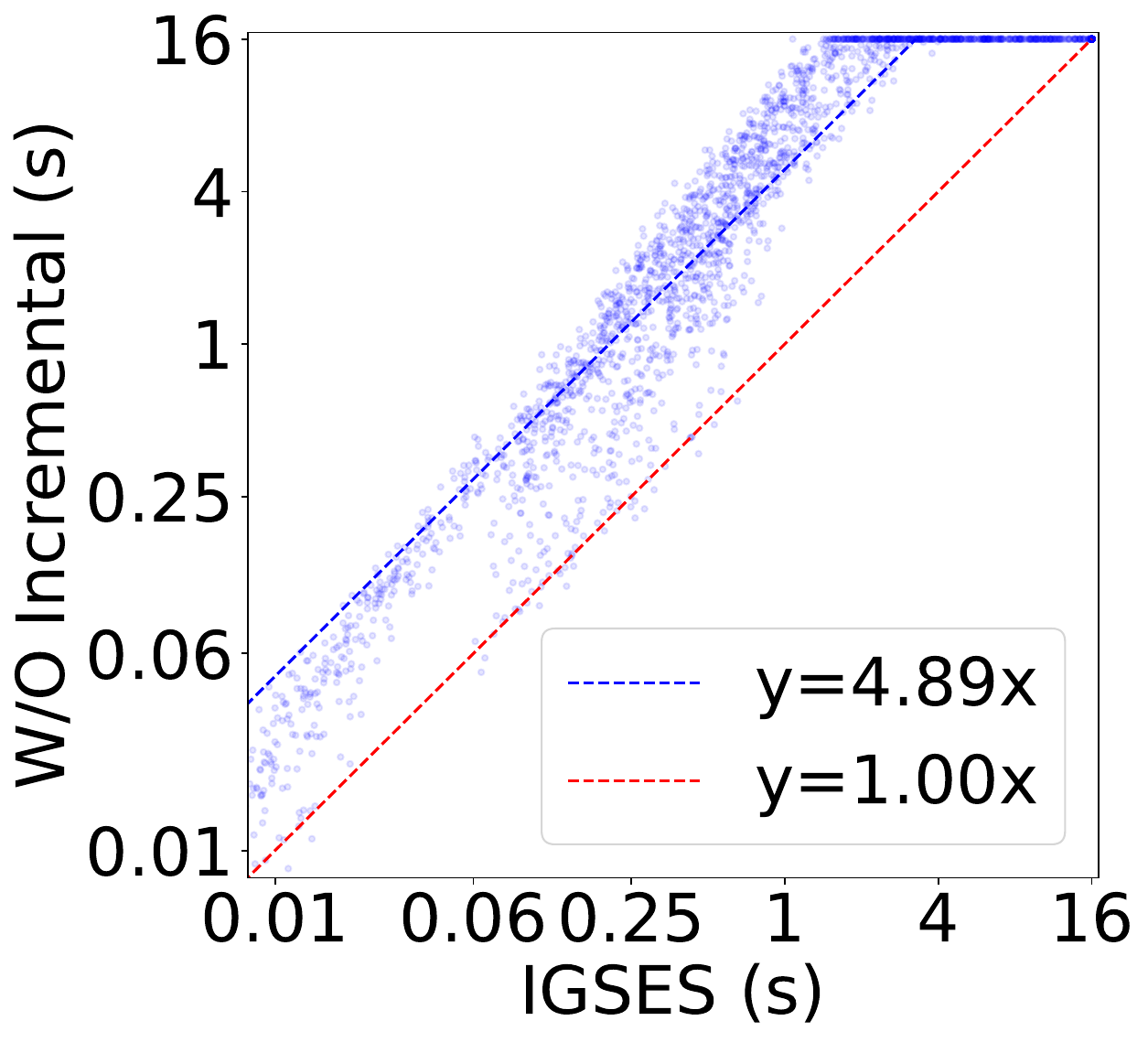}
      \caption{Incremental}
    \end{subfigure}%
    \hfill
    \caption{Ablation on four speedup techniques on all instances with delay probability $p=0.002$. In each figure, we compare IGSES to the setting that replaces one of its techniques by the GSES's choice. Each point in the graph represents an instance. The search time of an unsolved instance is set to 16 seconds.
    The blue line is the fitted on instances solved by at least one setting. Its sloped indicates the average speedup.}
    \label{fig: ablation_p002}

    \begin{subfigure}[b]{0.225\textwidth}
      \centering
      \includegraphics[width=1\textwidth]{figures/ablation/compare_speed_heuristics_p01.pdf}
      \caption{Heuristics}
    \end{subfigure}%
    \hfill
    \begin{subfigure}[b]{0.225\textwidth}
      \centering
      \includegraphics[width=1\textwidth]{figures/ablation/compare_speed_grouping_p01.pdf}
      \caption{Grouping}
    \end{subfigure}%
    \hfill    
    
    \begin{subfigure}[b]{0.225\textwidth}
      \centering
      \includegraphics[width=1\textwidth]{figures/ablation/compare_speed_branching_p01.pdf}
      \caption{Branching}
    \end{subfigure}%
    \hfill
    \begin{subfigure}[b]{0.225\textwidth}
      \centering
      \includegraphics[width=1\textwidth]{figures/ablation/compare_speed_incremental_p01.pdf}
      \caption{Incremental}
    \end{subfigure}%
    \hfill
    \caption{Ablation on four speedup techniques on all instances with delay probability $p=0.01$. In each figure, we compare IGSES to the setting that replaces one of its techniques by the GSES's choice. Each point in the graph represents an instance. The search time of an unsolved instance is set to 16 seconds.
    The blue line is the fitted on instances solved by at least one setting. Its sloped indicates the average speedup.}
    \label{fig: ablation_p01}

\end{figure}

\begin{figure}[!tbh]
    \begin{subfigure}[b]{0.225\textwidth}
      \centering
      \includegraphics[width=1\textwidth]{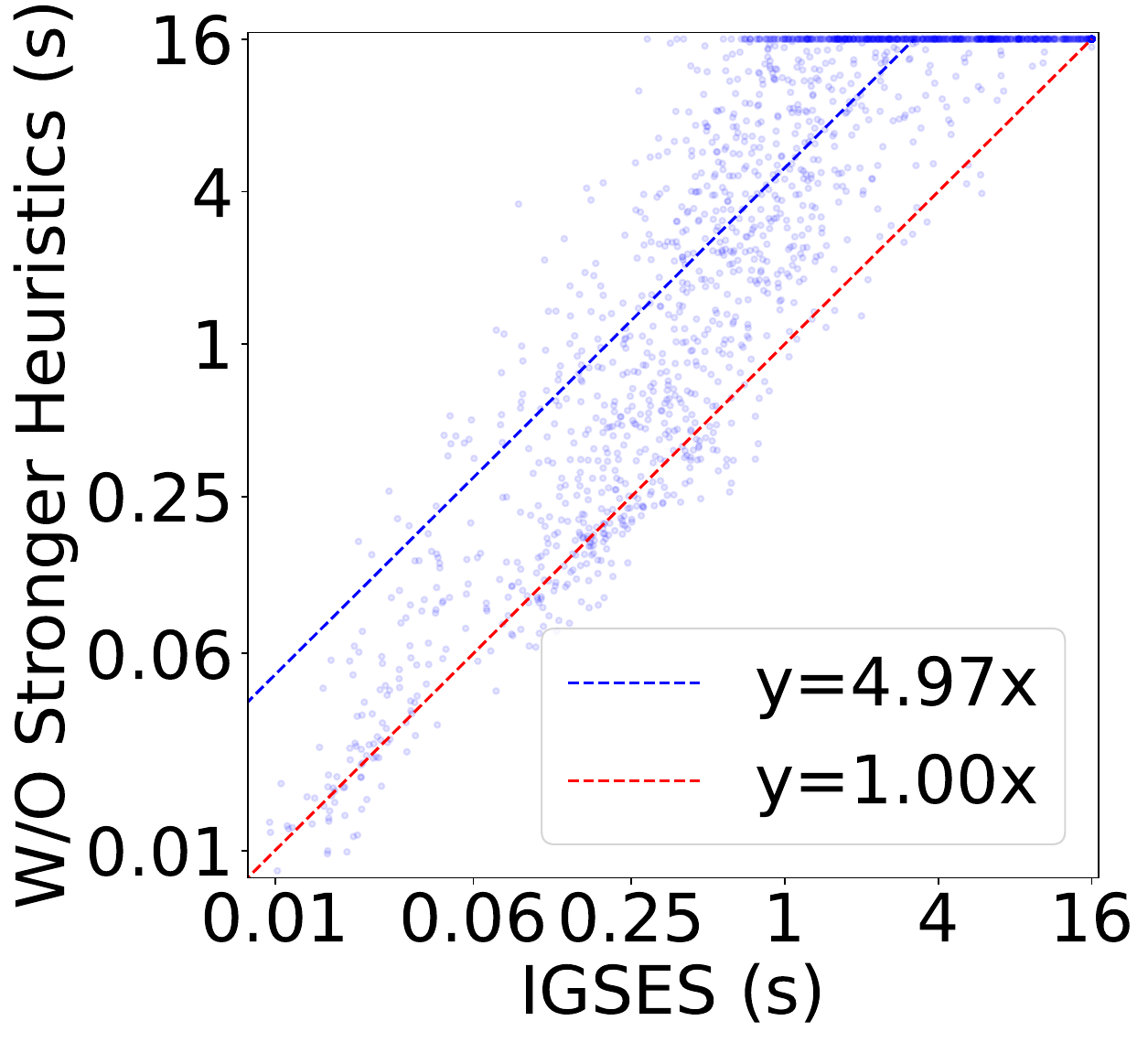}
      \caption{Heuristics}
    \end{subfigure}%
    \hfill
    \begin{subfigure}[b]{0.225\textwidth}
      \centering
      \includegraphics[width=1\textwidth]{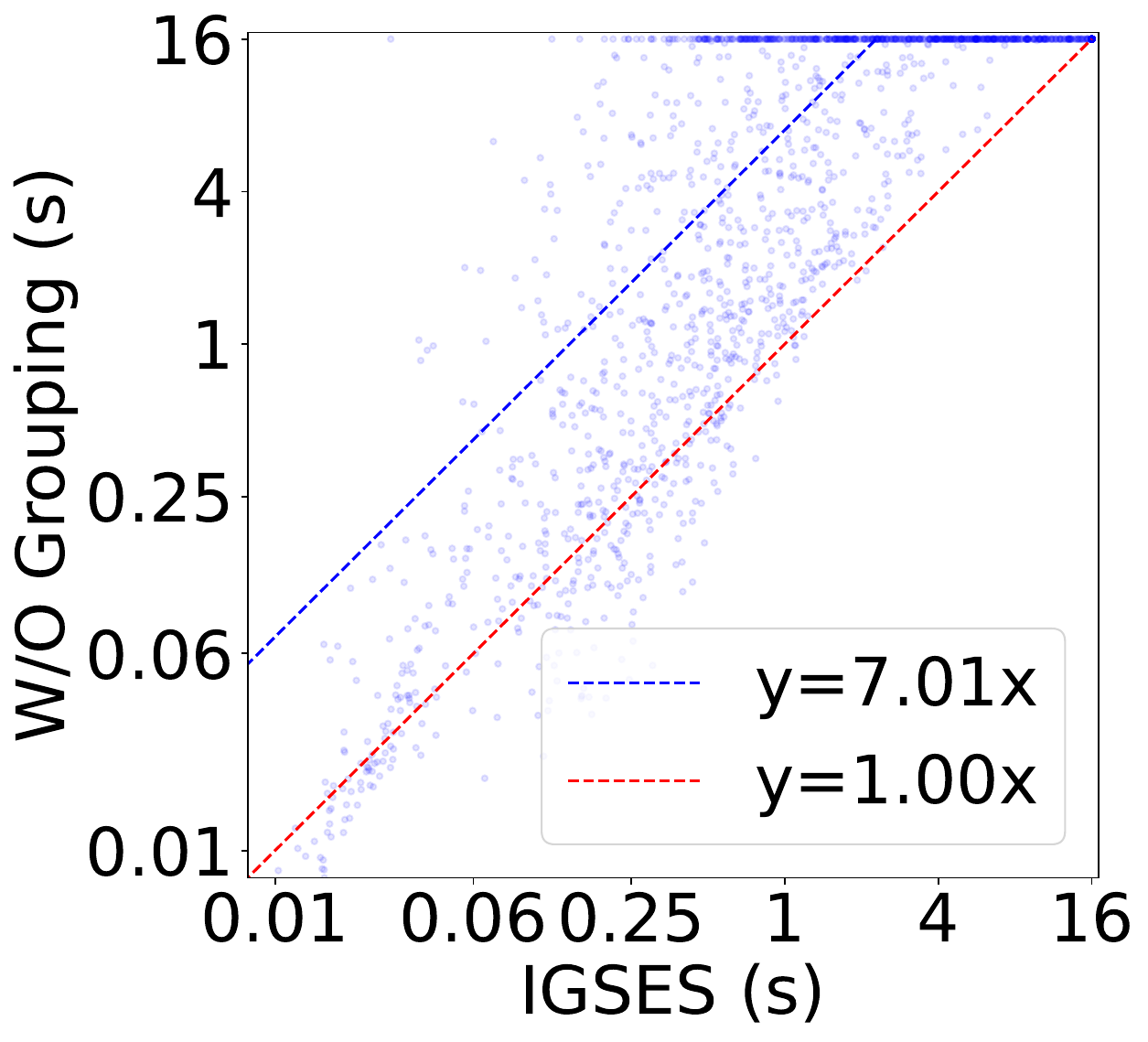}
      \caption{Grouping}
    \end{subfigure}%
    \hfill    
    
    \begin{subfigure}[b]{0.225\textwidth}
      \centering
      \includegraphics[width=1\textwidth]{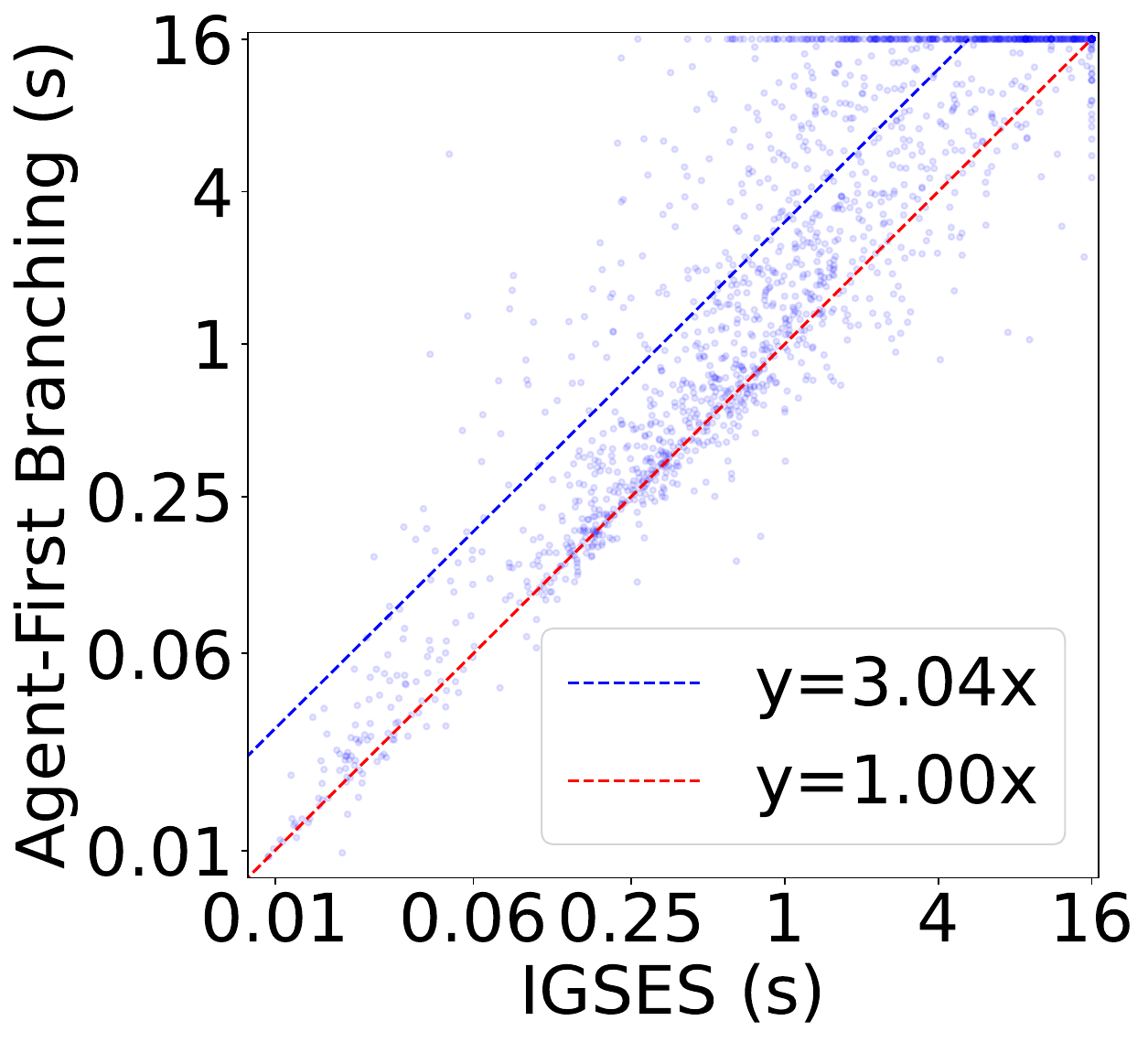}
      \caption{Branching}
    \end{subfigure}%
    \hfill
    \begin{subfigure}[b]{0.225\textwidth}
      \centering
      \includegraphics[width=1\textwidth]{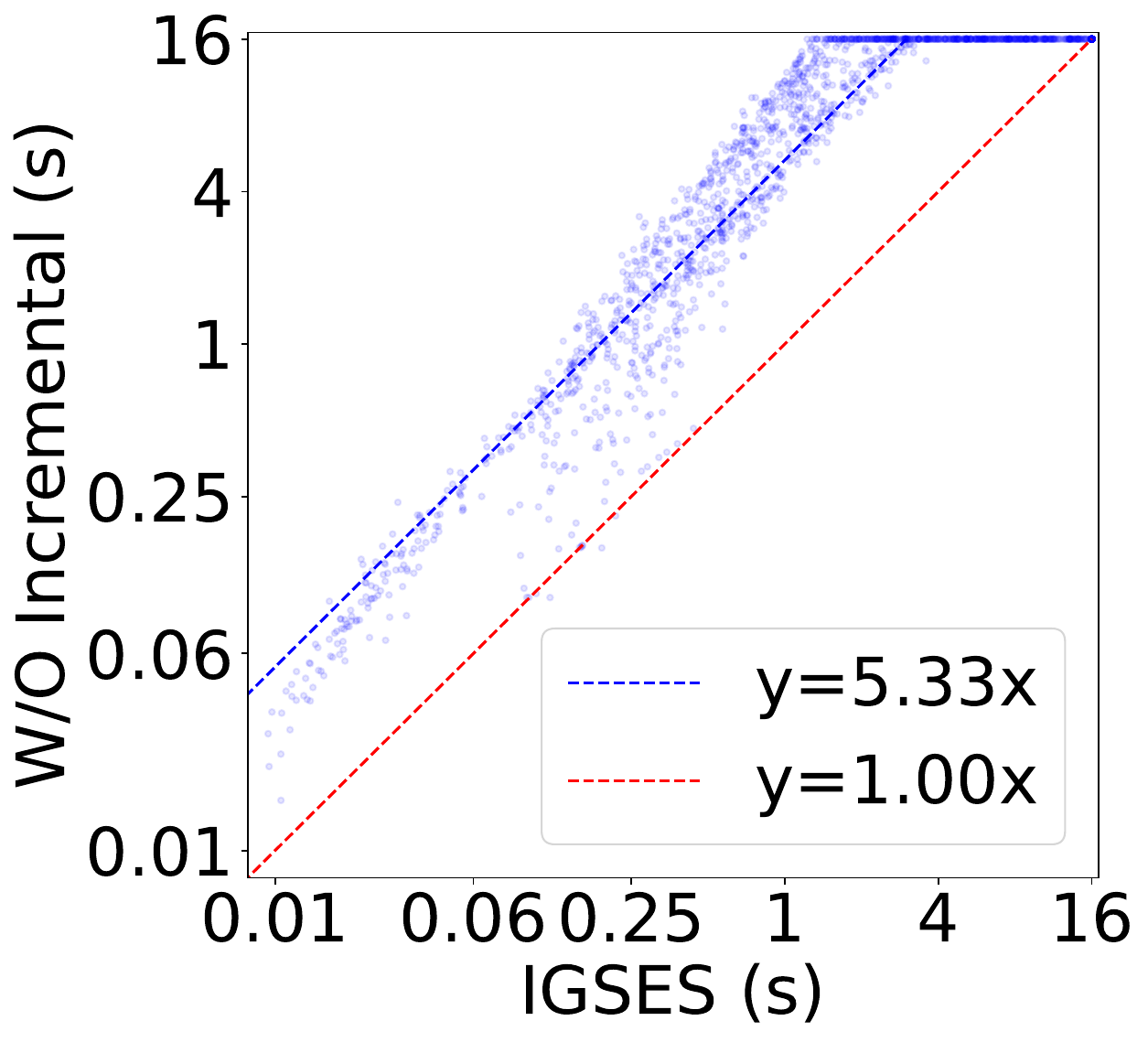}
      \caption{Incremental}
    \end{subfigure}%
    \hfill
    \caption{Ablation on four speedup techniques on all instances with delay probability $p=0.03$. In each figure, we compare IGSES to the setting that replaces one of its techniques by the GSES's choice. Each point in the graph represents an instance. The search time of an unsolved instance is set to 16 seconds.
    The blue line is the fitted on instances solved by at least one setting. Its sloped indicates the average speedup.}
    \label{fig: ablation_p03}

\end{figure}

\end{document}